\documentclass[letterpaper,twocolumn,10pt]{article}
\usepackage{usenix-2020-09}
\usepackage{graphicx} %
\usepackage{float}
\usepackage{caption}
\usepackage{subcaption}
\usepackage[shortlabels]{enumitem}
\usepackage{nicefrac}
\usepackage{mathtools}
\usepackage{amsmath}
\usepackage{amssymb}
\usepackage{amsthm}
\usepackage{comment}
\usepackage{worldflags}
\usepackage{multirow}
\usepackage{xpatch}
\usepackage{xspace}
\usepackage{algorithm}
\usepackage[square,numbers]{natbib}
\usepackage{pifont}
\usepackage[noend]{algpseudocode}
\usepackage[hang,flushmargin]{footmisc}
\usepackage{balance}

\RequirePackage[T1]{fontenc} 
\RequirePackage[varqu]{zi4} 

\newtheorem{thm}{Theorem}[section]

\newtheorem{cor}[thm]{Corollary}
\newtheorem{dfn}[thm]{Definition}

\setlist[itemize,1]{leftmargin=4mm}
\setlist[enumerate,1]{leftmargin=6mm}

\newcommand{\DANISH}{\texttt{PCAPS}\xspace}
\newcommand{\PCAPS}{\texttt{PCAPS}\xspace}
\newcommand{\OPT}{\texttt{OPT}\xspace}
\newcommand{\ALG}{\texttt{ALG}\xspace}
\newcommand{\CAP}{\texttt{CAP}\xspace}
\newcommand{\E}{\mathbb{E}\xspace}

\newcommand\blfootnote[1]{
    \begingroup
    \renewcommand\thefootnote{}\footnote{#1}
    \addtocounter{footnote}{-1}
    \endgroup
}

\hypersetup{
    linkcolor = {violet},
    citecolor = {teal},
    urlcolor = {teal}
}

\newcommand{\sref}[2]{\hyperref[#2]{#1 \ref{#2}}}
\newcommand{\mbf}[1]{{\mathbf{#1}}}

\usepackage{etoolbox}
\makeatletter
\patchcmd{\hyper@makecurrent}{%
    \ifx\Hy@param\Hy@chapterstring
        \let\Hy@param\Hy@chapapp
    \fi
}{%
    \iftoggle{inappendix}{%
        \@checkappendixparam{chapter}%
        \@checkappendixparam{section}%
        \@checkappendixparam{subsection}%
        \@checkappendixparam{subsubsection}%
        \@checkappendixparam{paragraph}%
        \@checkappendixparam{subparagraph}%
    }{}%
}{}{\errmessage{failed to patch}}

\newcommand*{\@checkappendixparam}[1]{%
    \def\@checkappendixparamtmp{#1}%
    \ifx\Hy@param\@checkappendixparamtmp
        \let\Hy@param\Hy@appendixstring
    \fi
}
\makeatletter

\newtoggle{inappendix}
\togglefalse{inappendix}

\apptocmd{\appendix}{\toggletrue{inappendix}}{}{\errmessage{failed to patch}}

\begin{document}

\date{} 

\title{
\vspace{-0.5cm}
Carbon- and Precedence-Aware Scheduling for Data Processing Clusters \vspace{-0.5cm}
}

\author{
{\rm Adam Lechowicz}\\
University of Massachusetts Amherst
\and
{\rm Rohan Shenoy}\\
University of California Berkeley
\and
{\rm Noman Bashir}\\
Massachusetts Institute of Technology
\and
{\rm Mohammad Hajiesmaili}\\
University of Massachusetts Amherst
\and
{\rm Adam Wierman}\\
California Institute of Technology
\and
{\rm Christina Delimitrou}\\
Massachusetts Institute of Technology
} %

\maketitle

\begin{abstract}

As large-scale data processing workloads continue to grow, their carbon footprint raises concerns. Prior research on carbon-aware schedulers has focused on shifting computation to align with availability of low-carbon energy, but these approaches assume that each task can be executed independently. 
In contrast, data processing jobs have precedence constraints (i.e., outputs of one task are inputs for another) that complicate decisions, since delaying an upstream ``bottleneck'' task to a low-carbon period will also block downstream tasks, impacting the entire job's completion time. 
In this paper, we show that carbon-aware scheduling for data processing benefits from knowledge of both time-varying carbon and precedence constraints. 
Our main contribution is $\texttt{PCAPS}$, a carbon-aware scheduler that interfaces with modern ML scheduling policies to explicitly consider the precedence-driven importance of  each task in addition to carbon.  
To illustrate the gains due to fine-grained task information, we also study $\texttt{CAP}$, a wrapper for any carbon-agnostic scheduler that adapts the key provisioning ideas of $\texttt{PCAPS}$.
Our schedulers enable a configurable priority between carbon reduction and job completion time, and we give analytical results characterizing the trade-off between the two.
Furthermore, our Spark prototype on a 100-node Kubernetes cluster shows that a moderate configuration of $\texttt{PCAPS}$ reduces carbon footprint by up to 32.9\% without significantly impacting the cluster's total efficiency.

\end{abstract}

\section{Introduction} \label{sec:intro}
Concerns about the climate impact of machine learning (ML) and artificial intelligence (AI) have primarily revolved around the carbon footprint during the training phase~\cite{Hanafy:23:CarbonScaler, Wiesner:2021:WaitAwhile} or, in some cases, the inference phase~\cite{Baolin:2023:Clover} of their life cycle. 
However, as the data requirements of foundation models have ballooned, the \textit{data processing} tasks that must be completed before training %
account for almost one-third of the cumulative computation for an AI model during its life cycle~\cite{WU:2022:SustainableAI}. 
Furthermore, foundation model \textit{finetuning} generally trains a model on a narrower data set that may require additional data processing~\cite{Mosbach:21:FinetuningBERT} -- as the finetuning of general purpose models (e.g., Llama) for specific tasks~\cite{Liu:24:TaskLLMFinetuning, Lin:24:DataEfficientFinetuning} has gained traction, the comparative fraction of computational demand borne by data processing tasks is expected to grow.

Therefore, efforts towards responsible and sustainable development in AI must consider and optimize the carbon footprint of data processing. 
Even beyond sustainability, companies such as Microsoft have implemented \textit{internal carbon pricing} for short- and long-term decisions~\cite{Microsoft:19, McKinsey:21} that put financial responsibility on business divisions for each metric ton of operational CO$_2$ that they emit.  In the data center context, most current schedulers do not consider the time-varying aspect of carbon intensity and the resulting compute-carbon impact -- this must change to accommodate additional operational concerns such as carbon pricing.

Data processing frameworks (e.g., Apache Spark) ingest workloads that are composed of \textit{precedence-constrained tasks}, where e.g., the outputs of one operation are the inputs to another~\cite{Zaharia:12}.  
There is a rich literature studying scheduling algorithms for this case of precedence-constrained tasks (e.g., represented as a directed acyclic graph (DAG)) that characterize large-scale data processing.  From the theoretical side, optimal scheduling of precedence-constrained tasks (in terms of total completion time) is known to be NP-hard~\cite{Lenstra:78}.  Although there has been progress in approximation techniques~\cite{Su:24:Tompecs, Chudak:99, Lassota:23, Li:17, Davies:20, Davies:21, Maiti:20, Su:23}, the hardness of the problem necessitates simple settings with relatively strong assumptions. From an experimental perspective, there have been several studies proposing data-driven and/or evolutionary approaches for scheduling, both in the general precedence-constrained tasks case and the specific data processing case~\cite{Hongzi:2019:Decima, Wu:18, Li:23, Grinsztajn:20, Zhou:22, Cheng:96:Genetic, Pezzella:08:Genetic, Davis:14:Genetic, Islam:21}.  In recent years, such works have leveraged learning techniques such as graph neural networks (GNNs) and reinforcement learning (RL) to learn an improved scheduling policy, showing significant improvements in experiments.  However, owing to the complexities of these approaches, theoretical guarantees for learning-based approaches have proven difficult to obtain.

Beyond the singular objective of job completion time, a select few works have considered settings that are closer to the carbon-aware problem we study in this paper~\cite{Su:24:Tompecs, Goiri:2012:GreenHadoop, Liu:23}.  These \textit{multi-objective} scheduling environments balance the objectives of e.g., reducing job completion time alongside another metric of interest, such as cost.  For instance, several works have considered energy efficiency in tandem with job completion time, from both theoretical and experimental perspectives.  \citet{Su:24:Tompecs} study \textit{energy-aware} list scheduling for precedence constrained tasks, %
giving theoretical bounds for an combined objective of energy consumption and performance.
GreenHadoop~\cite{Goiri:2012:GreenHadoop} is a MapReduce framework for data centers with local renewable sources that predicts the future availability of carbon-free (``green'') electricity and schedules jobs accordingly, subject to deadlines for individual jobs.  
\citet{Liu:23} consider job scheduling for low-carbon data center operation in a general model with both DAG and non-DAG jobs -- they develop an RL-based scheduler that focuses primarily on increasing energy-efficiency.

Despite these previous works, focusing on \textit{carbon-efficiency} rather than energy-efficiency requires different techniques.  In particular, while carbon-efficiency and energy-efficiency are sometimes complementary objectives, they are often contradictory~\cite{Hanafy:23} -- for instance, due to the \textit{time-varying} nature of carbon intensity, it may be advantageous to scale up during low-carbon periods (i.e., sacrificing energy efficiency) in exchange for the ability to scale down during high-carbon periods.
Works that \textit{do} consider carbon emissions (e.g., GreenHadoop) use abstractions, such as job-level deadlines and ``green'' vs. ``brown'' energy, that do not adequately model the current m.o. in data centers.

To address this multi-objective 
setting while catering to realistic scenarios, we propose that a middle-ground approach is needed -- namely, by drawing on techniques from the theoretical literature for precedence-constrained and carbon-aware scheduling, and simultaneously considering experimental advances, we seek a simple and interpretable framework that comes with guarantees in terms of the \textit{trade-off} between job completion time and carbon savings.

In this paper, we propose \PCAPS
(\textbf{P}recedence- and \textbf{C}arbon-\textbf{A}ware \textbf{P}rovisioning and \textbf{S}cheduling), a carbon-aware scheduler for data processing clusters.  \PCAPS is theoretically-inspired, leveraging a paradigm of interpretable and configurable \textit{threshold-based design} that informs decisions at each time step based on e.g., the current carbon intensity and/or carbon price.  In keeping with this inspiration, we give analytical results that characterize the trade-off between job completion time and carbon savings.
\PCAPS is also practically relevant, drawing on recent insights from ML-based DAG schedulers (e.g., Decima~\cite{Hongzi:2019:Decima}, LACHESIS~\cite{Zhou:22}, and others~\cite{Wu:18, Li:23, Grinsztajn:20}).  By interfacing with a score or probability distribution over available tasks, \PCAPS defines a notion of \textit{relative importance} (i.e., compared to other tasks) -- this allows it to make fine-grained carbon-aware decisions that take the DAG's structure into account, such as continuing to schedule bottleneck tasks even if carbon intensity is high.  See \autoref{sec:danish-design} for a formal description of \PCAPS's design.

As a simplification of \PCAPS, we additionally propose and study \CAP (\textbf{C}arbon-\textbf{A}ware \textbf{P}rovisioning), which takes the provisioning ideas of \PCAPS and generalizes them to interoperate with any carbon-agnostic scheduler.  
Without explicitly considering inter-task dependencies, \CAP changes the resources available to the cluster, capturing an intuition that clusters should \textit{throttle down} during high-carbon periods and vice versa~\cite{Hanafy:23:CarbonScaler, radovanovic2022carbon} -- see \autoref{sec:cap-design} for a description.

\begin{figure*}[t]
    \includegraphics[width=\linewidth]{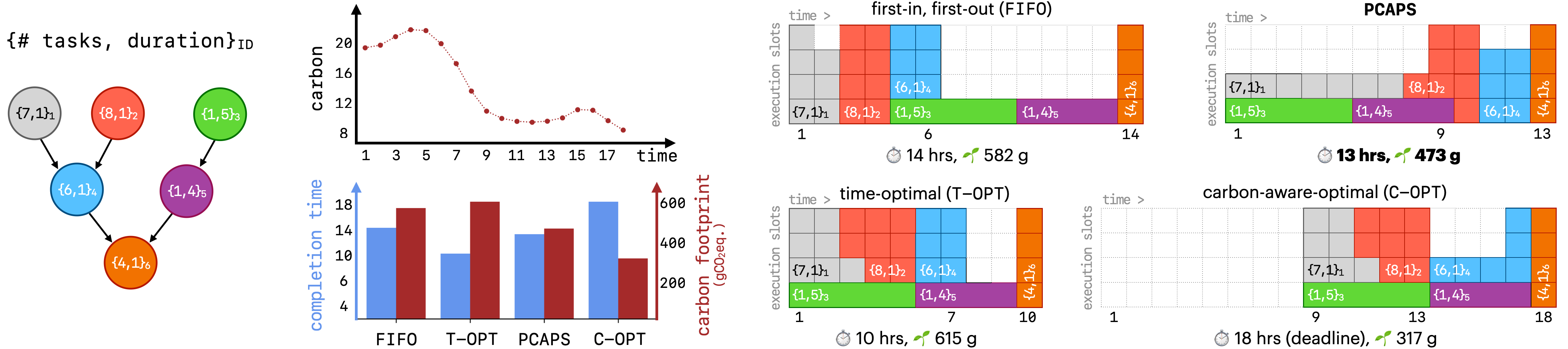}
    \vspace{-1em}
    \caption{ Four scheduling policies for a motivating DAG and 18-hour-long carbon intensity trace (on the left hand side).  
    Compared to a carbon-agnostic FIFO scheduler, the time-optimal approach (\texttt{T-OPT}) prioritizes starting the green and purple stages early to reduce completion time.  A carbon-aware-optimal approach (\texttt{C-OPT}) with a \textit{deadline} to finish the DAG within 18 hours reduces carbon emissions by 51.2\%, at the expense of increasing time by 28.5\% compared to FIFO.  
    By prioritizing green and purple stages during high-carbon periods, \PCAPS reduces carbon emissions by 23.1\% and still completes the job 7\% earlier compared to FIFO.} 
    \label{fig:motivation} \vspace{-1em}
\end{figure*}

We have implemented \PCAPS and \CAP as modules for Spark on Kubernetes and as extensions for a high-fidelity simulator~\cite{Hongzi:2019:Decima}.  Our experiments consider real and synthetic data processing workloads from Alibaba traces and TPC-H~\cite{TPCH:18, Alibaba:18}, alongside real carbon intensity traces from six power grids~\cite{electricity-map}.  
In our prototype implementation, we evaluate \PCAPS and \CAP on a 100-node Spark cluster.
We report the impact of carbon-aware policies on both job completion time and carbon savings, showing that \PCAPS and \CAP's configurable design enables notable carbon reduction for mild increases in \textit{end-to-end completion time}, which measures the total time to complete all jobs in a given experiment, measuring the system's overall throughput and efficiency.
We summarize our key contributions as follows:

\begin{enumerate}[itemsep=0.1cm]
    \item \PCAPS, 
    a carbon-aware scheduler that interfaces with a probability distribution over stages of a DAG, such as those provided by ML schedulers.  \PCAPS incorporates carbon into decisions that arbitrage between stages of a job at a granular level, obtaining a favorable trade-off between carbon savings and job completion time.
    \item \CAP, %
    a carbon-awareness module that dynamically adjusts cluster resources without replacing an existing scheduler (see \autoref{sec:cap-design}).  Compared to \PCAPS, it obtains a worse trade-off between carbon and completion time in exchange for flexibility and ease of implementation.
    \item  We analyze the \textit{carbon stretch factor} for \PCAPS and \CAP, which bounds the increase in job completion time due to carbon-aware actions (e.g., see \sref{Theorems}{thm:ksMakespan} and \ref{thm:danishMakespan}).
    \item  We implement \PCAPS and \CAP as extensions for a high-fidelity Spark simulator, alongside proof-of-concept prototypes for Spark on Kubernetes (see \autoref{sec:impl}).  We evaluate our proposed carbon-aware schedulers against baselines and a state-of-the-art ML scheduler (see \autoref{sec:eval}). \blfootnote{Our experiment code is available at \url{https://github.com/umass-solar/carbon-aware-dag/}.}
\end{enumerate}

\section{Problem and Motivation} \label{sec:prob}
This section formalizes the carbon-aware scheduling problem and motivates insights to contextualize 
our desiderata.

\vspace{-0.5em}
\subsection{Carbon-aware DAG scheduling problem}

Each job %
is represented as a directed acyclic graph (DAG) $\mathcal{J} = \{ \mathcal{V}, \mathcal{E} \}$, where each node in $\mathcal{V}$ is one of $n$ tasks, and each edge in $\mathcal{E}$ encodes precedence constraints between tasks -- e.g., for tasks $j, j' \in \mathcal{V}$, an edge $j \to j'$ indicates that $j'$ cannot start until after $j$ has completed.
A typical data processing cluster includes $K \geq 1$ machines (or executors).  More than one job can simultaneously run on a cluster -- e.g., given a set of current jobs $\{ \mathcal{J} \}$, the scheduler assigns tasks to machines over time while respecting precedence and capacity constraints.  We index continuous time by $t \geq 0$.

The goal of a typical scheduler is \textit{performance}, e.g., in terms of throughput, utilization, and average job completion time.  In this work, we additionally consider the goal of \textit{carbon-awareness} -- with respect to a time varying carbon signal given by a function $c(t) : t \geq 1$, a carbon-aware scheduler's objective is to minimize a combination of typical metrics (i.e., job completion time) and the overall carbon footprint (on both a per-job and a global, cluster basis).

Although future values of this carbon signal are unknown to the scheduler, in the rest of the paper, we follow prior work~\cite{Lechowicz:23, Bostandoost:24} and assume that it is bounded by constants $L$ and $U$ that are known to the scheduler, where $L \leq c(t) \leq U$.  In practice, the values of $L$ and $U$ can capture e.g., short-term forecasts of the best and worst carbon conditions on a given electric grid over the next 24 or 48 hours.

\vspace{-0.5em}
\subsection{Prior work and motivation}\label{sec:motiv}

Scheduling directed acyclic graphs (DAGs), or more broadly, precedence-constrained tasks, has been extensively studied.\\ 
Classic results establish the difficulty of this problem: even in its simplest forms, DAG scheduling is NP-hard~\cite{Lenstra:78}. 
To address this, prior work has developed heuristic methods and approximation algorithms~\cite{Su:24:Tompecs, Chudak:99, Lassota:23, Li:17, Davies:20, Davies:21, Maiti:20, Su:23}, ranging from the well-known list scheduling algorithm~\cite{Graham:66}, priority-based algorithms~\cite{Sels:12:Priority}, to more complex approaches such as genetic programming~\cite{Cheng:96:Genetic, Pezzella:08:Genetic, Davis:14:Genetic}. These methods often rely on simplifying assumptions, such as fixed task durations or centralized knowledge of the task graph.  %

In recent years, DAG scheduling has become %
a key problem in \textit{data processing frameworks} such as Apache Airflow, Beam, and Spark, which use DAGs to represent workflows. 
In Spark, each node of a job's DAG corresponds to a \textit{stage}, which encapsulates operations (\textit{tasks}) that can be executed in parallel over partitions of input data. Inter-stage dependencies impose precedence constraints: a stage can only begin once all ``parent'' stages have completed. 
Frameworks such as Spark typically implement simple scheduling strategies such as first-in, first-out (FIFO) and fair-share scheduling~\cite{SparkScheduling} -- these are explainable and efficient in terms of overhead, but suboptimal in terms of job completion time.

Recent works that revisit scheduling %
for data processing have explored %
learning-based techniques, such as reinforcement learning (RL) methods that dynamically learn scheduling policies~\cite{Hongzi:2019:Decima, Wu:18, Li:23, Grinsztajn:20, Zhou:22, Islam:21}.  Although these methods outperform default policies and hand-tuned heuristics in terms of job completion time, theoretical results for these techniques have proven difficult to obtain.

Carbon awareness adds a new dimension to the DAG scheduling problem -- an online scheduler must consider the time-varying carbon intensity while choosing to assign resources to specific nodes in the job DAG(s), with an overarching goal of reducing carbon footprint, combined with traditional metrics such as job completion time -- see \autoref{fig:motivation} for an illustration of this desired behavior for \DANISH, FIFO, and optimal schedules.
As discussed above, the state-of-the-art for carbon-agnostic DAG scheduling falls into two categories: theoretical models that focus on provably near-optimal schedules under idealized assumptions, and heuristic or learning-based methods that do not provide theoretical bounds but perform well in practice.
In adding carbon-awareness to the problem, we consider a \textit{middle ground} that balances between design goals of simplicity, interpretability, configurability, and performance. 
In particular, we seek a carbon-aware scheduler that is tractable for theoretical insight, offering provable bounds on, e.g., the trade-off between carbon and job completion time while not sacrificing the efficiency gains that come from, e.g., learning DAG structure.

\section{Theoretical Foundations} 
\label{sec:theory}
This section details theoretical underpinnings and intuition for our design in \autoref{sec:design}.
Recent literature has studied carbon-aware scheduling problems with a theoretical lens~\cite{Lechowicz:23,Lechowicz:24,Lechowicz:24CFL,Bostandoost:24}, spanning relatively simple suspend-resume~\cite{Lechowicz:23} to settings considering scaling and uncertainty in job lengths~\cite{Bostandoost:24}.
In these online carbon-aware scheduling problems, the key challenge is the inherent uncertainty in future carbon intensity values due to the proliferation of intermittent renewable energy sources. 

A common approach to manage this uncertainty is \textit{threshold-based design}~\cite{Lechowicz:23, Bostandoost:24:HotCarbon, Wiesner:2021:WaitAwhile}, %
that uses a predetermined and parameterized threshold function to inform decisions. 
Among studies that use this design paradigm, a common theoretical performance metric is \textit{competitive ratio}, %
which is the worst-case ratio ($\geq 1$) between the cost of an online algorithm vs. that of an optimal solution.  Algorithms designed using this metric are known to be pessimistic in practice~\cite{Lykouris:18, Purohit:18}.  Moreover, existing theoretical studies on carbon-aware scheduling focus on simple settings where, e.g., the job is bound by a deadline, the objective is only to reduce carbon, and precedence constraints are not considered.
However, threshold-based algorithms have been demonstrated to work well in practice: they are often close to optimal provided their inputs are reasonably accurate~\cite{Daneshvaramoli:24}.

Carbon-aware DAG scheduling exhibits an inherent trade-off between carbon savings and job completion time (JCT).  
Although worst-case metrics %
(i.e., bounds with respect to an intractable offline solution) have limited utility in this setting, it is still useful to quantify a trade-off between carbon and JCT -- to this end, we introduce two metrics that we use in the following sections.
We start by introducing some notation:  let $\OPT_K(\mathcal{J})$ denote the optimal makespan for job $\mathcal{J}$, and let $\ALG_K(\mathcal{J})$ denote the makespan for the schedule generated by some scheduler $\ALG$ (all using a maximum of $K$ machines).

\vspace{-0.5em}
\begin{dfn}[Carbon Stretch Factor (CSF)] \label{dfn:csf}
Given a scheduling policy (e.g., FIFO), let $\texttt{AG}$ denote the regular (i.e., carbon-agnostic) scheduling policy, and let $\texttt{CA}$ denote a carbon-aware variant of the same scheduling policy.
If $a$ is an upper bound such that $\texttt{AG}_K(\mathcal{J}) \leq a \cdot \OPT_K(\mathcal{J}) : \forall \mathcal{J}$, and $b$ is an upper bound such that $\texttt{CA}_K(\mathcal{J}) \leq b \cdot \OPT_K(\mathcal{J}) : \forall \mathcal{J}$, where $b \geq a$, then the \textbf{carbon stretch factor} is defined as $\nicefrac{b}{a}$, which indicates (multiplicatively) how much worse the makespan of $\texttt{CA}$ is compared to $\texttt{AG}$. Note that $\nicefrac{b}{a} \geq 1$. 
\end{dfn}
\vspace{-0.5em}

\noindent To quantify carbon savings, we define $C_\ALG(t)$ as the instantaneous carbon emissions at time $t$ due to decisions by scheduler $\ALG$.  It is a function of the number of executors active in $\ALG$'s schedule at time $t$ (denoted by $E_\ALG(t)$) and the current carbon intensity: $C_\ALG(t) \coloneqq c(t) E_\ALG (t)$.

\begin{dfn}[Carbon Savings] \label{dfn:carbonsavings}
Let $\texttt{AG}$ and $\texttt{CA}$ denote a carbon-agnostic and carbon-aware scheduler as outlined in \sref{Def}{dfn:csf}. %
For a job $\mathcal{J}$, if $\texttt{AG}$ runs from time step $0$ until $T$ (its completion time), and $\texttt{CA}$ operates from time $0$ to $T'$, %
then $\texttt{CA}$'s \textbf{carbon savings} are given by $\int_0^T C_{\texttt{AG}}(t) - \int_0^{T'} C_{\texttt{CA}}(t)$.
\end{dfn}

\noindent Using CSF and carbon savings, we describe the desired behavior of a carbon-aware scheduler for data processing.
A basic intuition in threshold-based designs is ``hedging'' between completing tasks now and waiting for lower-carbon periods that may arrive. 
To do this, thresholds rely on the \textit{range} of carbon intensities that are expected to appear in the near future (i.e., $L$ and $U$).  In the context of CSF, this translates into two conditions that a scheduler should satisfy:

\textbf{i) } If the fluctuation of carbon intensity is \textit{low} (e.g., $L$ and $U$ are close), the CSF should be close to $1$, i.e., JCT should be close to that of the carbon-agnostic algorithm.  

\textbf{ii) } If the fluctuation is \textit{high} (e.g., $L$ and $U$ are not close), the CSF should be \textit{finite}, i.e., the scheduler does not wait indefinitely to complete the job.  In threshold-based designs, this is often met by imposing a \textit{deadline} on the job~\cite{Goiri:2012:GreenHadoop, Lechowicz:23}.

In the context of the DAG scheduling for data processing workloads, additional unique challenges exist. For instance, in the single job settings considered by prior work, specifying a deadline for each job is straightforward~\cite{Lechowicz:23, Bostandoost:24}.  However, on a cluster scale that considers multiple jobs of unknown length and different arrival times, setting a proper deadline quickly becomes complicated. Instead, our schedulers (see \autoref{sec:design}) guarantee a minimum amount of job progress whenever there are outstanding tasks in the queue.  

Due to precedence constraints, carbon-aware scheduling actions that do not consider the structure of the DAG may inadvertently block bottleneck tasks from processing, having a large negative impact on JCT.  This gives a third condition:

\textbf{iii) } When fluctuation is \textit{high} (i.e., $L$ and $U$ are not close) and the system is in a high-carbon period,
a scheduler should carefully consider the \textit{structure} of a job's DAG, prioritizing \textit{bottleneck} tasks to use the limited cluster resources.

\noindent Conditions \textbf{i - iii)} summarize the desired high-level behavior of a carbon-aware scheduler for data processing workloads.  In the following section, we present \PCAPS that takes into account the above conditions in its design.

\section{Design} \label{sec:design}

\noindent In this section, we present \PCAPS, our Precedence- and Carbon-Aware Provisioning and Scheduling system, and then, as a flexible and easy-to-implement alternative, we present \CAP (Carbon-Aware Provisioning).

\vspace{-0.5em}
\subsection{\PCAPS} \label{sec:danish-design}

From the discussion in \autoref{sec:theory}, we seek an interpretable and configurable scheduler that satisfies the conditions $\textbf{i - iii)}$.
To this end, we introduce \PCAPS (Precedence- and Carbon-Aware Provisioning and Scheduling), which interfaces with a probabilistic DAG scheduler such as Decima~\cite{Hongzi:2019:Decima}.

\PCAPS's key idea is a metric of \textit{relative importance} (\sref{Def.}{dfn:rel-imp}) that is implicitly embedded in a probability distribution over tasks.  We define a configurable carbon and importance-aware threshold function that uses the \textit{relative importance} metric to make per-task fine-grained scheduling decisions -- as illustrated in \autoref{fig:danish-diagram}. 
Next, we detail how the theoretical literature on threshold-based algorithms inspires \PCAPS, describe its operation, and discuss analytical results that characterize the trade-off between carbon savings and JCT. 

\noindent \textbf{\PCAPS design. }
We first formalize a class of \textit{probabilistic schedulers} that \PCAPS interfaces with, giving a concrete example of an ML scheduler in this class.

\vspace{-0.5em}
\begin{dfn}[Probabilistic Scheduler] \label{dfn:pb}
At each \textbf{scheduling event},\footnote{Scheduling events include job arrivals, task completions, and machines becoming available.} a probabilistic scheduler generates a distribution $\{ p_{v,t} : v \in \mathcal{A}_t\}$, where $\mathcal{A}_t$ denotes the set of tasks that are ready to be executed at time $t$.
\end{dfn}
\vspace{-0.5em}

\noindent One example of a probabilistic scheduler is Decima~\cite{Hongzi:2019:Decima}, an RL-based scheduler for data processing workloads.  
Decima learns actions in the form of scores for each task -- a masked softmax is applied to these scores to obtain a probability distribution over $\mathcal{A}_t$, and the next scheduled task is sampled from this distribution.
Recall the motivation behind \PCAPS: 
in addition to ramping down during high-carbon periods and ramping up during low-carbon periods, bottleneck tasks (i.e., tasks with a large score) should be scheduled even if the carbon intensity is high to reduce JCT.
To this end, we define a notion of \textit{relative importance} that compares the probability mass assigned to a single task $v$ against other tasks in $\mathcal{A}_t$.

\begin{figure}[t]
    \centering 
    \includegraphics[width=0.95\linewidth]{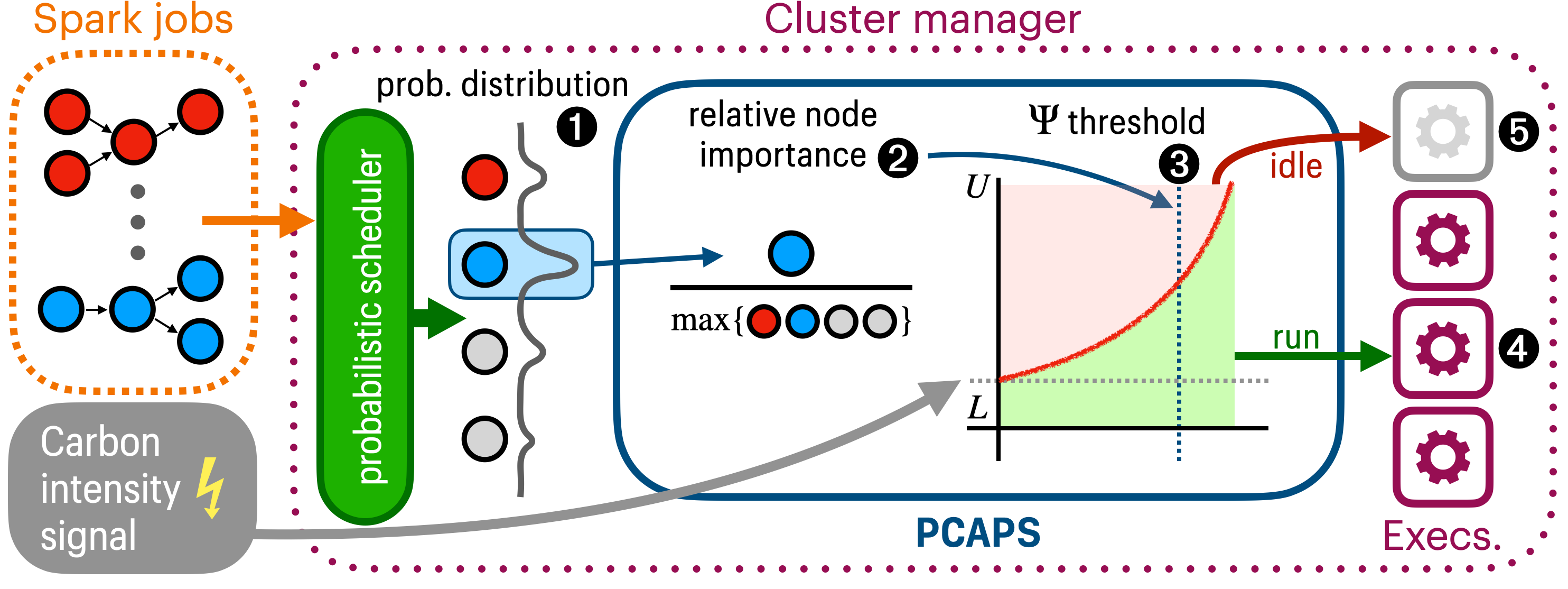} \vspace{-1em}
    \caption{\PCAPS interfaces with a probabilistic (\texttt{PB}) scheduling policy. Given a probability distribution over nodes \ding{202}, \PCAPS computes a \textit{relative importance} score \ding{203} that is used to determine which nodes should run based on the current carbon intensity \ding{204} -- e.g., bottleneck nodes impeding job completion run regardless of carbon \ding{205}, while less important nodes can be deferred for lower carbon periods \ding{206}.}
    \label{fig:danish-diagram} 
\end{figure}

 \vspace{-0.5em}
\begin{dfn}[Relative Importance] \label{dfn:rel-imp}
Given a time $t \geq 0$ and node $v \in \mathcal{A}_t$, the relative importance $r_{v,t}$ is defined:
{\small
\begin{equation*}
    r_{v,t} \coloneqq \frac{p_{v,t}}{\max_{u \in \mathcal{A}_t} p_{u,t}} \in [0,1].
\end{equation*}}
\end{dfn}
 \vspace{-0.5em}

\noindent If a task's relative importance is closer to $1$, the task is relatively \textit{more important}, and a value closer to $0$ implies the opposite.
Note that if $\vert \mathcal{A}_t \vert =1$ (i.e., only one task can be scheduled), the importance of that task is $1$. Leveraging inspiration from threshold-based design, we define scheduling decisions using a threshold function $\Psi_\gamma$ that considers the current carbon intensity and the relative importance of a task.
$\gamma \in [0,1]$ is a carbon-awareness parameter that controls the ``strictness'' of the function: $\gamma = 0$ recovers carbon-agnostic actions, while $\gamma = 1$ is maximally carbon-aware for tasks with low relative importance.  We define $\Psi_\gamma$ as:
{\small
\begin{equation*}
    \Psi_\gamma(r) \coloneqq \left( \gamma L + (1-\gamma) U \right) + \left[ U - \left( \gamma L + (1-\gamma) U \right)\right] \frac{\exp (\gamma r) -1 }{\exp (\gamma) -1}, \label{eq:Psi}
\end{equation*}
}
\noindent The function $\Psi_\gamma(.)$ exhibits an exponential dependence on $r$, the relative importance of a task.  This draws on continuous versions of online search~\cite{ElYaniv:01, Zhou:08}, where an exponential trade-off is derived by balancing the marginal reward of the current price against the risk that better prices exist in the future. 
We interpret relative importance analogously: high-importance tasks are scheduled regardless of carbon intensity to avoid the negative impact of not scheduling them, while low-importance tasks can be deferred to wait for lower-carbon periods with less impact on JCT. 
This threshold function is used in a \textit{carbon-awareness filter} of sampled tasks before they are scheduled-- we formalize this in \autoref{alg:danish}:

\begin{algorithm}[t]
	\caption{\PCAPS (Precedence- and Carbon-Aware Provisioning and Scheduling) }
	\label{alg:danish}
    {\small
	\begin{algorithmic}[1]
		\State \textbf{input:} hyperparameter $\gamma$, threshold function $\Psi_\gamma(\cdot)$, probabilistic (carbon-agnostic) scheduler \texttt{PB}
        \State \textbf{define:} a \textit{scheduling event} occurs whenever \texttt{PB} is invoked or the carbon intensity $c(t)$ changes
        \While{cluster active at time $t \ge 0$} 
        \If{scheduling event at time $t$} 
        \State Sample $v \in \mathcal{A}_t$ and probabilities $p_{v,t} : v \in \mathcal{A}_t$ from \texttt{PB}
        \State Compute relative importance $r_{v,t} = \frac{p_{v,t}}{\max_{u \in \mathcal{A}_t} p_{u,t}}$
        \If{$\Psi_\gamma(r_{v,t}) \geq c(t)$ \textbf{or} no machines currently busy}
        \State Send task $v$ to an available machine at time $t$
        \Else
        \State Idle until next scheduling event
        \EndIf
        \EndIf
        \EndWhile
	\end{algorithmic}
    }
\end{algorithm}	

\noindent
\PCAPS's carbon-awareness filter accomplishes all three of the motivation points defined in \autoref{sec:motiv}.
It schedules (or defers) tasks based on the current carbon intensity $c(t)$, with the effect of reducing execution during high-carbon periods.  Furthermore, the likelihood of a task being scheduled irrespective of the current carbon intensity is proportional to its importance in the DAG (i.e., in terms of precedence constraints).
Note that $\Psi_\gamma(1) = U$ -- tasks with high relative importance are always scheduled.  
While deferring a task has a negative impact on an individual job's JCT, \PCAPS is optimized for the case where multiple DAGs share a cluster.  Prioritizing outstanding bottleneck tasks across all jobs helps to manage the system's \textit{end-to-end completion time} (ECT) for e.g., a set of jobs.
In \autoref{fig:rel-imp}, we illustrate the intuitions behind \PCAPS's carbon-awareness filter using two sample job DAGs. 

\noindent\textbf{Analytical results. }
We analyze the \textit{carbon stretch factor} (\sref{Def.}{dfn:csf}) and \textit{carbon savings} (\sref{Def.}{dfn:carbonsavings}) for \PCAPS, with full proofs in \sref{Appendix}{apx:danish-proofs}.  $\texttt{PB}$ denotes a carbon-agnostic probabilistic scheduler throughout.
For an arbitrary job $\mathcal{J}$, we let $\mathcal{D}(\gamma, \mbf{c}) \in [0,1]$ denote a function that depends on the \textit{expected amount of deferrals}
(i.e., the tasks that $\PCAPS$ prevents from being scheduled, which depends on the carbon signal $\mbf{c}$).  See \sref{Appendix}{apx:danishMakespan} for a formal definition.
\begin{thm}\label{thm:danishMakespan}
    For time-varying carbon intensities given by $\mbf{c}$, the carbon stretch factor of \PCAPS is $1 + \frac{\mathcal{D}(\gamma, \mbf{c}) K}{2 - \frac{1}{K}}$.
\end{thm}
\noindent At a high level, $\mathcal{D}(\gamma, \mbf{c})$ describes the fraction of tasks (in terms of total runtime) that are deferred by \PCAPS with a given $\gamma$ and carbon trace $\mbf{c}$.  It is $\leq 1$ for any $\gamma$, and $\mathcal{D}(0, \mbf{c}) = 0$ for any $\mbf{c}$.  As $\gamma$ grows and \PCAPS becomes ``more carbon-aware'',  $\mathcal{D}(\gamma, \mbf{c})$ grows and CSF increases accordingly.

\begin{figure}[t]
    \centering 
    \includegraphics[width=\linewidth]{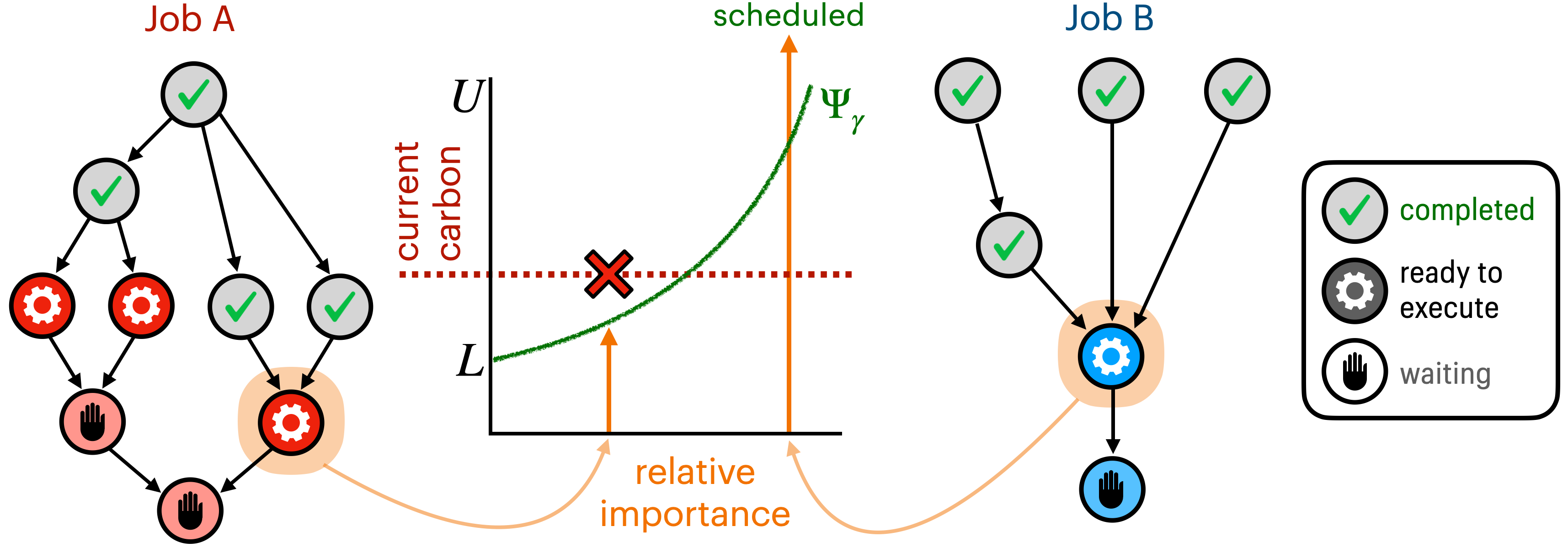} \vspace{-2em}
    \caption{Illustrating \PCAPS's carbon-awareness filter.  Jobs A and B are  DAGs found in TPC-H queries and Alibaba traces, respectively~\cite{TPCH:18, Alibaba:18}.  Highlighted nodes explain two scheduling outcomes.  In job A, the sampled node has low relative importance, so it is deferred.  In contrast, job B's sampled node is a \textit{bottleneck} task with high relative importance: even when the current carbon intensity is high, such tasks are scheduled to avoid increasing job completion time.}
    \label{fig:rel-imp} \vspace{-1em}
\end{figure}

We also analyze the \textit{carbon savings} of \PCAPS.  
Suppose \texttt{PB}'s schedule finishes at time $T$, and \PCAPS's finishes at time $T'$ (where $T \leq T'$).   We let $W$ denote the \textit{excess work} that \PCAPS must ``make up'' with respect to \texttt{PB}'s schedule (i.e., due to deferrals) -- note that this implicitly depends on the carbon stretch factor.
We let $\overline{s}_{-}^{(0,T)}$, $\overline{s}_{+}^{(0,T)}$, and $\overline{c}^{(T, T')}$ denote weighted average carbon intensity values based on the schedules of \texttt{PB} and \PCAPS.  In short, $\overline{s}_{-}^{(0,T)}$ captures the carbon emissions that \PCAPS avoids due to deferrals between time $0$ and time $T$,  
$\overline{s}_{+}^{(0,T)}$ captures the extra carbon (if any) incurred by \PCAPS due to higher utilization relative to \texttt{PB} between time $0$ and time $T$, %
while $\overline{c}^{(T, T')}$ captures the emissions that \PCAPS incurs after time $T$.  
See \sref{Appendix}{apx:danishMakespan} for formal definitions of these quantities.
\begin{thm}\label{thm:danishCarbonSavings}
   For time-varying carbon intensities given by $\mbf{c}$, \PCAPS yields carbon savings of $W ( \overline{s}_{-}^{(0,T)} - \overline{s}_{+}^{(0,T)} - \overline{c}^{(T, T')} )$.
\end{thm}
\noindent Taken together, \autoref{thm:danishMakespan} and \ref{thm:danishCarbonSavings}  characterize 
the carbon-time trade-off for \PCAPS, implying that a larger CSF unlocks greater potential carbon savings.

\vspace{-0.5em}
\subsection{Carbon-aware provisioning (\CAP)}  \label{sec:cap-design}

While \PCAPS captures all three intuition points in \autoref{sec:theory} by interfacing with a probabilistic scheduler, many existing data processing schedulers use simple policies such as FIFO~\cite{SparkScheduling}.  This naturally prompts the question of how \PCAPS can be \textit{simplified} to retain many of the same qualities, while interoperating with \textit{any} scheduler.
In particular, \PCAPS implicitly performs \textit{resource provisioning}, changing the amount of resources available to the cluster based on carbon -- this is a key technique used by prior work in carbon-aware scheduling~\cite{radovanovic2022carbon, Hanafy:23:CarbonScaler}.  In this section, we introduce \CAP (Carbon-Aware Provisioning), a simplified policy that applies a time-varying resource quota to the cluster and coexists with any underlying scheduler.  In what follows, we motivate the design and discuss analytical results on its carbon-JCT trade-off.

\begin{figure}[t]
    \centering 
    \includegraphics[width=\linewidth]{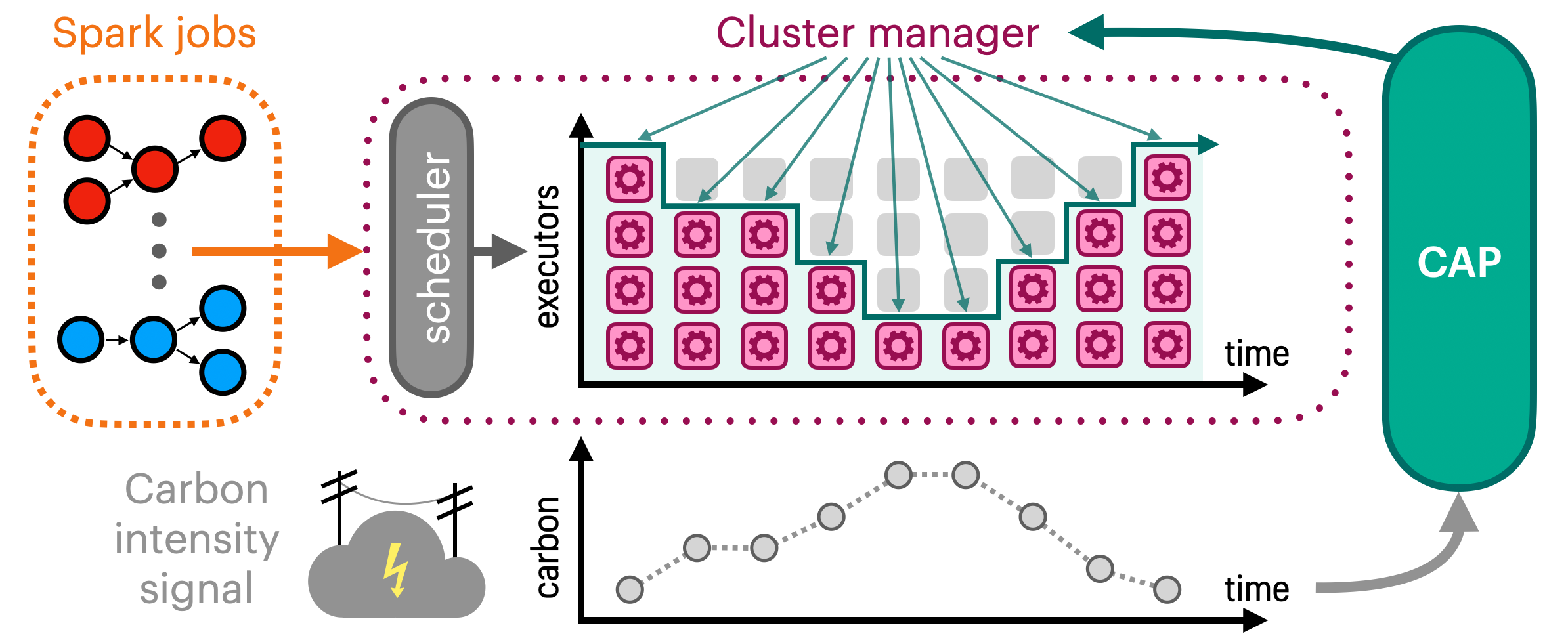}\vspace{-1em}
    \caption{The \CAP (Carbon-Aware Provisioning) module interacts directly with a \textit{cluster manager} %
    to specify the \textit{amount of resources} (e.g., no. of machines) that can be used at any given time, based on a \textit{carbon intensity signal}.  \CAP can be implemented without changes to an existing scheduling policy and/or the cluster manager.}
    \label{fig:cap-diagram}\vspace{-0.5cm}
\end{figure}

\noindent\textbf{\CAP design. } Given a cluster with $K$ machines, the possible resource quotas are given by $\{0, 1, \dots, K\}$. 
This set calls to mind the $k$-search problem~\cite{Lorenz:08}, where an online player must choose when to purchase $k$ items over a deadline-constrained sequence of time-varying prices. Variants of $k$-search have been applied to carbon-aware scheduling with deadlines~\cite{Lechowicz:23, Hanafy:23}.
\CAP uses the $k$-search threshold set, which captures the trade-off between executing now and waiting for better carbon intensities. 
Instead of using a deadline, \CAP frames the problem of determining a resource quota as \textit{repeated rounds} of $(K-B)$-search, where a minimum quota $B \in \{1, \dots, K\}$ always allows the cluster to use $\le B$ machines, ensuring continuous progress on jobs.  The thresholds are given by:
{\small
\begin{equation*}
\Phi_B = U; \ 
\Phi_{i+B} = U - \left(U - \frac{U}{\alpha} \right) \left( 1 + \frac{1}{\left(K-B\right)\alpha} \right)^{i-1} \kern-1em :  i \in \{1, \dots, K-B\},
\end{equation*}}
\noindent where $\alpha$ is the solution to $\left( 1 + \frac{1}{\left(K-B\right)\alpha} \right)^{\left(K-B\right)} \kern-1em = \frac{U-L}{U(1 - \frac{1}{\alpha})}$.  Each of these thresholds corresponds to a carbon intensity, and a quota is set based on how many values are \textit{above} the current carbon intensity.  Formally, the resource quota at time $t$ is $r(t) \gets \arg \max_{i\in \mathcal{R}} \Phi_i : \Phi_i \le c(t)$.  For ease of implementation, this quota is enforced without preemption; when machines become available, new task assignments are only allowed if $r(t)$ is greater than the number of busy machines.

\noindent\textbf{Analytical results. }
We analyze the \textit{carbon stretch factor} (\sref{Def.}{dfn:csf}) and \textit{carbon savings} (\sref{Def.}{dfn:carbonsavings}) for \CAP, with full proofs in \sref{Appendix}{apx:cap-proofs}.  $\texttt{AG}$ denotes a carbon-agnostic baseline scheduler throughout. Suppose \CAP's schedule completes at time $T'$.
We let $\mathcal{M}(B, \mbf{c})$ denote the minimum resource cap specified by \CAP at any point in its schedule (note this depends on the carbon signal $\mbf{c}$).  Formally, $\mathcal{M}(B, \mbf{c}) \coloneqq \arg \max_{i\in[K]} \Phi_i : \Phi_i \le c(t) \ \forall t \in [0, T']$.

\begin{thm}\label{thm:ksMakespan}
    For time-varying carbon intensities given by $\mbf{c}$, the carbon stretch factor of \CAP is $\left(\frac{K}{\mathcal{M}(B, \mbf{c})}\right)^2 \frac{2\mathcal{M}(B, \mbf{c})-1}{2K-1}$.
\end{thm}

\noindent We also analyze the \textit{carbon savings} of \CAP.  
If \texttt{AG}'s schedule finishes at time $T$ (where $T \leq T'$), we use $W$ as shorthand to denote the \textit{excess work} that \CAP must complete after time $T$ (i.e., after $\texttt{AG}$ has completed).  
As in \autoref{thm:danishCarbonSavings}, we let $\overline{s}^{(0,T)}$ and $\overline{c}^{(T, T')}$ denote weighted average carbon intensity values based on the schedules of \texttt{AG} and \CAP, respectively -- in short, $\overline{s}^{(0,T)}$ captures the carbon emissions that \CAP avoids by deferring $W$ amount of work relative to \texttt{AG}, while $\overline{c}^{(T, T')}$ captures the emissions that \CAP incurs after time $T$.  
See \sref{Appendix}{apx:ksCarbonSavings} for formal definitions of all three quantities.

\begin{thm}\label{thm:ksCarbonSavings}
   For time-varying carbon intensities given by $\mbf{c}$, \CAP yields carbon savings of $W ( \overline{s}^{(0,T)} - \overline{c}^{(T, T')} )$.
\end{thm}

\noindent \autoref{thm:ksMakespan} and \ref{thm:ksCarbonSavings} imply that a larger CSF unlocks greater carbon savings for \CAP.  We explore the relative performance of \PCAPS vs. \CAP in our experiments, in \autoref{sec:eval}.

\section{Implementation} \label{sec:impl}
\noindent We have implemented proof-of-concepts of \DANISH and \CAP for Apache Spark on Kubernetes -- see \autoref{sec:spark-k8s-imp} for details.  
We also conduct large-scale experiments in a realistic Spark simulator -- see \autoref{sec:spark-sim-imp} for how we extend an existing simulator~\cite{Hongzi:2019:Decima} to evaluate carbon-aware scheduling policies.
\vspace{-0.5em}

\subsection{Spark and Kubernetes integration} \label{sec:spark-k8s-imp}

\noindent \textbf{Resource scaling \& stage scheduling.\ }
In Spark deployed on a Kubernetes cluster, 
each application is submitted to the API server~\cite{k8s:2015} that creates a ``driver'' running in a pod.  
We use Spark's dynamic allocation feature, which enables the driver to create executor pods dynamically as needed by the application -- these executors connect with the driver and execute application code. 
Kubernetes handles the scheduling of (driver and executor) pods for each application, while the Spark driver selects stages to execute within an application.

To implement \CAP, we develop a Python daemon that gets carbon intensity from an API (e.g., Electricity Maps~\cite{electricity-map}) and adjusts the resources available to Spark.  \CAP sets a \textit{resource quota}~\cite{kubernetesResourceQuotas} within a dedicated namespace for Spark apps -- our implementation adjusts \textit{CPU and memory} quotas to correspond with a maximum number of executors.  When the quota is \textit{lowered}, existing pods are \textit{not} preempted, but new pods are not scheduled until usage falls below the quota.  
We implemented \DANISH as a pluggable scheduling service that coordinates between Spark and Kubernetes. 
The service includes inference for Decima~\cite{Hongzi:2019:Decima}. 
\DANISH gets carbon intensity from an API and collects context about the cluster and job states from Kubernetes and Spark.

While \CAP can be implemented without modifications to Spark or Kubernetes, we made two key changes for \DANISH.
First, we implemented a Kubernetes \textit{scheduler plugin}~\cite{k8s-scheduler-plugins:21} that communicates with \DANISH to determine which application should receive available resources.  
This builds on source code APIs exposed by the default \verb|kube-scheduler| and requires building/configuring a custom scheduler pod.  We restrict the scope of our plugin to a dedicated namespace for Spark apps.
Next, we made changes to Spark~\cite{Spark:16} such that each application communicates with \DANISH before choosing the next stage for execution -- Spark provides scripts to build a pod Docker image~\cite{Merkel:14:Docker} based on a custom build.

\vspace{0.05cm}
\noindent \textbf{Setting level of parallelism.\ }
In a Spark DAG, each \textit{stage} (i.e., node) includes multiple tasks that are parallelizable over multiple executors.
Setting a \textit{parallelism limit} (number of executors working on a stage) is a key component of Spark scheduling (e.g., see \cite[Section 5.2]{Hongzi:2019:Decima}). 
More executors are not necessarily better: assigning many executors to a stage that does not benefit from parallelism \textit{blocks} them from working on other jobs in the queue.
For \textit{carbon-aware} scheduling, we enable \DANISH and \CAP to set new parallelism limits for the current job each time a stage is scheduled, and particularly to set \textit{lower} limits during high-carbon periods (e.g., see conditions \textbf{i)} and \textbf{ii)}, \autoref{sec:theory}).

In \DANISH, if a stage is deferred, it \textit{idles} (see \sref{Alg.}{alg:danish}) the newly freed executors that prompted a scheduling event.  Otherwise, the stage's parallelism limit is set to $P' \coloneqq \lceil P \cdot \min \left\{ \exp( \gamma (L - c_t) ), (1-\gamma) \right\} \rceil$, where $P$ is the limit chosen by Decima.  This mirrors the exponential trade-off in \DANISH's design -- e.g., when the current carbon $c_t$ is close to $L$ the limit is set to $\lceil (1- \gamma) P \rceil$, and as $c_t$ grows, the limit decreases exponentially to $1$.
For \CAP, given that the underlying scheduler specifies a parallelism limit $P$, \CAP first attempts to schedule a stage with $P' = \lceil P \cdot \nicefrac{r(t)}{K} \rceil$, where $\nicefrac{r(t)}{K}$ is the ratio of the resource quota vs. the total number of executors.  If the number of available executors is less than $P'$, the current stage takes all of the remaining available executors.

\vspace{-0.5em}
\subsection{Spark simulator environment} \label{sec:spark-sim-imp}
\citet{Hongzi:2019:Decima} developed a simulator %
that is a faithful representation of Spark's \textit{standalone mode}  (i.e., where Spark is the cluster manager), achieving an error (i.e., in run times) of within 5\%~\cite[Fig. 18]{Hongzi:2019:Decima}.
This simulator captures all first-order effects of Spark execution (e.g., delays in executor movement, parallelism overheads) -- it has since seen wide use in Spark contexts~\cite{Shin:24, Hu:24, Gertsman:23, Mathur:23, Li:23:JSS, Bengre:21}.
We implement \DANISH and \CAP as extensions to this simulator, which provides fast evaluation and flexibility.  We make the following modifications:

\noindent$\blacktriangleright$ \textit{Carbon accounting: } Each job's carbon footprint is measured \textit{ex post facto} to avoid impacting simulator fidelity.  Once an experiment is complete, existing computations (e.g., executor times) and a carbon trace are used to tally the footprint.

\noindent$\blacktriangleright$ \textit{\CAP: } We implement \CAP as a wrapper over three carbon-agnostic schedulers in the simulator: FIFO, Decima, and Weighted Fair (a heuristic tuned for the simulator's test jobs).

\noindent$\blacktriangleright$ \textit{\DANISH: } We implement \DANISH to interface with Decima, which provides a probability distribution over tasks.

\noindent With these modifications, the simulator allows us to quickly test many scenarios with a high degree of accuracy.
\vspace{-0.5em}

\section{Evaluation} \label{sec:eval}
\noindent We evaluate our carbon-aware schedulers in a prototype cluster and a realistic Spark simulator, using workloads from TPC-H benchmarks~\cite{TPCH:18} and Alibaba production DAG traces~\cite{Alibaba:18}.  We answer the following questions:
\begin{enumerate}[leftmargin=*]
    \item How do \DANISH and \CAP navigate the trade-off between carbon emissions and job completion time?
    \item How do \DANISH and \CAP adapt to changes in carbon intensity characteristics and workload characteristics?
\end{enumerate}
\vspace{-0.5em}

\subsection{Experimental setup} 
\label{sec:carbon-traces}

\begin{table}[t]
\caption{Summary of carbon intensity trace characteristics, including the duration, granularity, minimum, maximum, mean, and coefficient of variation (\textit{higher value implies more variation}) for carbon intensities.} \label{tab:characteristics} \vspace{-1em}
\begin{tabular}{|l|lllll|}
\hline
  & \multicolumn{5}{c|}{\begin{tabular}[c]{@{}c@{}}Avg. Carbon Intensity \vspace{-0.1em}\\ {\footnotesize \textit{(in gCO$_2$eq./kWh)} \cite{electricity-map}} \end{tabular}} \\ \cline{2-6} 
\multirow{-2}{*}{\begin{tabular}[c]{@{}c@{}}Grid \vspace{-0.1em}\\ Code \end{tabular}}           & \multicolumn{1}{l|}{ {\small Duration} }   & \multicolumn{1}{l|}{{\small Min.}}  & \multicolumn{1}{l|}{{\small Max.}} & \multicolumn{1}{l|}{{\small Mean}} & {\small Coeff. Var.}  \\ \hline
\worldflag[length=0.75em, width=0.6em]{US} {\small \texttt{\textbf{PJM}}}                               & \multicolumn{1}{l|}{}                                                                                      & \multicolumn{1}{l|}{293} & \multicolumn{1}{l|}{567} & \multicolumn{1}{l|}{425} & \multicolumn{1}{l|}{0.110}   \\ \cline{1-1} \cline{3-6} 
\worldflag[length=0.75em, width=0.6em]{US} {\small \texttt{\textbf{CAISO}}}      & \multicolumn{1}{l|}{}                                                                                      & \multicolumn{1}{l|}{83}  & \multicolumn{1}{l|}{451} & \multicolumn{1}{l|}{274} & \multicolumn{1}{l|}{0.309}    \\ \cline{1-1} \cline{3-6} 
\worldflag[length=0.75em, width=0.6em]{CA} {\small \texttt{\textbf{ON}}}   & \multicolumn{1}{l|}{}                                                                                      & \multicolumn{1}{l|}{12}  & \multicolumn{1}{l|}{179}  & \multicolumn{1}{l|}{50} & \multicolumn{1}{l|}{0.654}   \\ \cline{1-1} \cline{3-6}
\worldflag[length=0.75em, width=0.6em]{DE} {\small \texttt{\textbf{DE}}}   & \multicolumn{1}{l|}{}                                                                                      & \multicolumn{1}{l|}{130} & \multicolumn{1}{l|}{765} & \multicolumn{1}{l|}{440} &\multicolumn{1}{l|}{0.280}  \\ \cline{1-1} \cline{3-6} 
\worldflag[length=0.75em, width=0.6em]{AU} {\small \texttt{\textbf{NSW}}} & \multicolumn{1}{l|}{}                                                                                      & \multicolumn{1}{l|}{267} & \multicolumn{1}{l|}{817} & \multicolumn{1}{l|}{647} & \multicolumn{1}{l|}{0.143}    \\ \cline{1-1} \cline{3-6} 
\worldflag[length=0.75em, width=0.6em]{ZA} {\small \texttt{\textbf{ZA}}}     & \multicolumn{1}{l|}{\multirow{-6}{*}{\small \begin{tabular}[c]{@{}l@{}}01/01/2020-\\ 12/31/2022 \\ Hourly \\ granularity\\ 26,304\\ data points\end{tabular}}}                                                                                      & \multicolumn{1}{l|}{586} & \multicolumn{1}{l|}{785}  & \multicolumn{1}{l|}{713} & \multicolumn{1}{l|}{0.046}                                                                                               \\ \hline
\end{tabular}
\end{table}
\begin{figure}[t]
    \centering 
    \vspace{-1em}
    \includegraphics[width=\linewidth]{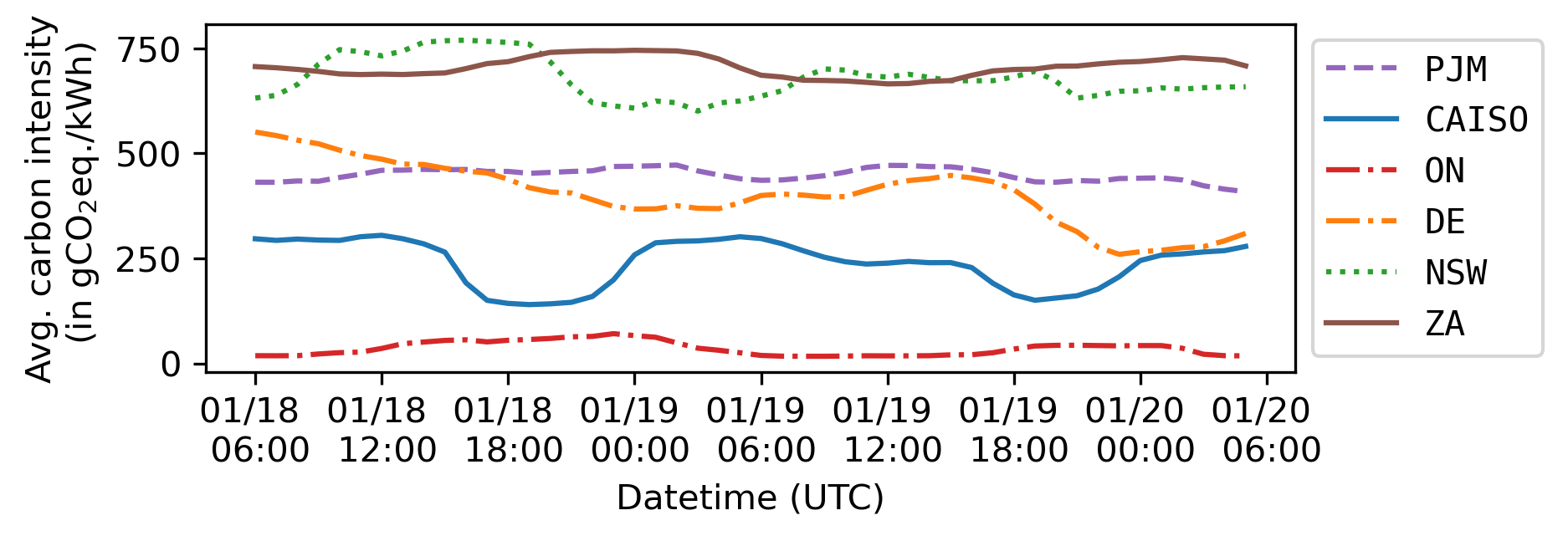} \vspace{-2.5em}
    \caption{ Time-varying carbon intensity for six grids (detailed in \autoref{tab:characteristics}) over 48 hours in January 2021.}
    \label{fig:CI-traces} \vspace{-1em}
\end{figure}

\noindent
\textbf{Carbon intensity traces.}
We use historical carbon traces from six regions~\cite{electricity-map} -- each trace provides hourly carbon intensity data in grams of CO$_2$ equivalent per kilowatt-hour (gCO$_2$eq./kWh).  The chosen power grids represent different energy generation mixes and thus different characteristics in terms of average carbon intensity and variability; we evaluate how these impact the behavior of \DANISH and \CAP.
In \autoref{tab:characteristics} and \autoref{fig:CI-traces}, we give snapshots of each region, showing how grid characteristics impact time-varying carbon intensity.  Larger \textit{coefficients of variation} (the ratio of the standard deviation to the mean) correspond to greater renewable penetration -- for instance, a large fraction of \verb|CAISO|'s capacity is solar PV, while the capacity in \verb|ZA| is predominantly coal.

To better observe the behavior of our carbon-aware schedulers, we follow prior work~\cite{Goiri:2012:GreenHadoop} and scale time in our experiments such that $1$ minute of \textit{real time} corresponds to $1$ hour of \textit{experiment time} -- since carbon intensity is reported hourly, this approximates a scenario where each job works with large amounts of data and runs for several hours, as is becoming common in e.g., data curation for LLMs~\cite{Liu:24, Chen:24, Villalobos:24, Bother:24, Brown:20}.

\vspace{0.05cm}
\noindent
\textbf{Workload traces.}
For workloads, we use TPC-H benchmarks~\cite{TPCH:18} and real DAG traces from a production Alibaba cluster~\cite{Alibaba:18}.  
We construct workloads such that the inter-arrival times follow a Poisson distribution while specific jobs are randomly picked from the respective traces. 
In the main body, we consider an average inter-arrival time of 30 minutes (30 real-time seconds), with additional experiments measuring the impact of this parameter in \sref{Appendix}{apx:exp}.

The TPC-H queries we experiment with operate on synthetic data with scales of 2 GB, 10 GB, and 50 GB -- these correspond to average real durations of 180 seconds, 386 seconds, and 1,261 seconds when given a single executor.
In our prototype experiments, we also construct workloads based on DAG information from the Alibaba trace~\cite{Alibaba:18}.  These DAGs exhibit a realistic power law distribution (many DAGs of short duration, few DAGs of long duration), they have 66 nodes on average, and an average total duration (on one executor) of 7,989 seconds. %
We scale all durations by $\nicefrac{1}{60}$ to match our experiment scale -- this yields jobs that take 2.2 real-time minutes to complete on average.  %

In the simulator, each experiment is run over a full carbon trace (spanning three years of data).  In the prototype, each experiment is run for several trials, where each starts at a uniformly randomly chosen time in the carbon trace, and new carbon intensities are reported once per (real-time) minute.  In both implementations, the upper and lower bounds of $U$ and $L$ correspond to the maximum and minimum forecasted carbon intensities over a lookahead window of 48 hours.

\begin{figure*}[t]
    \centering
    \includegraphics[width=0.95\linewidth]{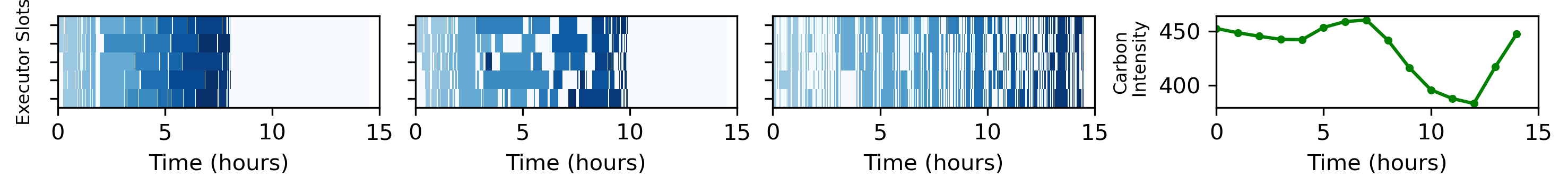}
    \vspace{-0.2cm}
    {\raggedleft
    \hspace{5.5em} \textbf{\textit{(a)}} Decima \hspace{6em} \textbf{\textit{(b)}} \PCAPS \hspace{6em} \textbf{\textit{(c)}} \CAP-FIFO \hspace{6em} \textbf{\textit{(d)}} Carbon intensity \hfill}
    \caption{ Visualizing executor usage over time for three schedulers, \textit{(a)} Decima, \textit{(b)} \PCAPS, and \textit{(c)} \CAP-FIFO in a small simulator cluster with 5 executors and 20 TPC-H jobs, over a 15 hour period in the \textnormal{\texttt{DE}} grid. \textit{(d)} In the executor plots, each job is a unique shade of blue, while ``idle'' executors are indicated by a white background. }
    \vspace{-0.4cm}
    \label{fig:multiplot}
\end{figure*}

\vspace{0.05cm}
\noindent\textbf{Baselines.}
We compare against the following baselines:

$\blacktriangleright$ \textbf{Default Spark/Kubernetes behavior (default)}:
The default behavior of Spark on Kubernetes -- Spark uses first in, first out (FIFO) to choose stages within a job, while the Kubernetes scheduler mediates between pods of each job during execution~\cite{SparkKubernetes}. 
In the simulator's Spark standalone mode, this baseline implements only the FIFO scheduling.

$\blacktriangleright$ \textbf{Decima}:
An RL scheduler for Spark that is optimized for job completion time~\cite{Hongzi:2019:Decima}.  We use the simulator's training environment to train Decima for 20,000 epochs.

$\blacktriangleright$ \textbf{Weighted Fair}: A heuristic that assigns executors proportionally to each job's workload, with tuned weights to improve performance on the simulated workloads~\cite{Hongzi:2019:Decima}. 

$\blacktriangleright$ \textbf{GreenHadoop:} A MapReduce framework proposed to leverage green energy by matching workloads with the availability of solar~\cite{Goiri:2012:GreenHadoop}.  This framework predates Spark, so we adapt its key ideas for DAG scheduling in the simulator -- see \autoref{apx:extra-details} for implementation details.

\vspace{0.05cm}
\noindent 
\textbf{Metrics.} We use three metrics to evaluate the fidelity of our approach in reducing carbon footprint without significantly impacting the completion time. 

\vspace{0.05cm}
$\blacktriangleright$ \textbf{Carbon Footprint}: 
We report the carbon footprint for various scheduling policies as a percentage decrease of the carbon-agnostic default baseline unless stated otherwise. 
The values are in the range of [-100\%, $\infty$), with negative values indicating a carbon reduction and positive values indicating an increase relative to the baseline. 
 
$\blacktriangleright$ \textbf{Job Completion Time (JCT)}:
We report the average job completion time across all the jobs in each experimental run. 
We report JCT as a fraction of the average JCT for the carbon-agnostic default baseline unless stated otherwise. 
The values can be in the (0, $\infty$) range, with below 1 indicating a reduction in JCT and above 1 indicating an increase in JCT.

$\blacktriangleright$ \textbf{End-to-end Completion Time (ECT)}:
We report the total time to complete all the jobs in a given experiment as the end-to-end completion time (ECT) as a fraction of the ECT for carbon-agnostic default. 
Its values lie in the same range as JCT. 
However, while JCT focuses on individual jobs, ECT represents the system's throughput and efficiency. 
Also, as \DANISH and \CAP focus on minimizing the total carbon footprint for a set of jobs and not the instantaneous rate of carbon consumption, ECT is a better metric for performance with respect to time in our case.

\subsection{Carbon-aware schedulers in action} \label{sec:in-action}

Before moving to our main results, we demonstrate the carbon-aware behavior of our schedulers in \autoref{fig:multiplot}, which visualizes the schedules generated by Decima, \PCAPS and \CAP-FIFO during a short period in the \verb|DE| grid on the simulator.
\PCAPS makes fine-grained scheduling decisions, idling specific executors during the high-carbon period ($t = (5,8)$) while keeping bottleneck tasks running, achieving the lowest carbon footprint of the three schedules shown.  This is in contrast to \CAP-FIFO, which applies a resource quota uniformly across the cluster without consideration of bottlenecks -- note the gaps in \CAP-FIFO's schedule that are straight vertical lines across multiple executors.  Just before $t = 5$, a large gap in the schedule indicates that \CAP-FIFO cannot run tasks because it did not prioritize bottleneck tasks early on.

\begin{table}[h]
\caption{Summary of prototype results averaged over all six carbon traces.  Each metric is normalized with respect to the Spark / Kubernetes default.  \DANISH and \CAP are configured to be moderately carbon aware. %
} 
\vspace{-1em} 
\label{tab:results-prototype}
{\small
\begin{tabular}{|l|l|l|l|l|}
\hline
{\footnotesize \textit{\textbf{\begin{tabular}[c]{@{}l@{}}Metric normalized \\ w.r.t. Default\end{tabular}}}} & {\footnotesize \textbf{Default} } & {\footnotesize \textbf{Decima}~\cite{Hongzi:2019:Decima} } & {\footnotesize \textbf{\CAP} } & {\footnotesize \textbf{\DANISH} } \\ \hline
{\footnotesize CO$_2$ Reduction (\%) } & 0\% & 1.2\% & 24.7\% & \textbf{32.9\%} \\ \hline
{\footnotesize Avg. ECT } & 1.0 & 0.857 & 1.126 & \textbf{1.013 }\\ \hline
{\footnotesize Avg. JCT } & 1.0 & 0.852 & 1.996 & \textbf{1.381} \\ \hline
\end{tabular}
} \vspace{-1em}
\end{table}

\subsection{Prototype experiments} \label{sec:eval-proto}

Our prototype is deployed on an OpenStack cluster running Kubernetes v1.31 and Spark v3.5.3 (both modified per \autoref{sec:impl}) in Chameleon Cloud~\cite{Keahey:20:Chameleon}. 
Our testbed consists of 51 \verb|m1.xlarge| virtual machines, each with 8 VCPUs and 16GB of RAM.  
One VM is designated as the control plane node, while the remaining 50 are workers, each hosting two executor pods.
Our Spark configuration allocates 4 VCPUs and 7GB of RAM to each of the 100 executors\footnote{ Spark's default \textit{memory overhead factor} is 10\%. The difference between 7GB ($\times 2$) and the 16GB RAM of each worker is to accommodate this~\cite{SparkConfiguration}.}.  
To avoid a known issue with Spark's dynamic allocation feature that can cause it to hang on Kubernetes~\cite{SparkKubernetes}, we configure an upper limit of 25 executors that can be allocated to any single job.  
We implement a carbon intensity API that replays historical traces to test our carbon-aware schedulers in the prototype.  
In our prototype, we implement default and Decima as the baselines. 
Unless stated otherwise, the results are averaged over the batch sizes of 25, 50, and 100 jobs. 
Furthermore, the results for each experimental configuration are averaged over 10 trials.

\begin{figure}[t]
    \centering 
    \includegraphics[width=0.8\linewidth]{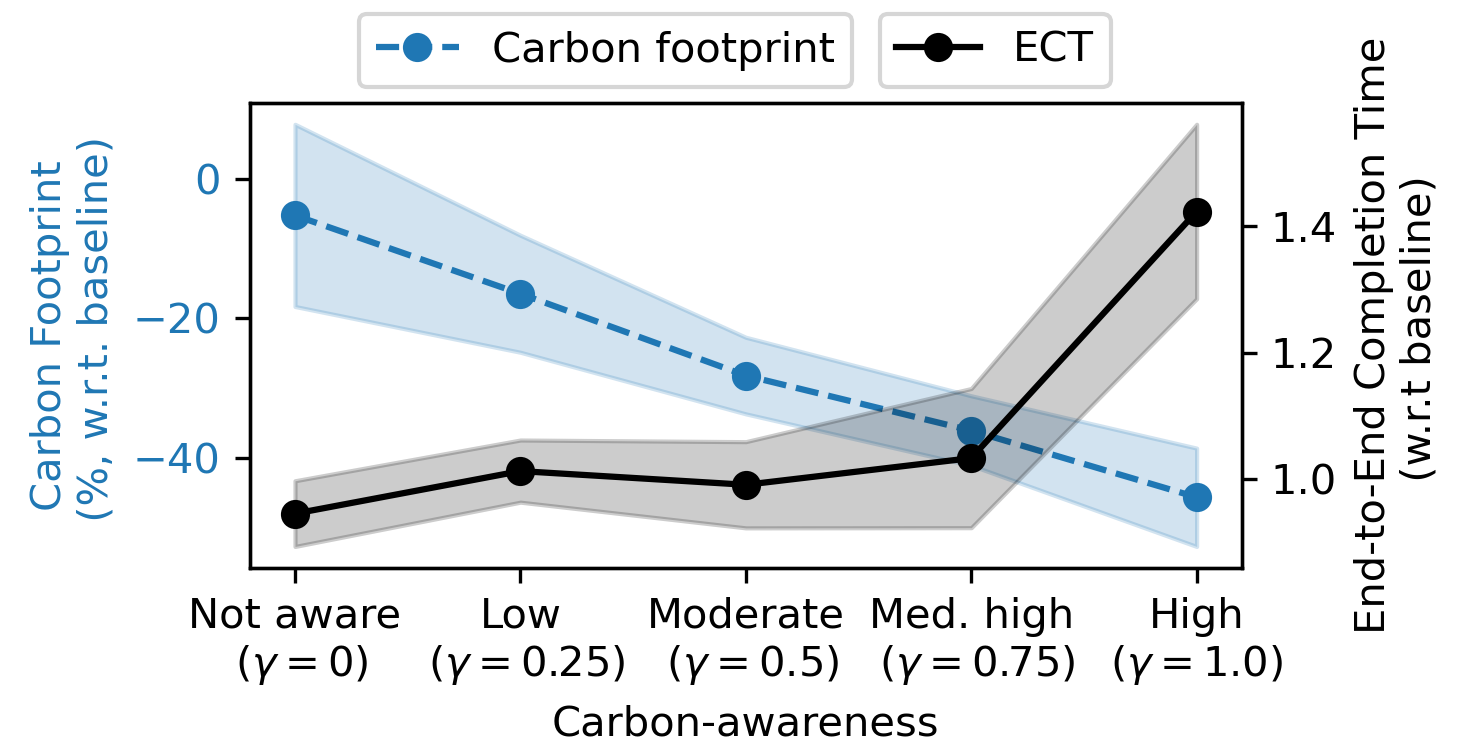} \vspace{-0.4cm}
    \caption{Relative carbon footprint and end-to-end completion times (w.r.t. the Spark/Kubernetes default) for \DANISH in the prototype, with five different degrees of carbon-awareness ($\gamma$). The shaded region denotes the standard deviation across 10 random trials. }
    \label{fig:gamma-tradeoff-proto}
    \includegraphics[width=0.8\linewidth]{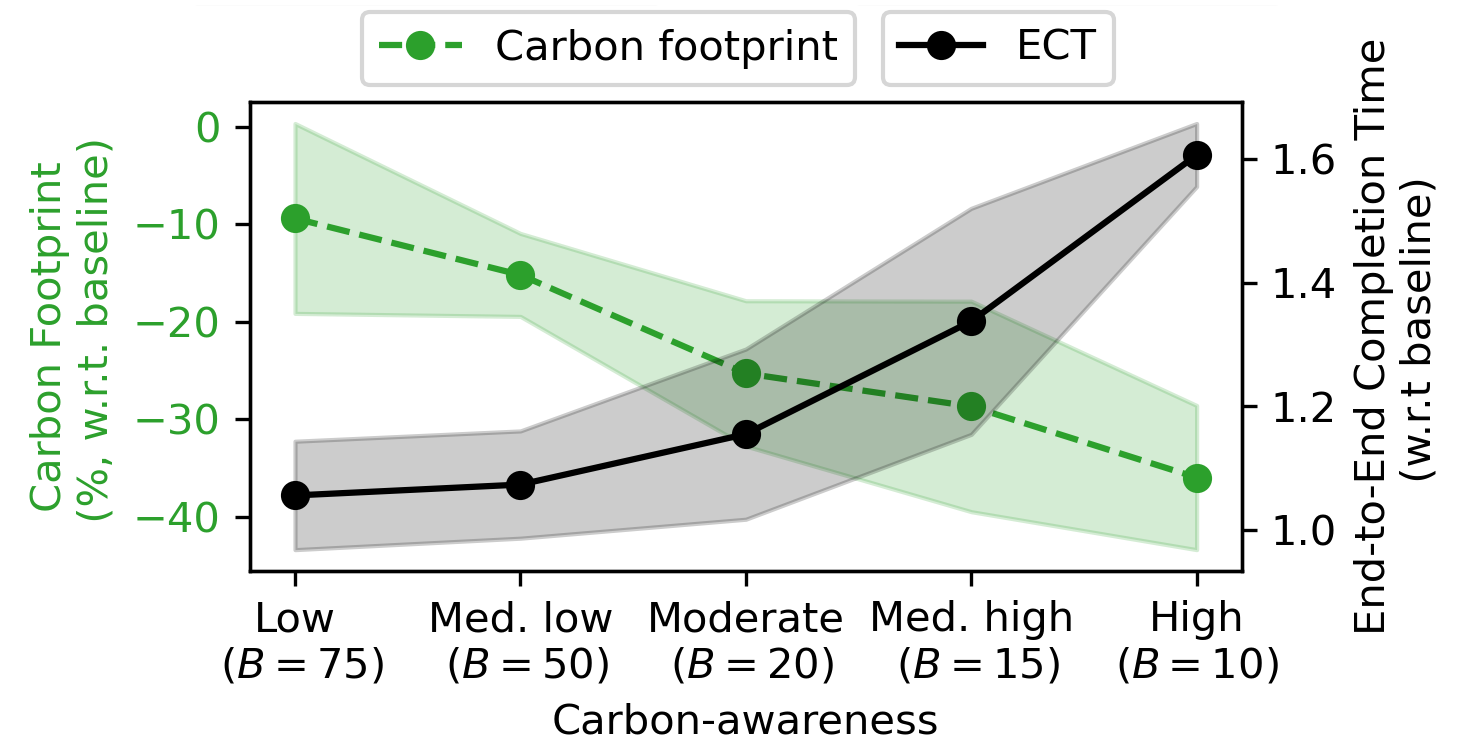} 
    \vspace{-0.4cm}
    \caption{ Relative carbon footprint and end-to-end completion times (w.r.t. the Spark/Kubernetes default) for \CAP in the prototype, with five different degrees of carbon-awareness ($B$).  The shaded region denotes the standard deviation across 10 random trials. }
    \label{fig:B-tradeoff-proto} \vspace{-0.5cm}
\end{figure}

\smallskip
\noindent \textbf{Results. \ } \autoref{tab:results-prototype} presents the results for our prototype experiments.
\DANISH and \CAP configured to be \textit{moderately carbon-aware} (\CAP is configured with $B = 20$ and \DANISH is configured with $\gamma = 0.5$) achieve average carbon reductions of 32.8\% and 24.6\% compared to the default baseline, respectively.  Compared to Decima, \DANISH reduces carbon by 32.1\%.  

In terms of JCT, \DANISH performs significantly worse than the default and Decima baselines; JCT increases by 38.1\% and 62\% with respect to the default and Decima, respectively. 
This is expected since Decima targets minimizing average JCT. 
However, the key objective of \DANISH is to reduce the total carbon footprint without increasing ECT. 
\DANISH increases average ECT by only 12.4\% compared to Decima and 1.3\% compared to the default.  \CAP also performs well and increases average ECT by 12.6\% on top of the default.

\begin{figure}[t]
    \centering 
    \includegraphics[width=0.8\linewidth]{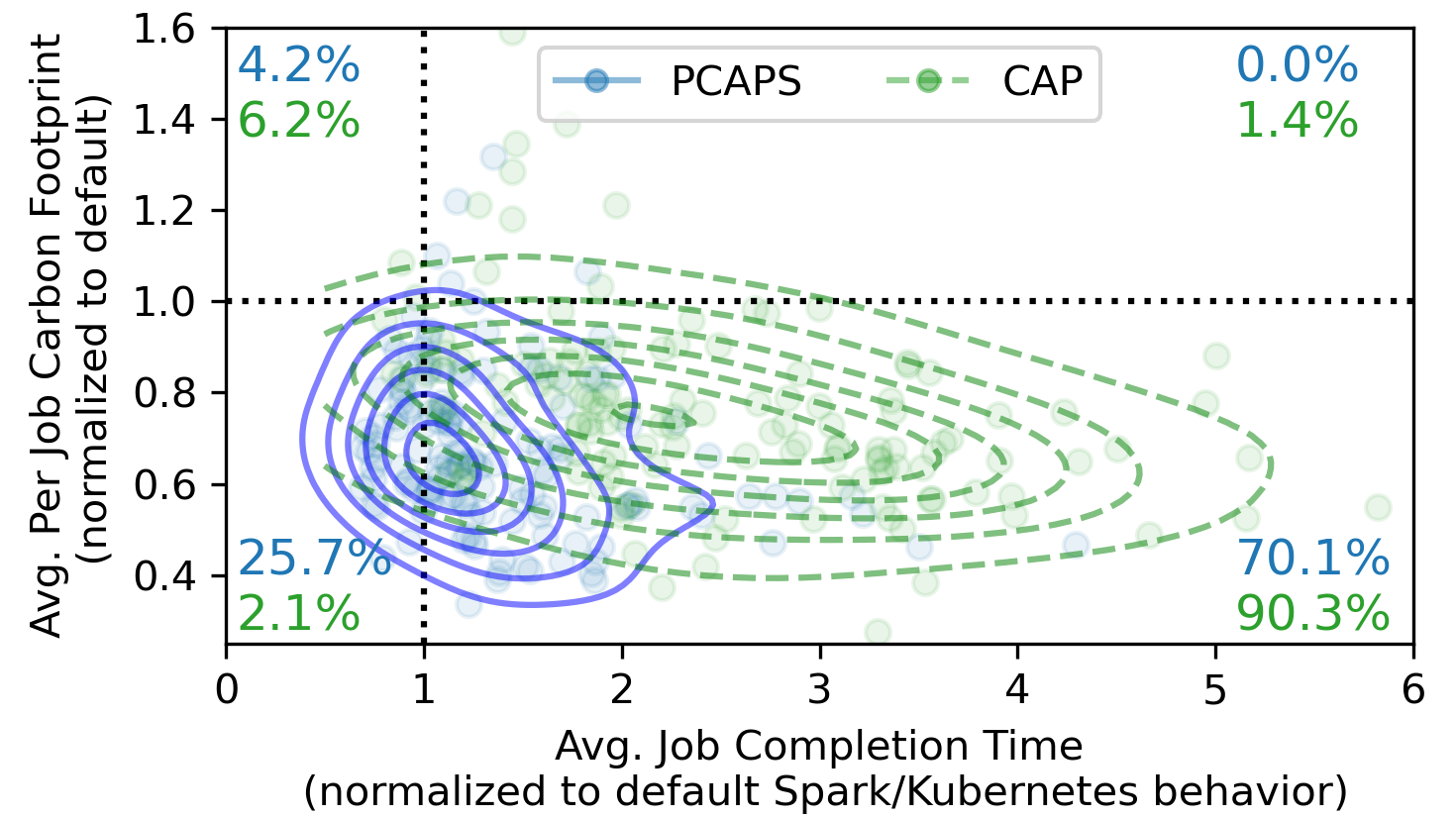} 
    \vspace{-0.4cm}
    \caption{ Per-job metrics for \DANISH and \CAP in the prototype -- each point represents a single trial's average JCT and per-job carbon footprint.  Trials are normalized so that the Spark/Kubernetes default is represented by $(1, 1)$ -- this splits the plot into quadrants, and each is annotated with the percentage of trials it holds (\DANISH is the top percentage, and \CAP is the bottom).  Contour lines outline a Gaussian KDE for each point cluster. %
    }
    \label{fig:scatter-proto} \vspace{-0.6cm}
\end{figure}

\noindent \textbf{Trade-offs between carbon and job completion time. \ }
We next test several parameter settings for \DANISH and \CAP to configure their carbon awareness in the \verb|DE| grid region with batches of 50 TPC-H or Alibaba jobs.  
\autoref{fig:gamma-tradeoff-proto} plots the carbon-time trade-off for five settings of \DANISH relative to the Spark/Kubernetes default.
Increasing the carbon awareness of \DANISH improves carbon savings at the expense of longer ECT, with a most pronounced effect for values of $\gamma$ approaching $1$.  
Conversely, \autoref{fig:B-tradeoff-proto} plots the same carbon-time trade-off for five settings of \CAP; \CAP sacrifices more in ECT (relative to \DANISH) for the same amount of carbon savings.

On a per-job level, we observe similar trends in \autoref{fig:scatter-proto}, which plots average JCT and average \textit{per-job} carbon footprint across trials of prototype experiments where \DANISH and \CAP are configured to be moderately carbon-aware.  Contour lines represent Gaussian kernel density estimators of the \textit{outcome distribution} -- note that the location of each ``hot spot'' represents the scheduler average.  Splitting the plot into quadrants, we observe that \DANISH improves on the baseline scheduler's per-job carbon footprint in 95.8\% of trials, corresponding to the lower two quadrants.  \DANISH improves on \textit{both} carbon and completion time in 25.7\% of cases, while \CAP does so in only 2.1\% of cases.

\noindent \textbf{Effects of carbon intensity trace characteristics. \ }
Next, we analyze the effect of grid characteristics on the carbon-time trade-offs of our carbon-aware schedulers using subsets of each carbon trace (via 30 trials with 25, 50, and 100 jobs).

\autoref{fig:traces-proto} plots the carbon reduction and average ECT of \DANISH, \CAP, and Decima.  
Decima is carbon-agnostic and shows a minimal reduction in carbon that stays relatively constant across all regions.  \DANISH and \CAP incorporate grid behavior into their decisions, and we observe a positive relationship between the variability of a carbon trace and the resulting carbon reduction.
For example, in \verb|ZA| where carbon is relatively constant, high-carbon periods do not prompt \DANISH to defer tasks since the \textit{future potential reductions} are insignificant.  
In contrast, high-carbon periods in the \verb|CAISO| grid correspond to nighttime scenarios on the grid, where the prospect of daytime solar bolsters future potential reductions.
These interactions are further illustrated through ECT -- grid regions with more intermittent and variable energy mixes drive increases in ECT in exchange for more carbon reduction, since \DANISH and \CAP \textit{wait} for potential reductions.

\begin{figure}[t]
    \centering 
    \includegraphics[width=0.9\linewidth]{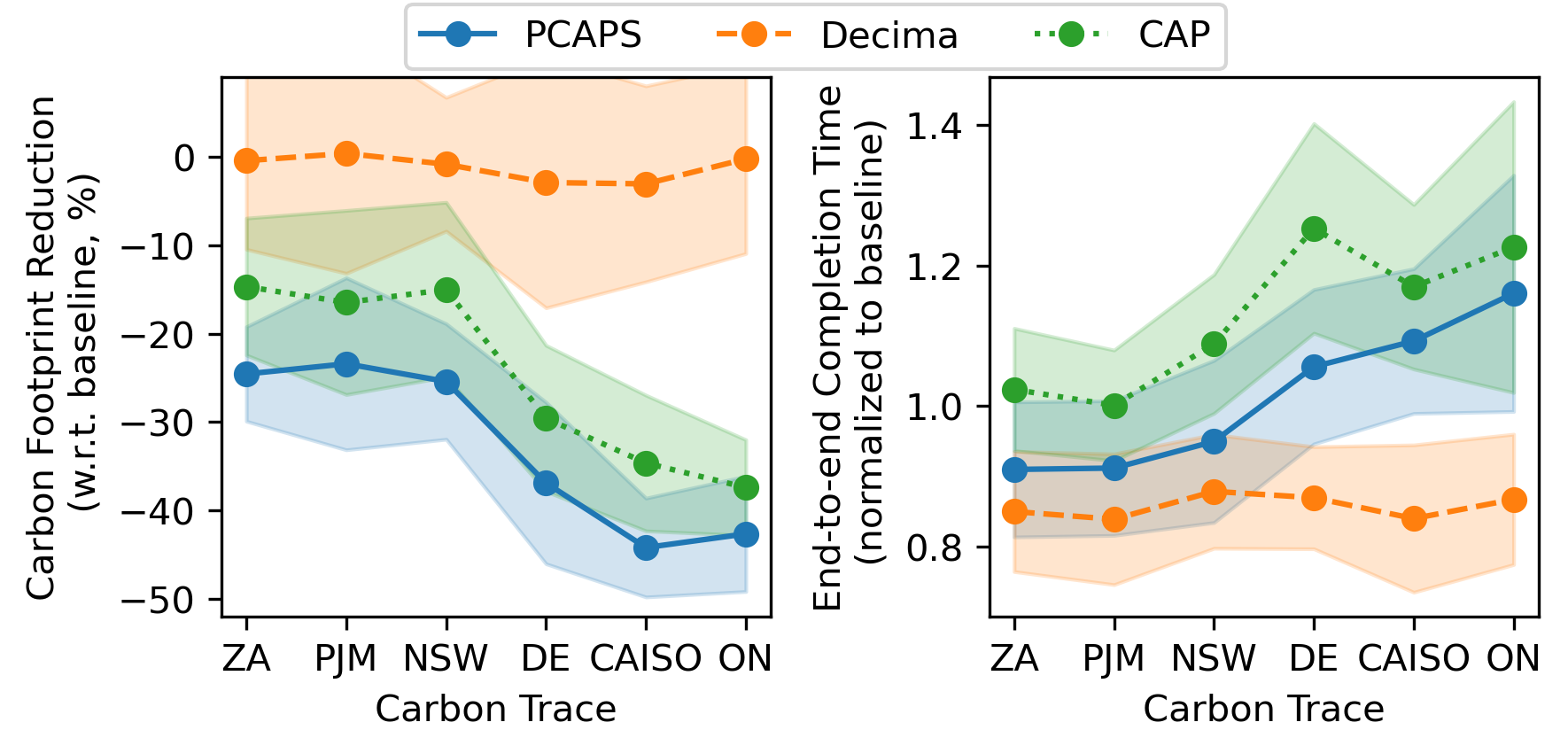}  
    \vspace{-0.4cm}
    \caption{ Carbon reduction \textit{(left)} and ECT \textit{(right)} for \DANISH, \CAP, and Decima in six grid regions. Shaded regions denote standard deviation across 30 trials. }
    \label{fig:traces-proto} 
    \vspace{-0.5cm}
\end{figure}

\vspace{-0.5em}

\subsection{Simulator experiments}
\label{sec:eval-sim}

In the simulator, we evaluate \DANISH and \CAP using TPC-H workloads, comparing them against \textbf{Weighted Fair} and \textbf{GreenHadoop} baselines, in addition to the Decima and default baselines from prototype experiments. 
We renamed the default baseline as FIFO for simulation-based experiments for accurate representation.

\noindent\textbf{Simulator fidelity.} To establish the fidelity of our simulator, we illustrate the granular effect of differences for a batch of 50 TPC-H jobs in \sref{Appendix}{apx:fifo-vs-k8s}.  
A notable difference between the prototype and the simulator is the relative performance of the main baseline (FIFO in the simulator, Spark/Kubernetes default in the prototype).  
In short, the simulator's FIFO scheduler \textit{over-assigns} executors to individual jobs, blocking others from entering service (thus increasing JCT) -- this also increases its relative carbon footprint compared to the default behavior of our prototype.

\begin{table*}[t]
\caption{Summary of results for simulator experiments averaged over all 6 tested carbon traces.  Each metric is normalized with respect to the default Spark FIFO behavior.  \DANISH and \CAP are configured to be moderately carbon-aware, and end-to-end completion time measures the total flow time for batches of jobs arriving continuously. } \vspace{-1em} \label{tab:results-simulator}
{\footnotesize
\begin{tabular}{|l|l|l|l|l|l|l|l|l|}
\hline
\multirow{2}{*}{\footnotesize \textbf{Metric} \textit{(normalized with respect to FIFO)}} & \multirow{2}{*}{\textbf{FIFO}} & \multirow{2}{*}{\textbf{W. Fair}} & \multirow{2}{*}{\textbf{Decima}~\cite{Hongzi:2019:Decima}} & \multirow{2}{*}{\textbf{GreenHadoop}~\cite{Goiri:2012:GreenHadoop}} & \multicolumn{3}{c|}{\textbf{\CAP}} & \multirow{2}{*}{\textbf{\DANISH}} \\ \cline{6-8} 
& & & & & FIFO & W. Fair & Decima & \\ \hline
Carbon Reduction (\%)                                                                & 0\% & 12.1\% & 21.5\% & 8.2\%  & 22.7\% & 34.2\% & 31.1\% & \textbf{39.7\%}                             \\ \hline
Avg. End-to-End Completion Time                                                 & 1.0 & 0.972 & 0.970 & 1.077 & 1.108 & 1.011 & 1.061 & \textbf{1.045 }                            \\ \hline
Avg. Job Completion Time                                                        & 1.0 & 0.652 & 0.654 & 1.918 & 2.274 & 1.217 & 1.479 & \textbf{1.436}                             \\ \hline
\end{tabular}
} \vspace{-0.1cm}
\end{table*}

\smallskip
\noindent \textbf{Results. \ } 
\autoref{tab:results-simulator} presents our top-line results showing \DANISH and \CAP achieve significant reductions in carbon emissions compared to the baselines. 
Configured to be moderately carbon-aware, \DANISH achieves an average reduction of 23.1\% compared to Decima, and a reduction of 39.7\% compared to FIFO.  \CAP achieves an average carbon reduction of 22.7\% when implemented on top of FIFO, 25.1\% on top of Weighted Fair, and 14.5\% on top of Decima. 
For 25, 50, and 100 jobs, \DANISH increases average end-to-end completion time (ECT) by 7.7\% with respect to Decima, which is only a 4.5\% degradation compared to FIFO.  For \CAP, average ECT increases by 10.8\% when implemented on top of (and compared to) FIFO, 4.0\% on top of Weighted Fair, and 9.3\% on top of Decima.   

At a per-job level, \DANISH increases the average job completion time (JCT) by 119.57\% compared to Decima.  Similarly, \CAP increases average JCT by 127.4\%, 86.6\%, and 126.8\% when implemented with FIFO, Weighted Fair, and Decima respectively.  Larger increases in JCT relative to ECT 
happen because \DANISH allows more queue build-up during high-carbon periods, but ``makes up for lost time'' in low-carbon periods.  

\begin{figure}[t]
    \centering 
    \includegraphics[width=0.8\linewidth]{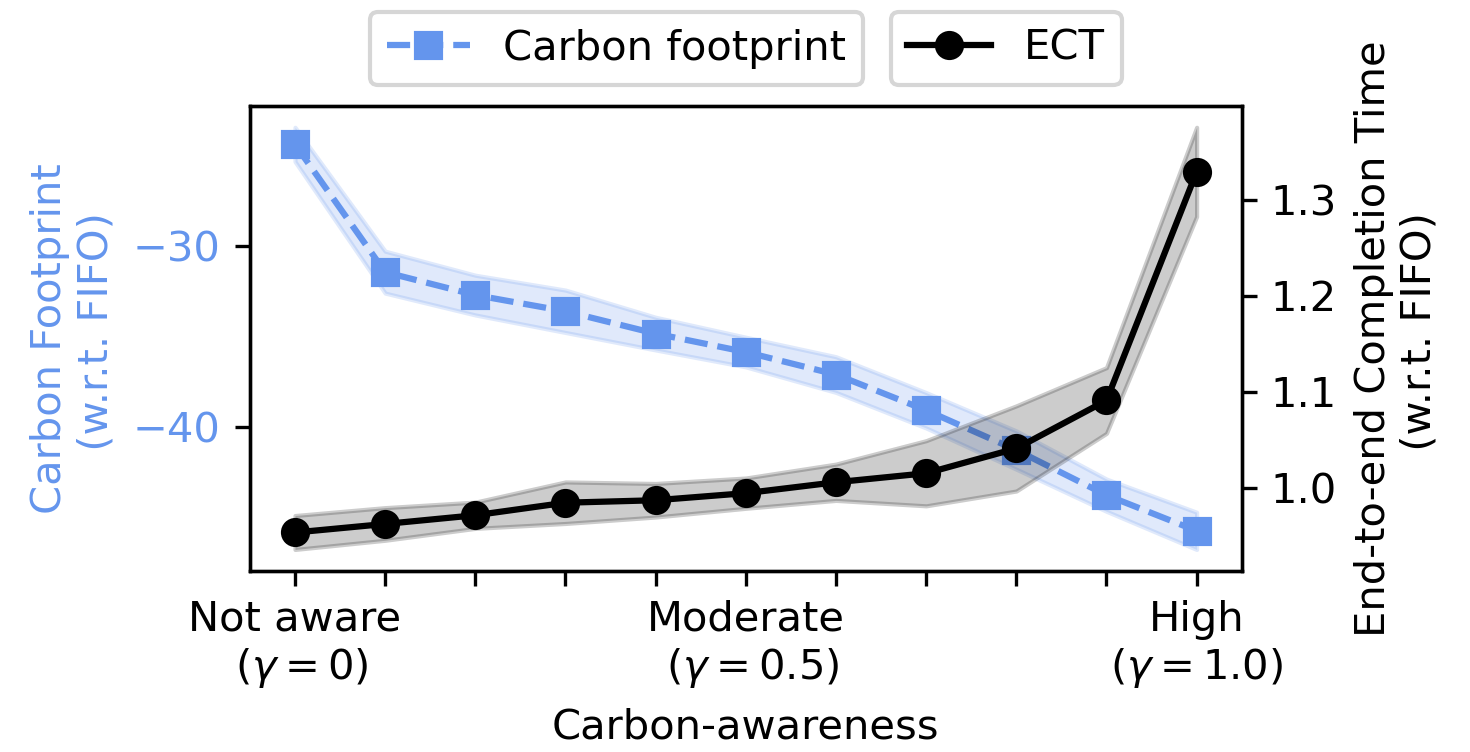} \vspace{-1em}
    \caption{ Relative carbon footprint and end-to-end completion times (with respect to FIFO) for \DANISH in simulator experiments, given different values of $\gamma$ that correspond to degrees of carbon-awareness. Shaded region denotes standard deviation across carbon trace. }
    \vspace{-0.5cm}
    \label{fig:gamma-tradeoff-sim}
\end{figure}
\begin{figure}[h]
    \centering 
    \includegraphics[width=0.8\linewidth]{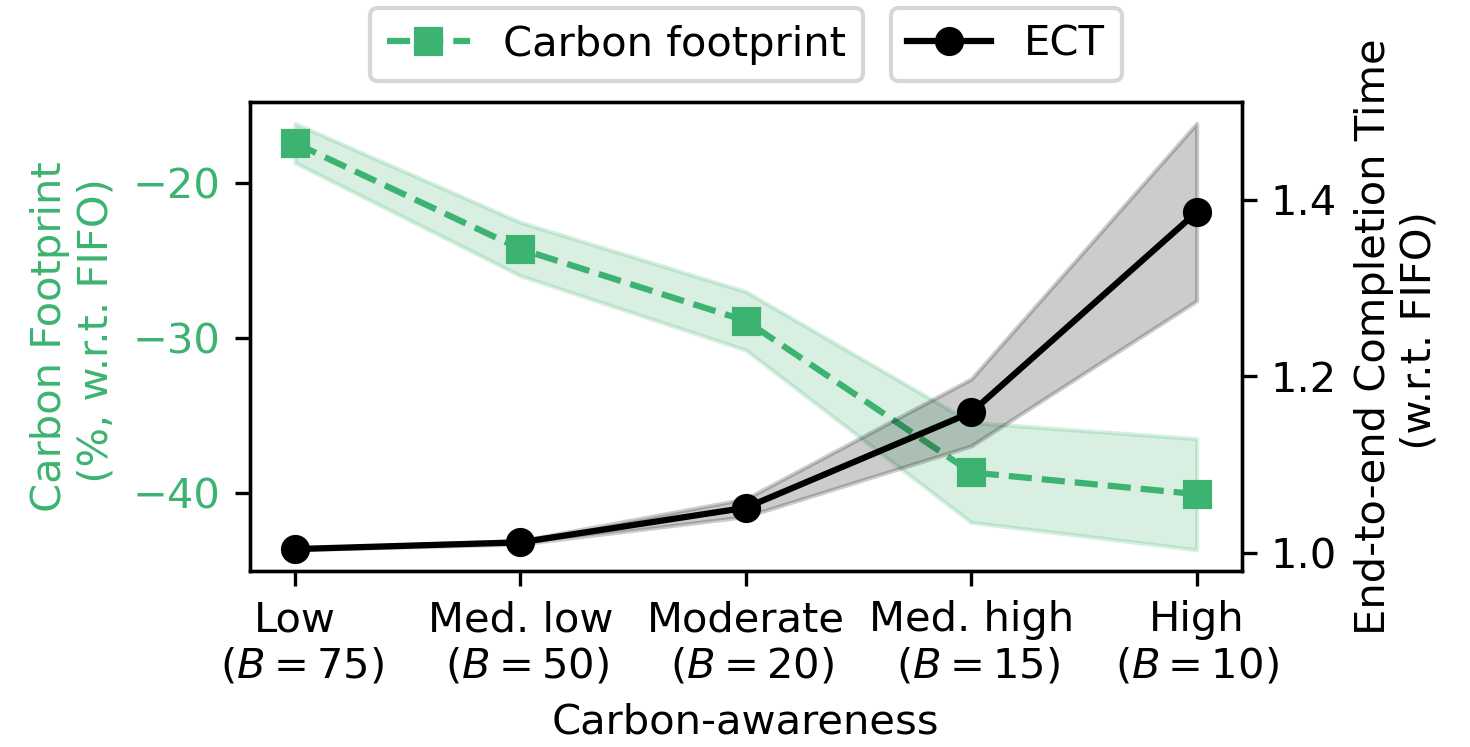} \vspace{-1em}
    \caption{ Relative carbon footprint and end-to-end completion times (with respect to FIFO) for \CAP-FIFO in simulator experiments, given different values of $B$ that correspond to degrees of carbon-awareness.  Shaded region denotes standard deviation across carbon trace. }
    \label{fig:B-tradeoff-sim} \vspace{-0.5cm}
\end{figure}

\noindent \textbf{Trade-offs between carbon and job completion time. \ }
In the results summary, we observe positive carbon reduction in exchange for degradation in JCT and ECT.  Since \DANISH and \CAP can be configured to be more or less carbon-aware, we explore this \textit{trade-off} in the \verb|DE| grid with batches of 50 jobs.  We vary hyperparameters $\gamma$ (\DANISH) and $B$ (\CAP) to measure the impact of configuration on both carbon and JCT.

In \autoref{fig:gamma-tradeoff-sim}, we illustrate this trade-off for \DANISH compared against FIFO.  As the carbon-awareness of \DANISH increases (indicated by the value of $\gamma$), the carbon savings of \DANISH improve at the expense of longer ECT.  This effect is most pronounced for large values of $\gamma$ approaching $1$, because $\DANISH$ defers many tasks to lower carbon periods.
Conversely, \autoref{fig:B-tradeoff-sim} illustrates the trade-off for \CAP-FIFO, showing a similar trend of improving carbon at the expense of longer ECT.  Compared to \autoref{fig:gamma-tradeoff-sim}, \CAP-FIFO sacrifices more in terms of ECT for the same or lower amounts of carbon savings, and the increase in completion time begins earlier (at lower degrees of carbon-awareness).

\noindent \textbf{Advantages of relative importance. \ }
Between \DANISH and \CAP-Decima, the carbon-agnostic scheduler is identical -- thus, performance differences can be attributed to the key ideas behind \DANISH, namely relative importance (see \autoref{sec:danish-design}).  In what follows, we examine this in detail using the \verb|DE| grid region with batches of 50 jobs.  We configure \DANISH and \CAP-Decima with ten parameter settings for varying carbon-awareness.  \autoref{fig:capdecima-vs-danish} plots the result of this experiment, where each dot denotes the outcome of one trial.  We fit a cubic polynomial to the outcomes of both methods to illustrate the key trend: \DANISH exhibits a strictly better trade-off between carbon footprint and ECT.  For trials where either method achieves carbon savings between 35\% and 45\%, \DANISH increases ECT by an average of 7.9\%, while \CAP-Decima increases it by an average of 42.7\%.  Conversely, for those trials where either method increases ECT by between 0\% and 10\%, \DANISH achieves average carbon savings of 35.6\%, while \CAP-Decima achieves an average savings of only 20.1\%.

\begin{figure}[t]
    \centering 
    \includegraphics[width=0.8\linewidth]{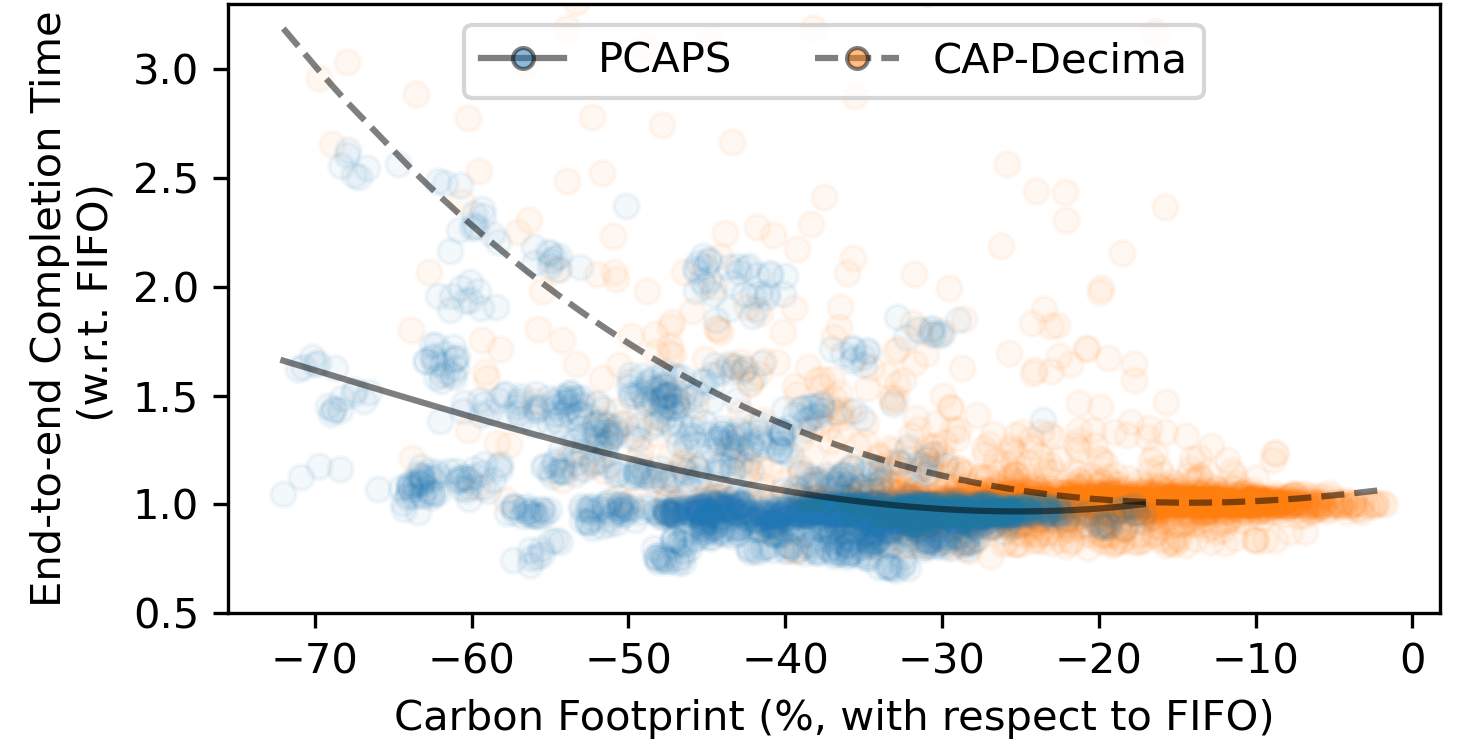} \vspace{-1em}
    \caption{ Relative carbon footprint %
    vs. end-to-end completion time %
    for \DANISH and \CAP-Decima in simulator experiments, given varying parameters $\gamma \in [0.1, 1.0]$ and $B \in \{ 5, 10, \dots , 85\}$ that correspond to degrees of carbon-awareness.  Each dot represents an individual trial, and lines represent a cubic polynomial of best fit. } %
    \label{fig:capdecima-vs-danish} \vspace{-0.6cm}
\end{figure}

\noindent \textbf{Effects of carbon intensity trace characteristics. \ } The top-line results in \autoref{tab:results-simulator} average over all six grid regions.  
We configure both schedulers to be moderately carbon-aware and characterize each carbon trace based on its \textit{coefficient of variation}. In \autoref{fig:traces-sim}, we plot the carbon reduction and ECT for each of \CAP-FIFO, \DANISH, and Decima.  
Decima's carbon reduction relative to the 
default is higher relative to that observed in the prototype -- this is due to differences between Spark's standalone FIFO scheduler and the Spark/Kubernetes behavior of our prototype (see \sref{Appendix}{apx:fifo-vs-k8s}).  Similarly, \DANISH's ``baseline'' carbon reductions increase alongside Decima's, and \CAP gains relative ground compared to \CAP-FIFO in the simulator.
We observe similar trends overall -- grid regions with more intermittent and variable energy mixes due to renewables drive increases in both carbon reduction and ECT, as observed in \autoref{fig:traces-proto}.

\begin{figure}[t]
    \centering 
    \vspace{-0.5em}
    \includegraphics[width=0.9\linewidth]{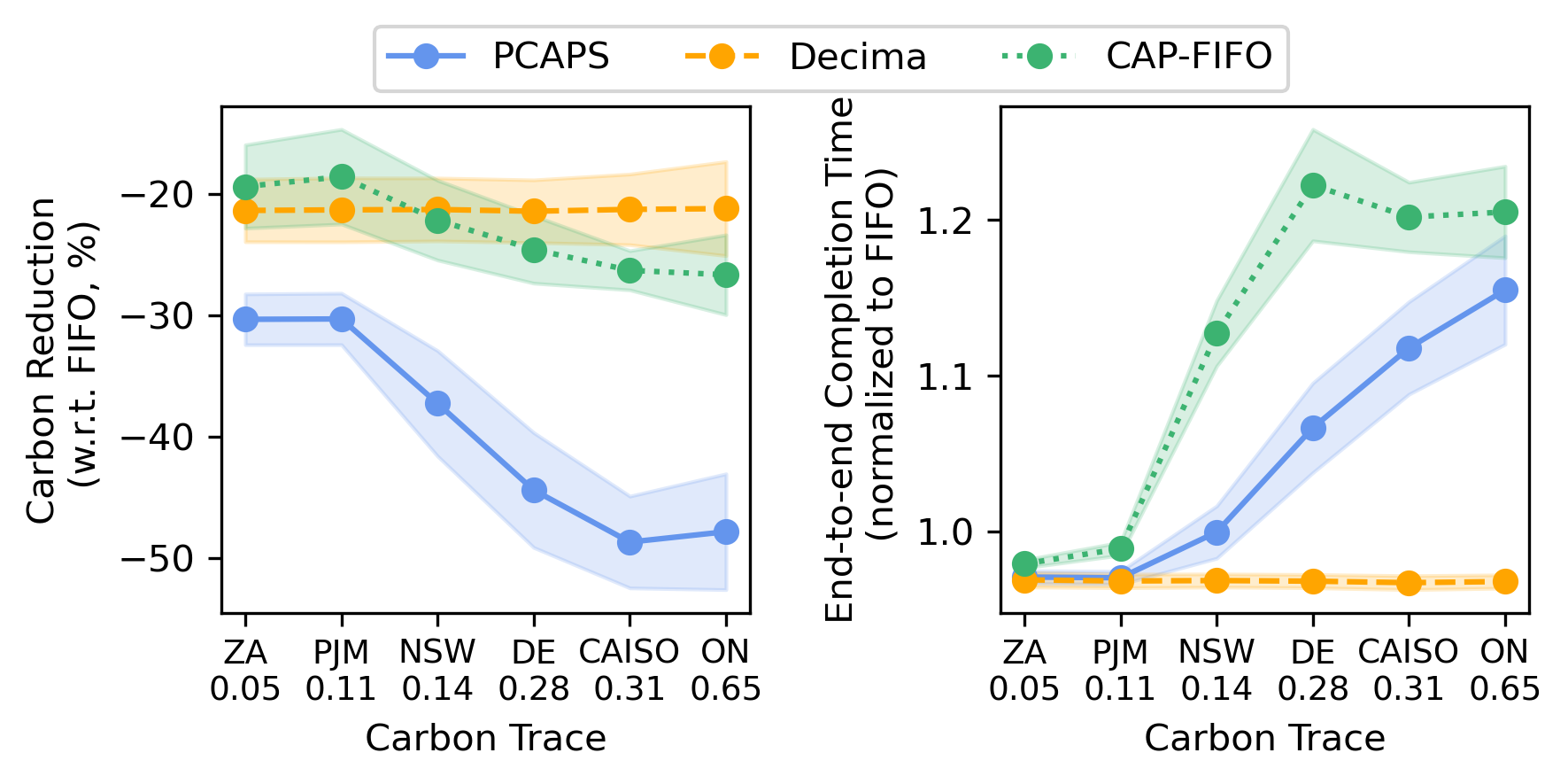} \vspace{-0.4cm}
    \caption{ Carbon reduction \textit{(left)} and increase in end-to-end completion time \textit{(right)} for \CAP, \DANISH, and Decima (relative to FIFO) in six grid regions.  Shaded areas denote the standard deviation across a carbon trace. }
    \label{fig:traces-sim} 
    \vspace{-0.4cm}
\end{figure}

\subsection{Takeaways}
Through evaluation in a realistic simulator and a prototype cluster, we show that a moderately carbon-aware \DANISH reduces carbon emissions by up to 39.7\% in exchange for modest increases ($< 10\%$) in ECT for batches of 25, 50, and 100 data processing jobs.  
\CAP configured to be moderately carbon-aware ($B=20$) is also effective, reducing carbon by up to 25.1\% with respect to the scheduler it is implemented on, in exchange for slightly larger increases in ECT.  While \CAP does not take dependencies into account and is suboptimal in terms of the carbon-time trade-off (see \autoref{fig:capdecima-vs-danish}), it is easier to implement and more general than \DANISH.
Intuitively, the performance of all carbon-aware techniques exhibits a dependence on the time-varying behavior of the power grid.
Greater carbon savings can be achieved in regions with more renewables (thus more variability) in carbon intensity.

\section{Conclusion} \label{sec:conclusion}

\DANISH demonstrates that an augmented scheduler that directly takes both the carbon cost of computation and precedence constraints (e.g., in the DAG of a data processing job) into consideration can achieve a favorable trade-off between carbon savings and completion time.  Through experiments, we show that \DANISH's configurability enables a scheduling policy that meaningfully reduces the carbon footprint of data processing without prohibitive increases in completion time.  Furthermore, our detailed analytical and experimental study of \CAP provides another avenue towards configurable carbon-aware scheduling that is broadly applicable and easy-to-implement.

\balance
\bibliographystyle{plainnat}
\bibliography{paper.bib}

\clearpage
\appendix
\onecolumn
\section*{Appendix}

\section{Evaluation Supplements}

In this section, we give additional experiment setup details and results that are deferred from the main body (i.e., in \autoref{sec:eval})

\subsection{Deferred setup details} \label{apx:extra-details}

\subsubsection{\textbf{GreenHadoop~\cite{Goiri:2012:GreenHadoop} adaptation and implementation}} \label{apx:greenhadoop}

This implementation begins by considering green (renewable) energy and brown (non-renewable) energy values available in the carbon traces obtained from Electricity Maps~\cite{electricity-map}. The system derives a "green window" by iterating over future minutes and summing the portion of executor capacity that can be powered purely by renewable energy sources until it meets the outstanding workload. Similarly, a "brown window" is derived by computing the number of executor minutes to finish the outstanding workload, assuming the system is run at maximum executor capacity. These two windows bracket the best-case scenarios of optimizing purely for carbon and for time. 

To balance carbon usage and JCT, the system derives a final window size via a convex combination of the green and brown windows, modulated by a tunable parameter, $\theta$, which details the carbon-awareness of the system. By default, $\theta$ is set to $0.5$, which balances between $0$ (carbon-agnostic) and $1$ (fully-carbon-aware) At each scheduling decision, the system utilizes all of the available green energy in the future timesteps, and computes the fraction of available brown energy required to complete all jobs by the end of our convex window. Within that determined executor limit, tasks are dispatched using a standard FIFO queue, similar to how inter-job dependencies are managed in Hadoop. While this approach separates the carbon-aware resource provisioning from the job ordering logic, it does not embed the carbon-importance of tasks within each job.

\subsubsection{\textbf{Differences between Spark standalone FIFO baseline and default Spark / Kubernetes behavior}} \label{apx:fifo-vs-k8s}

In the main body, in \autoref{sec:eval-proto}, we discuss some notable  differences between the \textit{baselines} in the prototype experiments and the simulator experiments.
In the simulator experiments, we use the default Spark standalone cluster scheduler (i.e., FIFO) as a baseline.  In contrast, as the prototype is implemented on top of Spark and Kubernetes, we use the combined default behavior of the Spark DAG scheduler and Kubernetes scheduler as a baseline.

The difference between these baselines is particularly pronounced when comparing e.g., the relative carbon footprint of Decima against the primary baseline -- in the simulator, Decima's carbon footprint is a 21.5\% improvement over FIFO, while in the prototype Decima improves on the default behavior by only 1.2\% -- similar trends are evident in average job completion time.  This difference in performance can be attributed to a difference in how per-job parallelism limits are set by the FIFO scheduler in the simulated Spark standalone environment vs. the Spark/Kubernetes behavior configured on our prototype cluster.  In the Spark standalone mode of operation, the default FIFO behavior assigns up to $N$ executors to each stage of a job, where $N$ is the number of tasks within said stage.  In contrast, our Spark/Kubernetes cluster scheduler is configured to assign up to $25$ executors across all stages of any job to avoid an issue where executors from previously completed stages continue to use cluster resources, causing Spark to hang.  

\begin{figure*}[h]
    \vspace{-1em}
    \centering
    \includegraphics[width=0.9\linewidth]{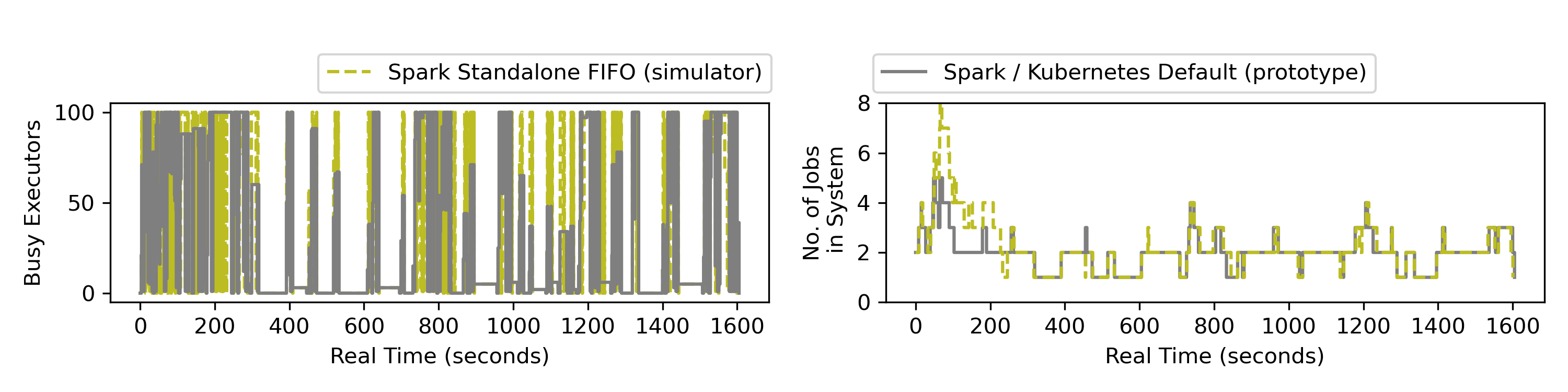} \vspace{-1em}
    \caption{ Executor usage (left) and number of jobs in the system (right) for an identical group of 50 TPC-H jobs in both the simulator (FIFO) and prototype (Spark / Kubernetes default) clusters, both with a maximum of 100 executors. }
    \label{fig:fifo-vs-k8s-timeseries}
\end{figure*}

In \autoref{fig:fifo-vs-k8s-timeseries}, we illustrate the granular effect of these different policies for an experiment of 50 TPC-H jobs with identical ordering and identical interarrival times in the simulator and prototype, respectively.  The left-hand side plots the number of busy executors, while the right-hand side plots the number of jobs in the system.  Note that the number of busy executors in the Spark / Kubernetes prototype more frequently drops below 100, particularly when the number of jobs in the system is low (e.g., between 1000 - 1200 seconds).  This corresponds to the executor cap configured for the Spark / Kubernetes default behavior, which results in more efficient executor usage and reduced blocking (i.e., when a job enters the system, it is more likely that there will be free executors to work on it).  As a result, the Spark / Kubernetes behavior improves on the Spark standalone FIFO scheduler by 18.8\% in carbon footprint and 22.1\% in average job completion time in an experiment with 50 TPC-H jobs, mirroring the broader trends observed in our simulator and prototype experiments.

\subsection{Deferred experiments} \label{apx:exp}

In this section, we present the results of additional experiments in the simulator and the prototype that are omitted from the main body due to space considerations.

\subsubsection{\textbf{Impact of total number of jobs}}

In the main results presented in \autoref{sec:eval}, we evaluate the performance of tested algorithms under experiments with 25, 50, and 100 continuously arriving jobs.  We explore the impact of varying this number of jobs in each experiment below, with \DANISH and \CAP configured to be moderately carbon-aware, in the \verb|DE| grid region.

\autoref{fig:num-jobs-sim} plots the average carbon reduction, end-to-end completion time, and average per-job completion time with respect to FIFO for \DANISH, Decima, and \CAP on top of FIFO in the simulator environment, using experiments with 12, 25, 50, 100, and 200 jobs.  We find that the relative ordering of all three techniques largely stays constant, although measuring results for a small number of jobs (e.g., 12, 25) is intuitively prone to more variance as the end-to-end results are more easily impacted by e.g., one or two random jobs.  Out of the three metrics, carbon footprint is the most stable.  As the number of jobs increases, all metrics ``converge'' to some average behavior.  One interesting exception is \CAP-FIFO -- due to the ``blocking'' behavior of the FIFO scheduler in the simulator that is more prone to queue build-up (e.g., see \autoref{apx:fifo-vs-k8s}), a larger total number of jobs results in longer job completion times for \CAP-FIFO.

\begin{figure*}[h]
\begin{minipage}{0.32\linewidth}
        \centering
    \includegraphics[width=\linewidth]{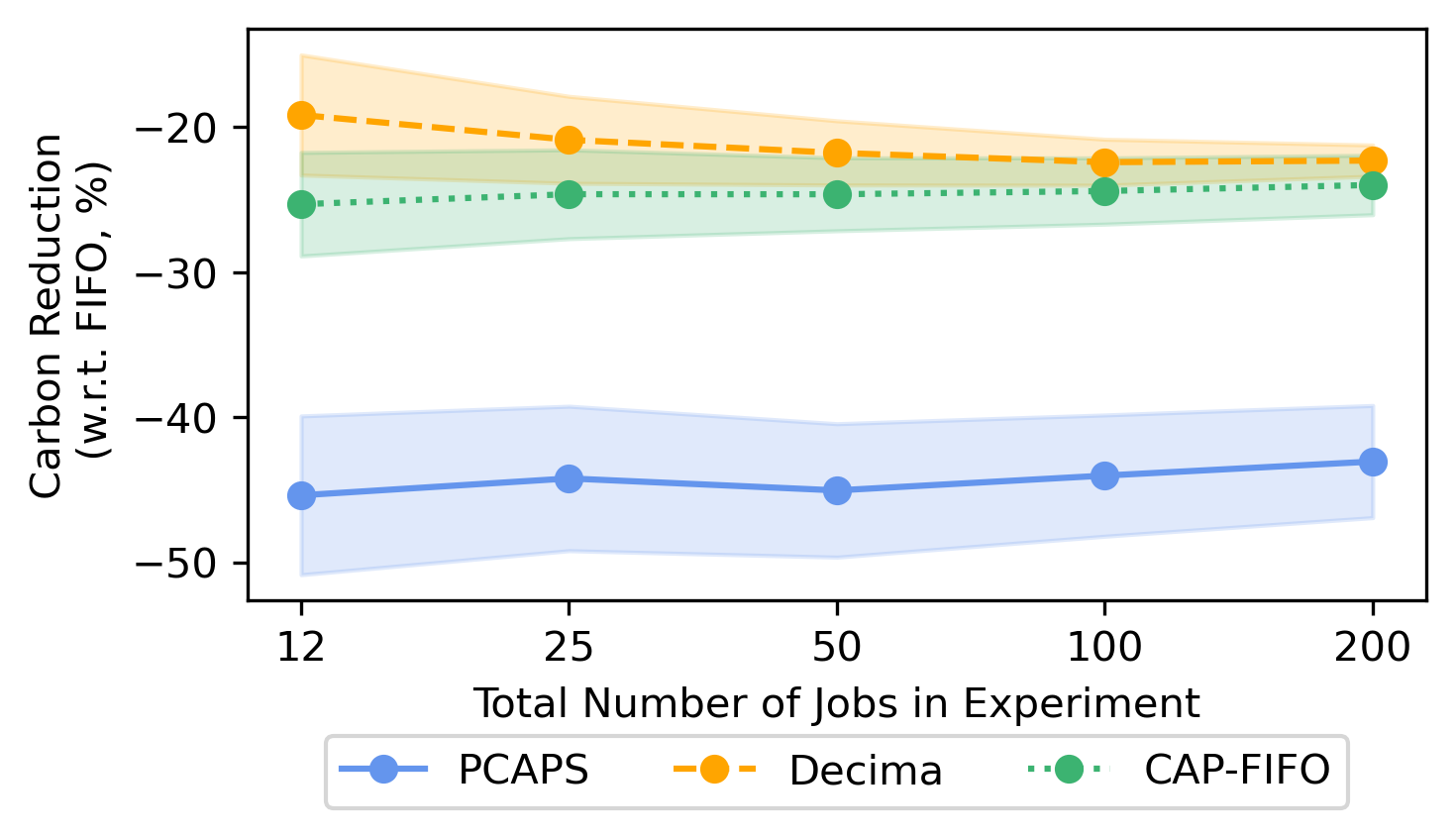} \vspace{-1em}
    {\centering \textbf{\textit{(a)}}}
\end{minipage} \hfill
\begin{minipage}{0.32\linewidth}
        \centering
    \includegraphics[width=\linewidth]{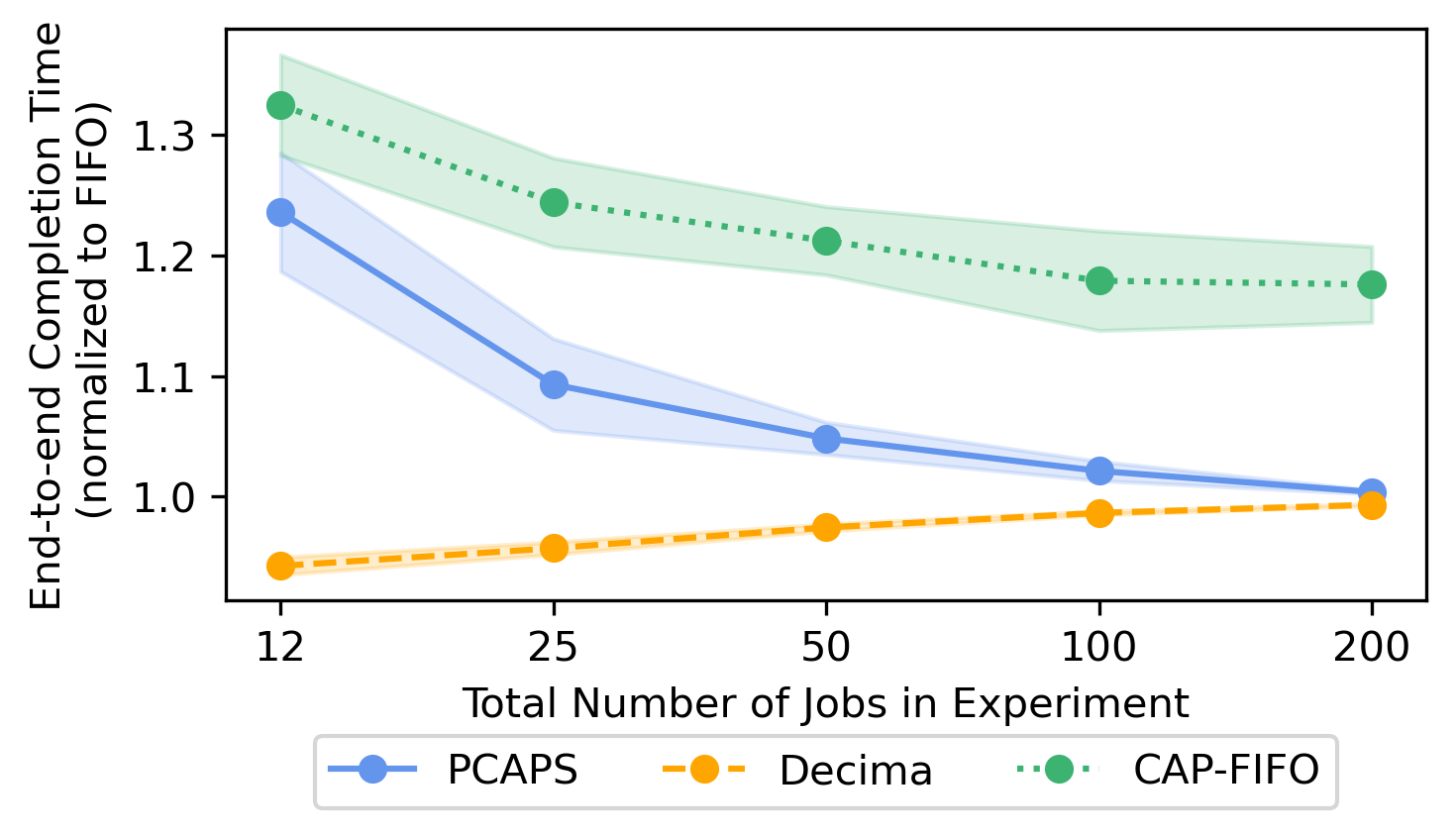} \vspace{-1em}
    {\centering \textbf{\textit{(b)}}}
\end{minipage} \hfill
\begin{minipage}{0.32\linewidth}
        \centering
    \includegraphics[width=\linewidth]{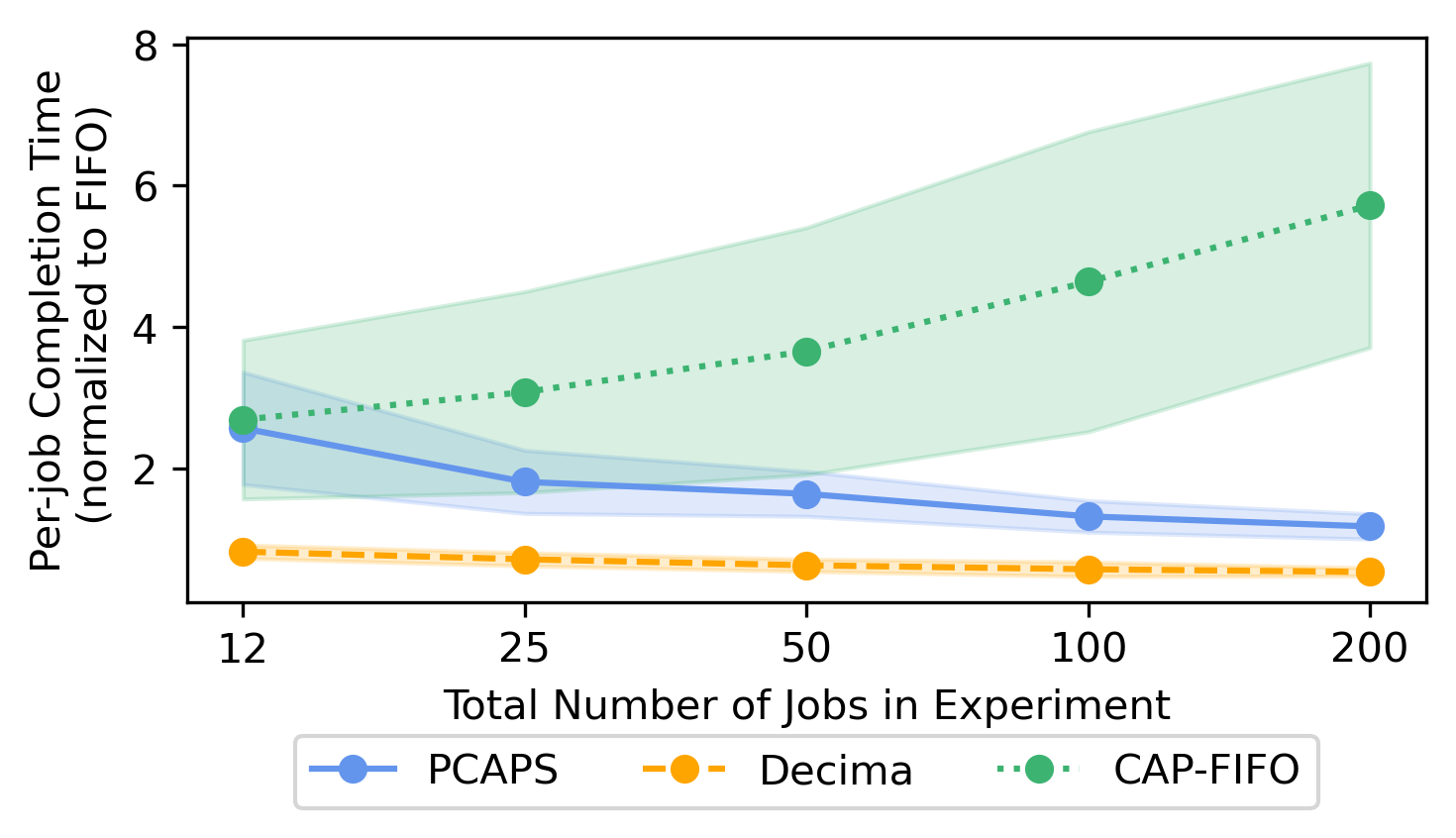} \vspace{-1em}
    {\centering \textbf{\textit{(c)}}}
\end{minipage}
\caption{  \textbf{\textit{(a)}} Carbon reduction, \textbf{\textit{(b)}} end-to-end completion time, and \textbf{\textit{(c)}} average job completion time achieved by \DANISH, \CAP-FIFO, and Decima (relative to FIFO) in a single grid region for varying experiment sizes.  Shaded regions denote the standard deviation across the entire carbon trace.  } \label{fig:num-jobs-sim}
\end{figure*}
\begin{figure*}[h]
\begin{minipage}{0.32\linewidth}
        \centering
    \includegraphics[width=\linewidth]{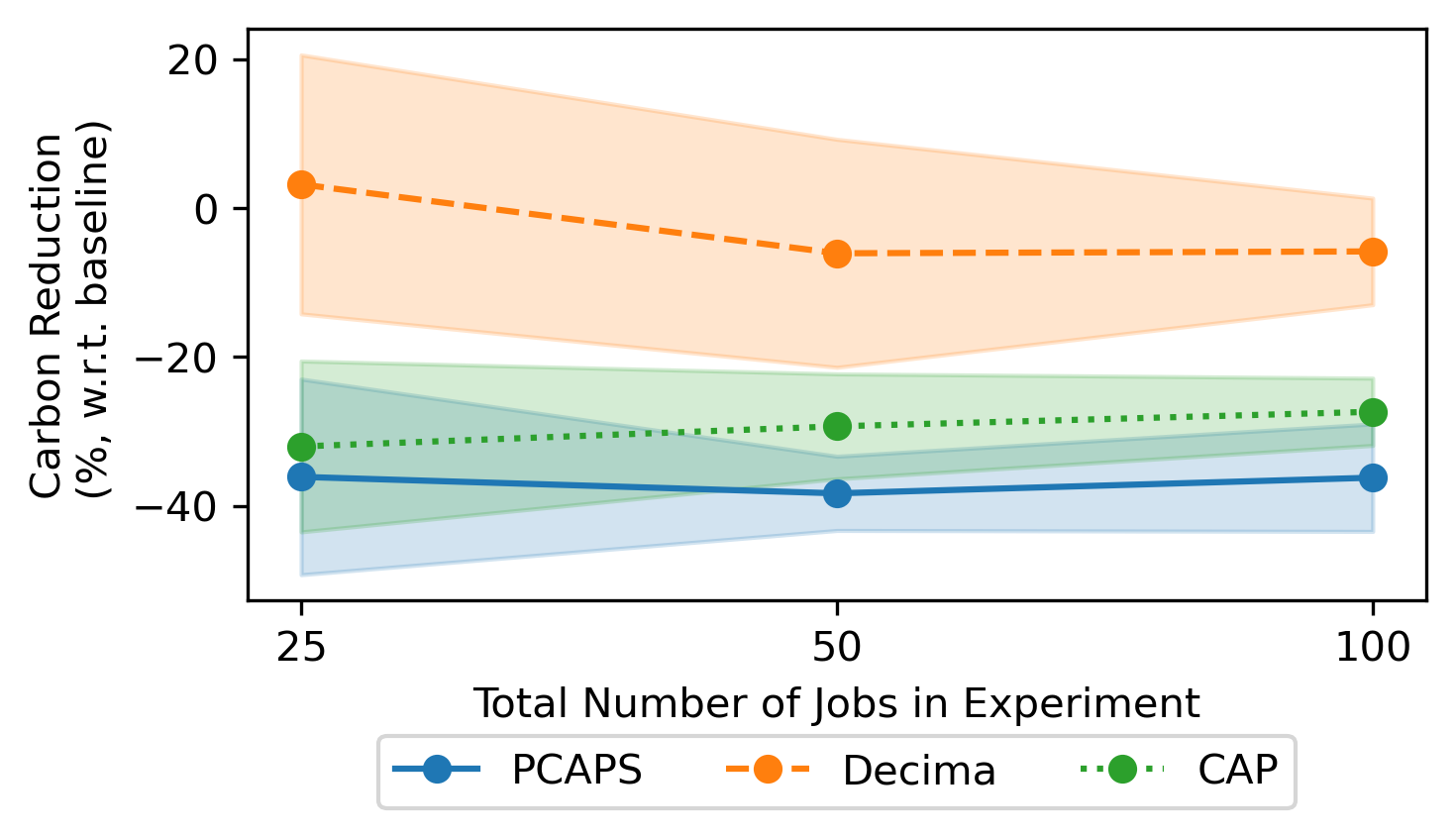} \vspace{-1em}
    {\centering \textbf{\textit{(a)}}}
\end{minipage} \hfill
\begin{minipage}{0.32\linewidth}
        \centering
    \includegraphics[width=\linewidth]{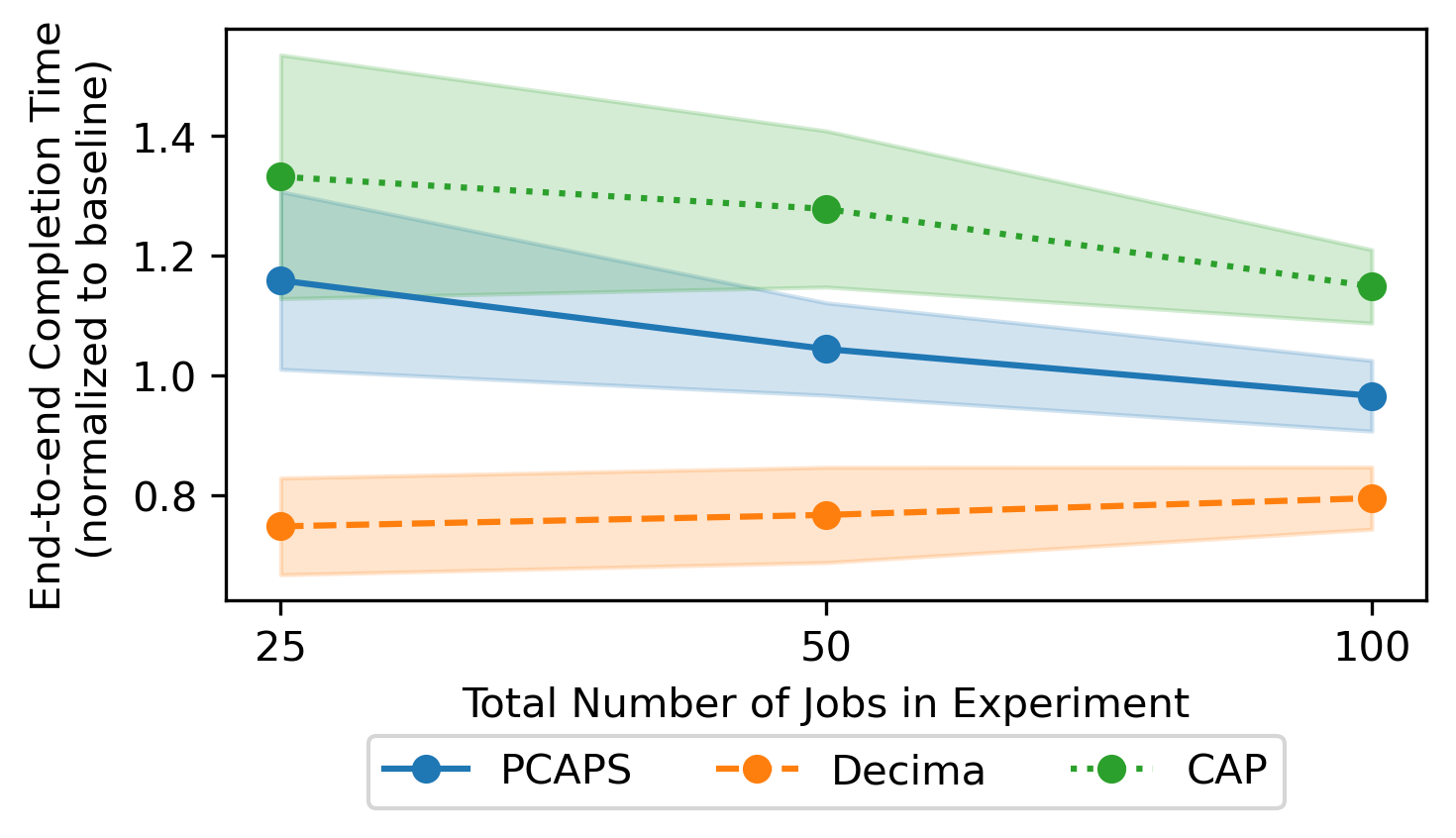} \vspace{-1em}
    {\centering \textbf{\textit{(b)}}}
\end{minipage} \hfill
\begin{minipage}{0.32\linewidth}
        \centering
    \includegraphics[width=\linewidth]{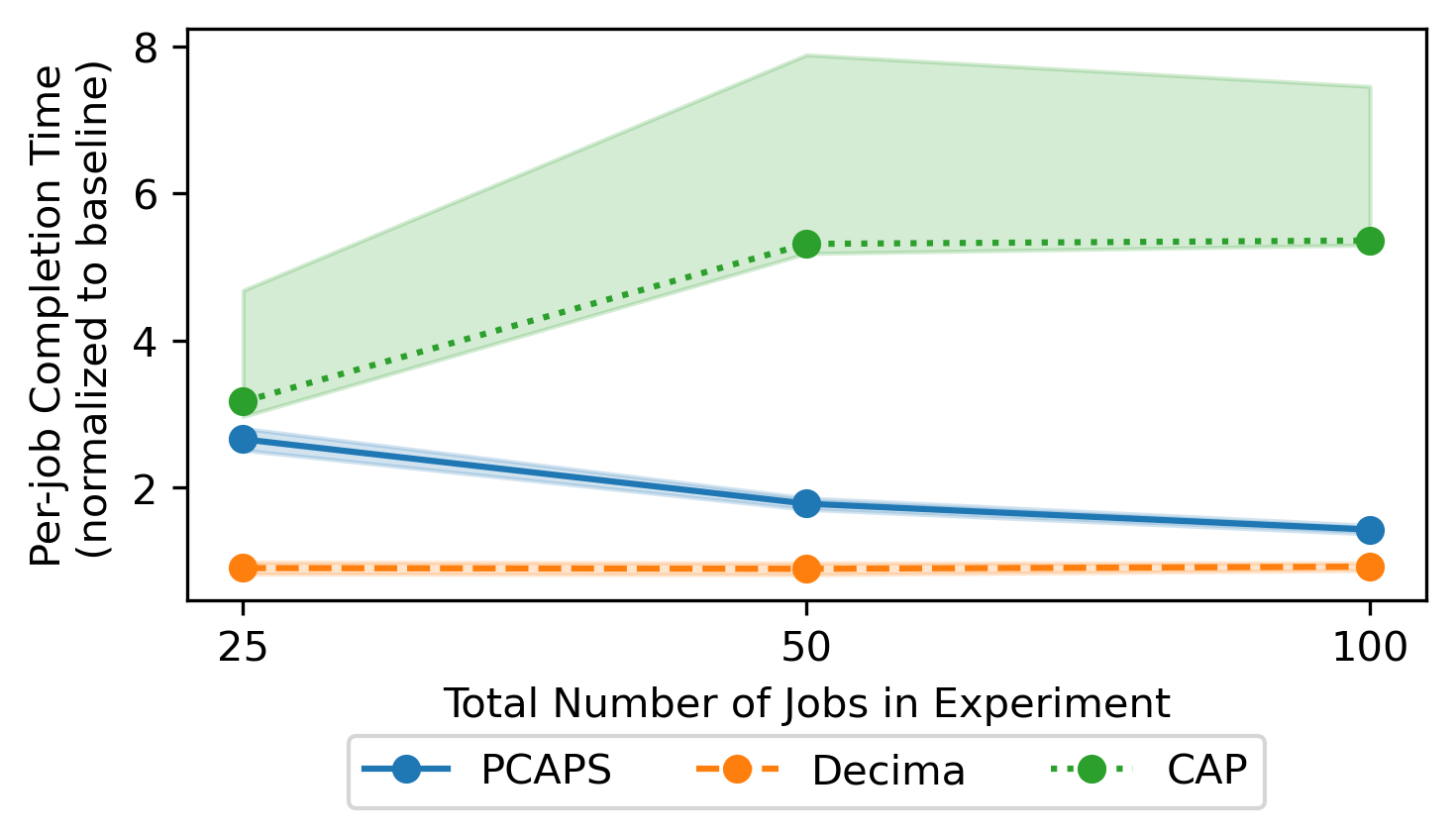} \vspace{-1em}
    {\centering \textbf{\textit{(c)}}}
\end{minipage} 
\caption{ \textbf{\textit{(a)}} Carbon reduction, \textbf{\textit{(b)}} end-to-end completion time, and \textbf{\textit{(c)}} average job completion time achieved by \DANISH, \CAP, and Decima (relative to the Spark/Kubernetes default) in a single grid region for varying experiment sizes.  Shaded regions denote the standard deviation across 10 random trials. } \label{fig:num-jobs-proto}
\end{figure*}

In \autoref{fig:num-jobs-proto}, we plot the same metrics with respect to the default Spark/Kubernetes scheduler in the prototype implementation for \DANISH, Decima, and \CAP, using experiments with 25, 50, and 100 jobs.  These results generally mirror those shown in the simulator above, although \CAP implemented on top of the default Spark/Kubernetes behavior does not exhibit the same increase in per-job completion time for larger experiment sizes that \CAP-FIFO does in the simulator -- this is because the blocking behavior of FIFO is more pronounced than in the default Spark/Kubernetes scheduler (e.g., see \autoref{apx:fifo-vs-k8s}).

\subsubsection{\textbf{Impacts of submission rate}}
In the main results presented in \autoref{sec:eval}, we evaluate the performance of tested algorithms under continuous job arrivals following a Poisson process with an average interarrival time of $\nicefrac{1}{\lambda} = 30$ minutes ($30$ seconds in real time).  In what follows, we explore the impact of varying this interarrival time below, where note that smaller interarrival times imply that the cluster is more heavily utilized.  \DANISH and \CAP algorithms are both configured to be moderately carbon-aware, and the tested grid region is \verb|DE|.

In \autoref{fig:lambda-sim}, we plot the average carbon reduction, end-to-end completion time, and average per-job completion time with respect to FIFO for \DANISH, Decima, and \CAP-FIFO in the simulator environment.  We find that the relative ordering of algorithm performance remain largely the same, but differences arise particularly for small interarrival times -- in these scenarios where the cluster is more heavily utilized, we find that both \DANISH and Decima benefit from more intelligent scheduling decisions with respect to FIFO, which manifests in higher carbon reductions relative to FIFO and lower end-to-end completion times.

\begin{figure*}[h]
\begin{minipage}{0.32\linewidth}
        \centering
    \includegraphics[width=\linewidth]{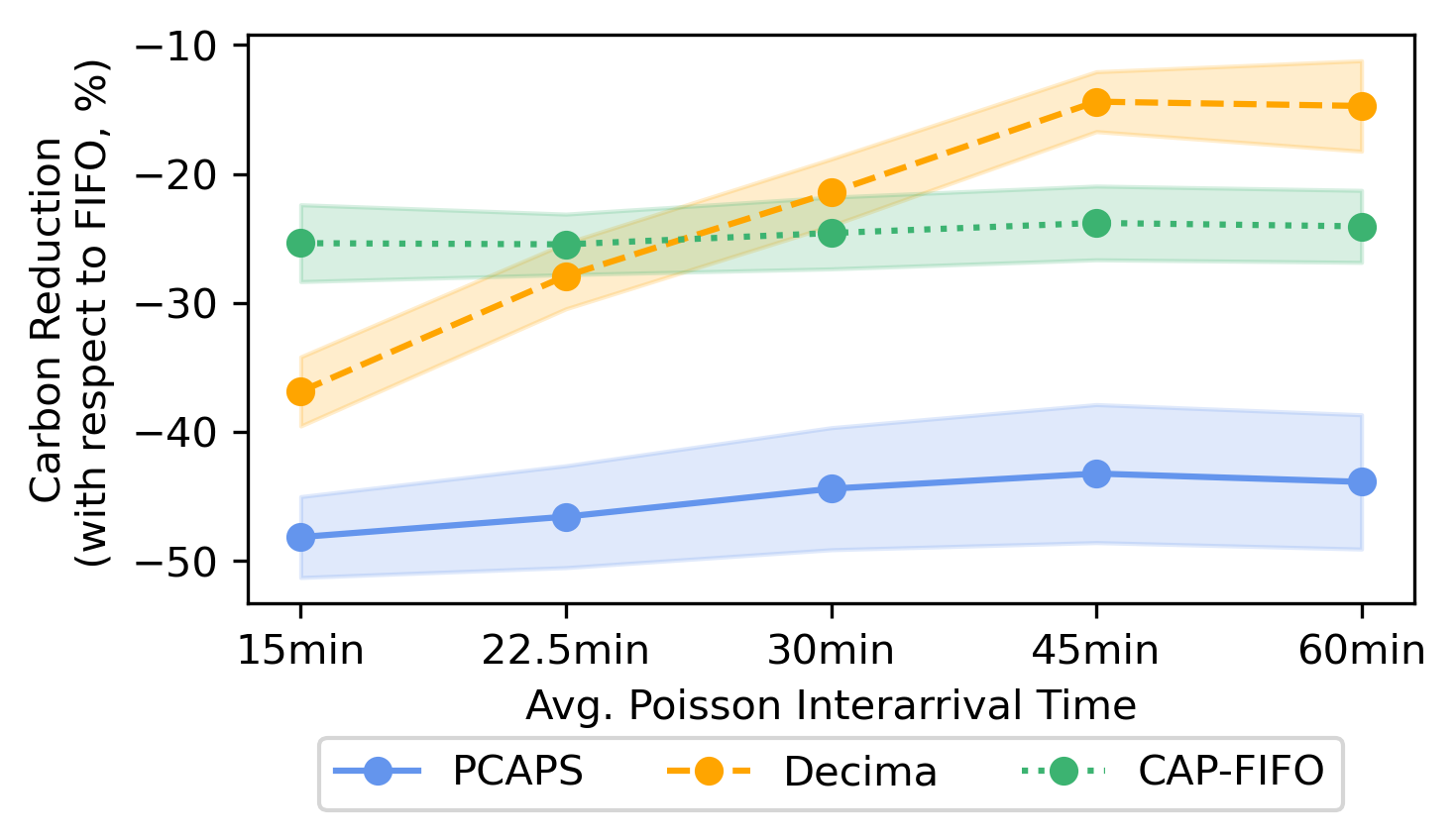} \vspace{-1em}
    {\centering \textbf{\textit{(a)}}}
\end{minipage} \hfill
\begin{minipage}{0.32\linewidth}
        \centering
    \includegraphics[width=\linewidth]{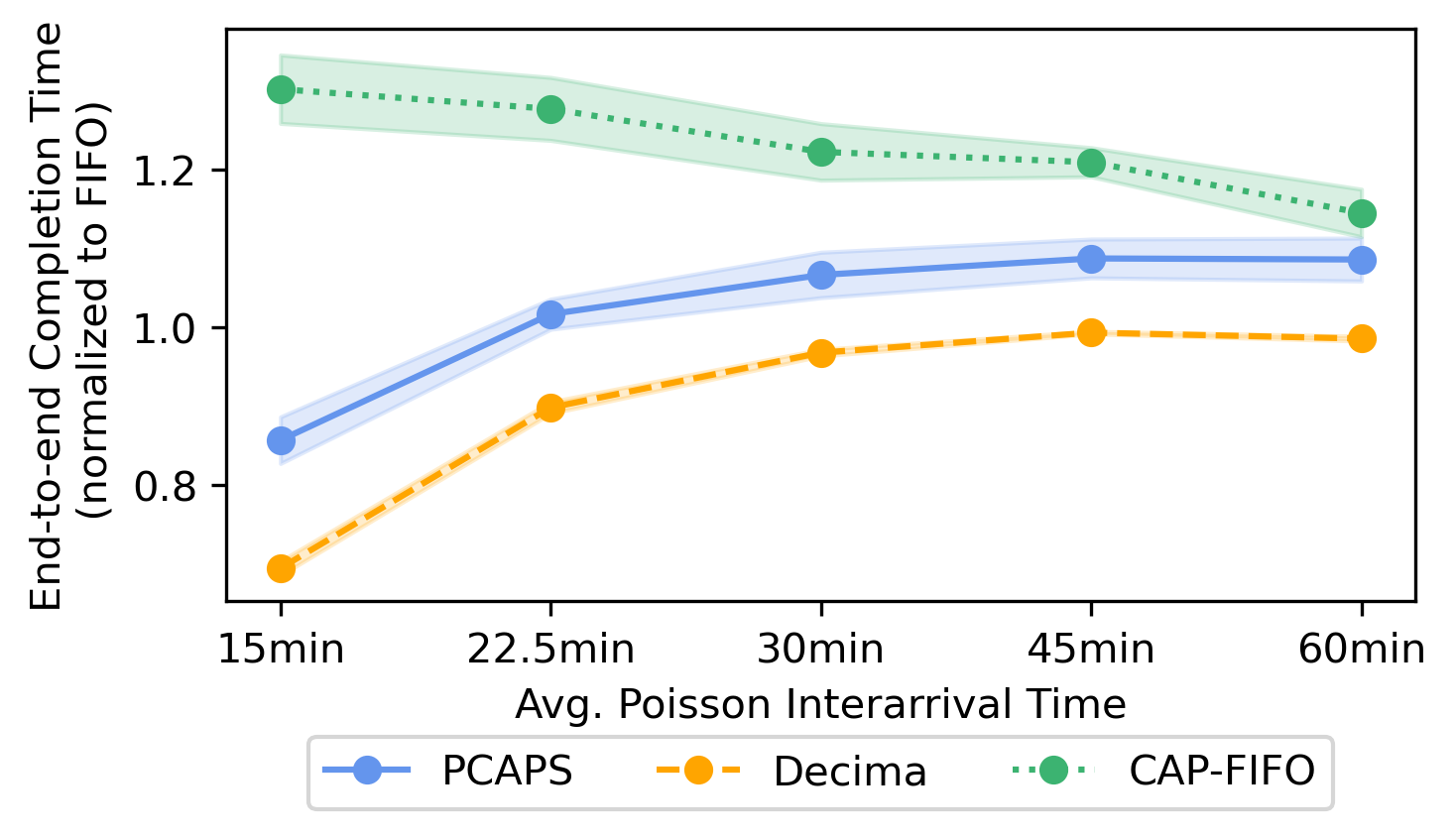} \vspace{-1em}
    {\centering \textbf{\textit{(b)}}}
\end{minipage}\hfill
\begin{minipage}{0.32\linewidth}
        \centering
    \includegraphics[width=\linewidth]{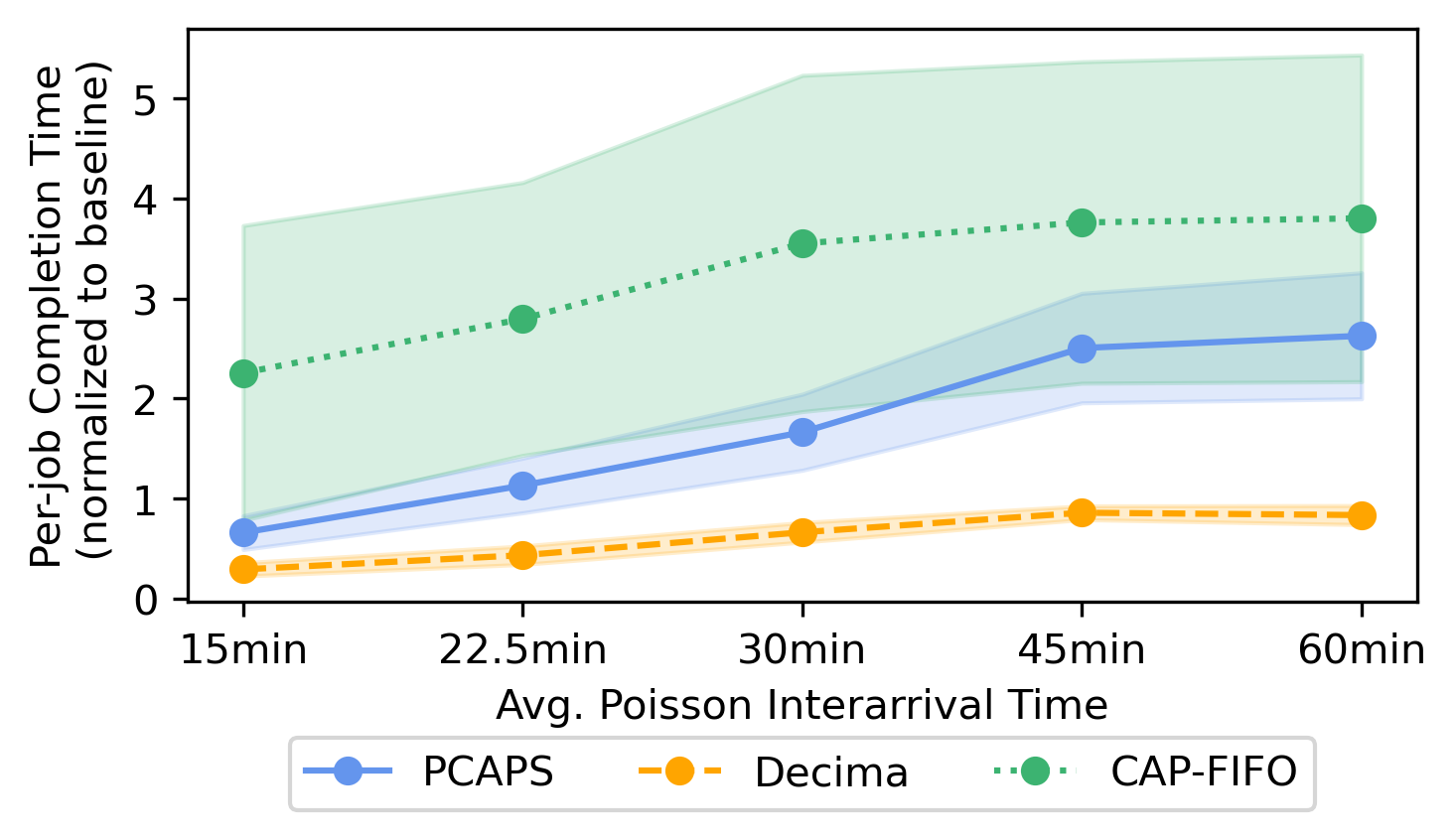} \vspace{-1em}
    {\centering \textbf{\textit{(c)}}}
\end{minipage}
\caption{ \textbf{\textit{(a)}} Carbon reduction, \textbf{\textit{(b)}} end-to-end completion time, and \textbf{\textit{(c)}} average job completion time achieved by \DANISH, \CAP-FIFO, and Decima (relative to FIFO) in a single grid region for varying Poisson interarrival times.  Shaded regions denote the standard deviation across the entire carbon trace. } \label{fig:lambda-sim}
\end{figure*}

In \autoref{fig:lambda-proto}, we plot the same metrics with respect to the default Spark/Kubernetes scheduler in the prototype implementation for \DANISH, Decima, and \CAP.  These results largely mirror those shown in the simulator above -- the most notable change is the improved performance of \DANISH and Decima (in terms of both end-to-end completion time and carbon footprint) for small interarrival times, i.e., when the cluster is heavily utilized.

\begin{figure*}[h]
\begin{minipage}{0.32\linewidth}
        \centering
    \includegraphics[width=\linewidth]{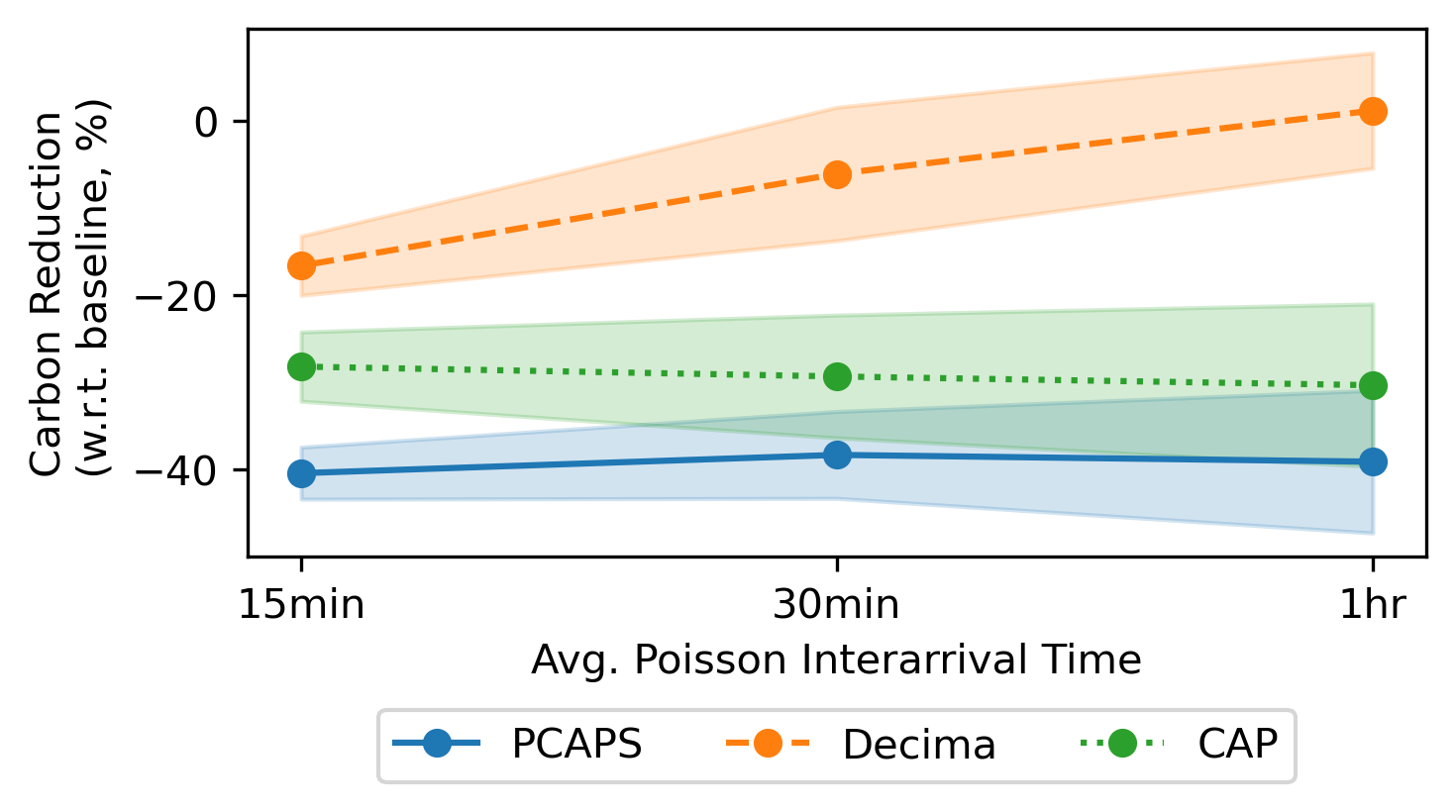} \vspace{-1em}
    {\centering \textbf{\textit{(a)}}}
\end{minipage} \hfill
\begin{minipage}{0.32\linewidth}
        \centering
    \includegraphics[width=\linewidth]{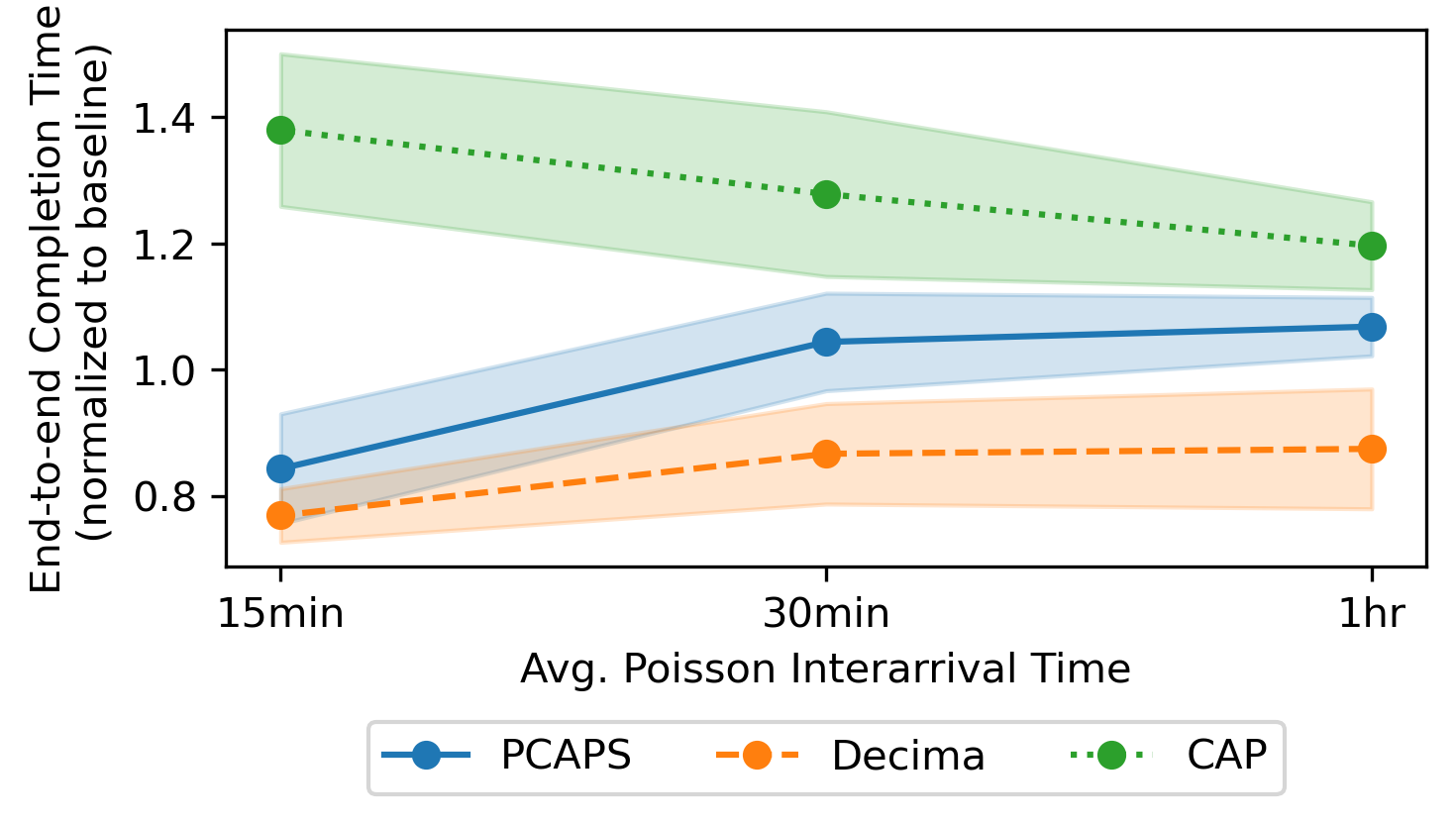} \vspace{-1em}
    {\centering \textbf{\textit{(b)}}}
\end{minipage}\hfill
\begin{minipage}{0.32\linewidth}
        \centering
    \includegraphics[width=\linewidth]{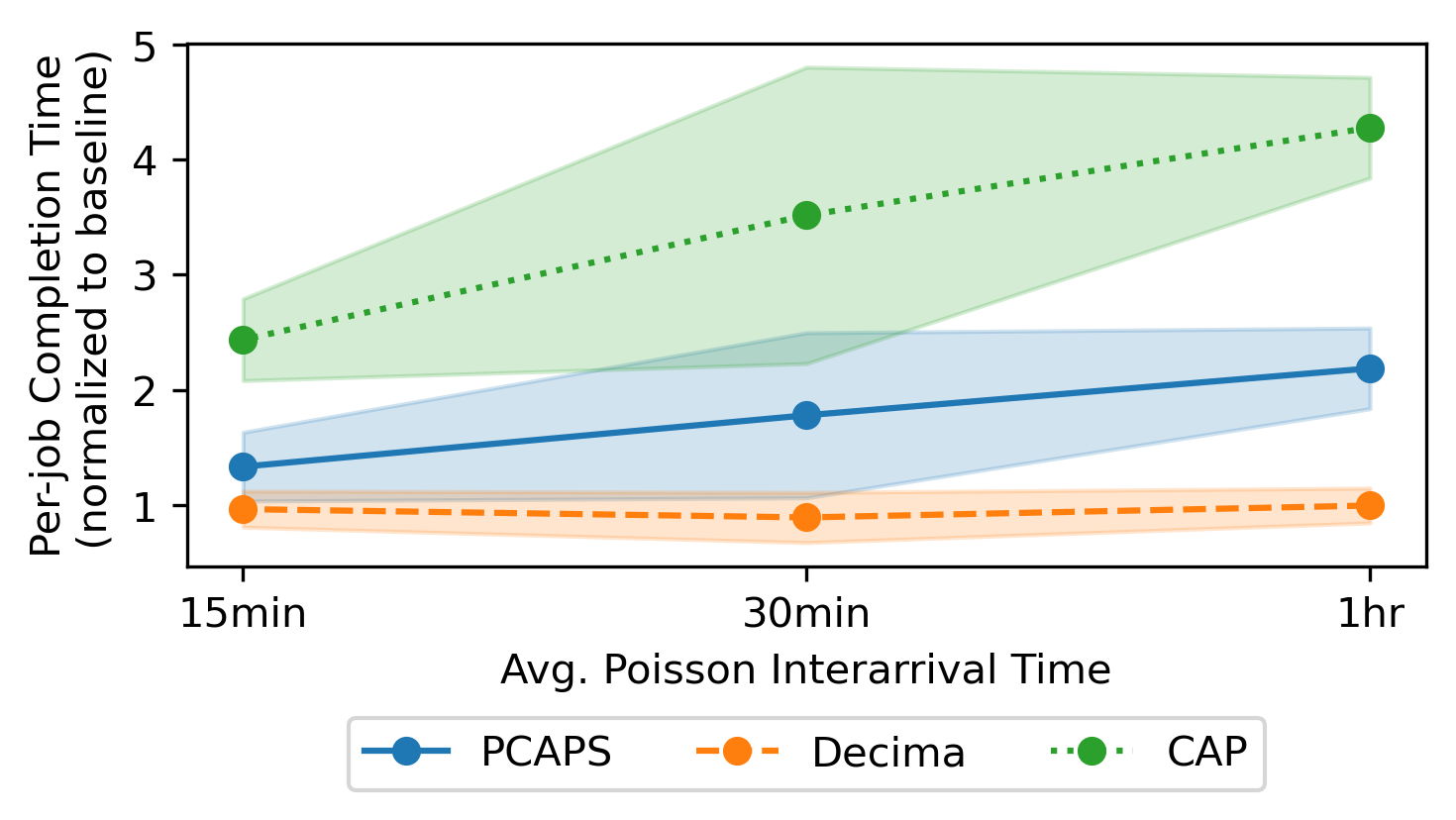} \vspace{-1em}
    {\centering \textbf{\textit{(c)}}}
\end{minipage}
\caption{ \textbf{\textit{(a)}} Carbon reduction, \textbf{\textit{(b)}} end-to-end completion time, and \textbf{\textit{(c)}} average job completion time achieved by \DANISH, \CAP, and Decima (relative to the Spark/Kubernetes default) in a single grid region for varying Poisson interarrival times.  Shaded regions denote the standard deviation across 10 random trials. } \label{fig:lambda-proto}
\end{figure*}

\subsubsection{\textbf{Carbon-awareness logic latency}}
The logic of \DANISH and \CAP naturally introduce latency due to the overhead required to make carbon-aware decisions.  We quantify this \textit{scheduling overhead} for FIFO, Decima, \CAP-FIFO, and \DANISH in the simulator below, in  the \verb|DE| grid region over 1000 simulated trials.  Note that we  report the latency in real-time (i.e.,  without the scaling mentioned  in \autoref{sec:carbon-traces}).  These measurements do not include the latency of e.g., calls to an external carbon intensity API.

In \autoref{fig:latency}(a), we plot the average latency (in milliseconds) of invoking each scheduler once, given that there is exactly $\{1, 5, 10, 25, 50, 75, 100\}$ outstanding TPC-H jobs in the queue.  We find that simple decision rule policies (FIFO and \CAP-FIFO) exhibit consistently low latencies below 5 milliseconds, where \CAP adds an overhead of a few milliseconds.  In contrast, Decima and \PCAPS, which both use a GNN+RL model to run inference at each invocation, intuitively exhibit latency that grows as a function of the number of outstanding jobs.  Relative to Decima, \PCAPS adds an overhead of a few milliseconds at each invocation, and this overhead is constant (i.e., it does not grow with the number of jobs in the queue).

In \autoref{fig:latency}(b), we plot the average latency over a full experiment, where the initial number of jobs in the queue is one of $\{1, 5, 10, 25, 50, 75, 100\}$ -- note that this latency goes down over time as jobs are completed.  It is computed as a ratio between the total time spent in the scheduler and the number of scheduler invocations over the experiment length.  The findings in this metric are largely similar, with lower overall averages due to the averaging over a full experiment (as opposed to averaging over a single queue length).
Overall, these results imply that the latency impact of carbon-awareness is minimal -- in the context of long-running big-data workloads, the sub-20 millisecond latencies observed are likely to be insignificant compared to other overheads on the cluster.

\begin{figure*}[h]
\hfill 
\begin{minipage}{0.32\linewidth}
        \centering
    \includegraphics[width=\linewidth]{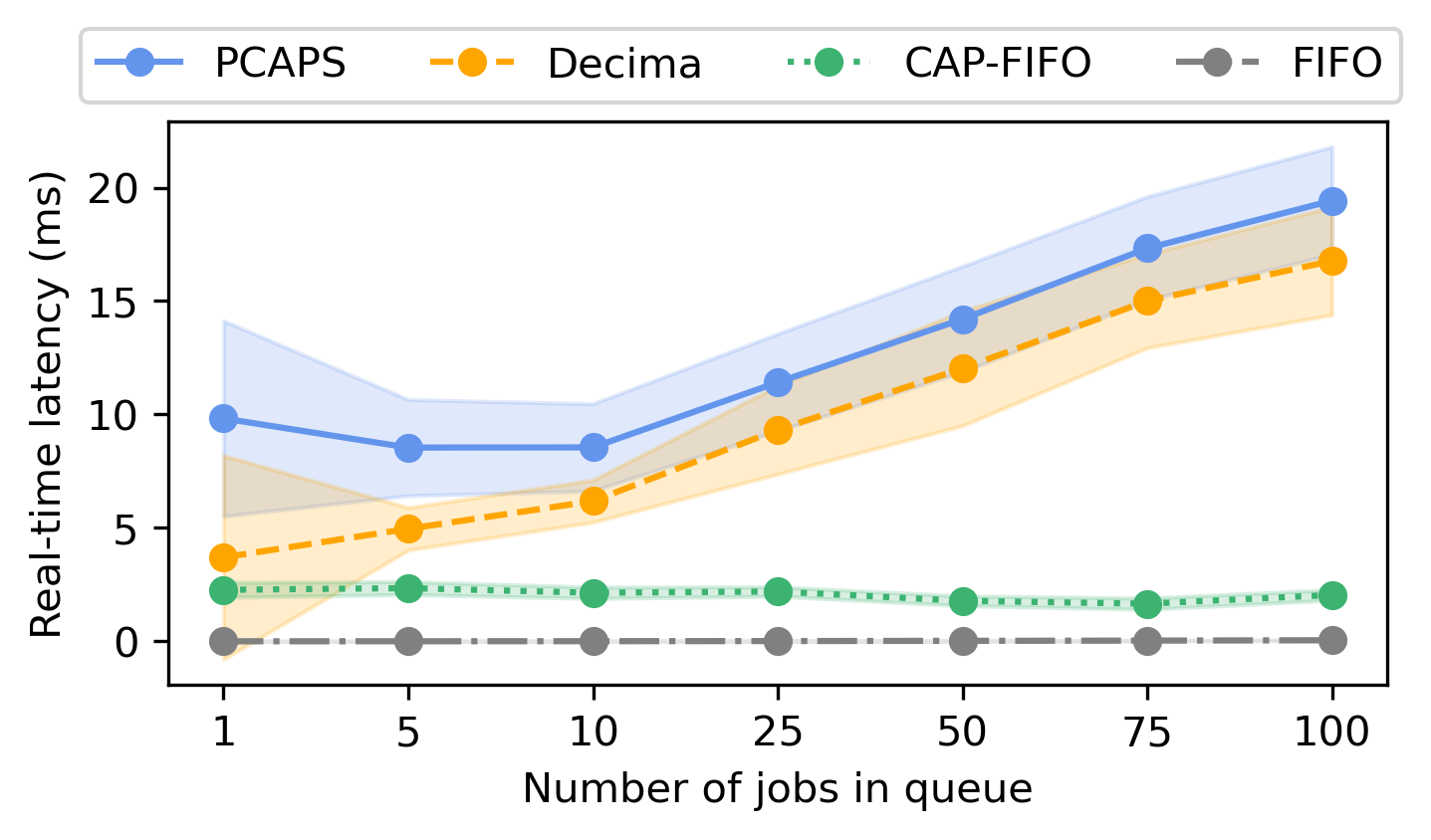} \vspace{-1em}
    {\centering \textbf{\textit{(a)}}}
\end{minipage} \hfill
\begin{minipage}{0.32\linewidth}
        \centering
    \includegraphics[width=\linewidth]{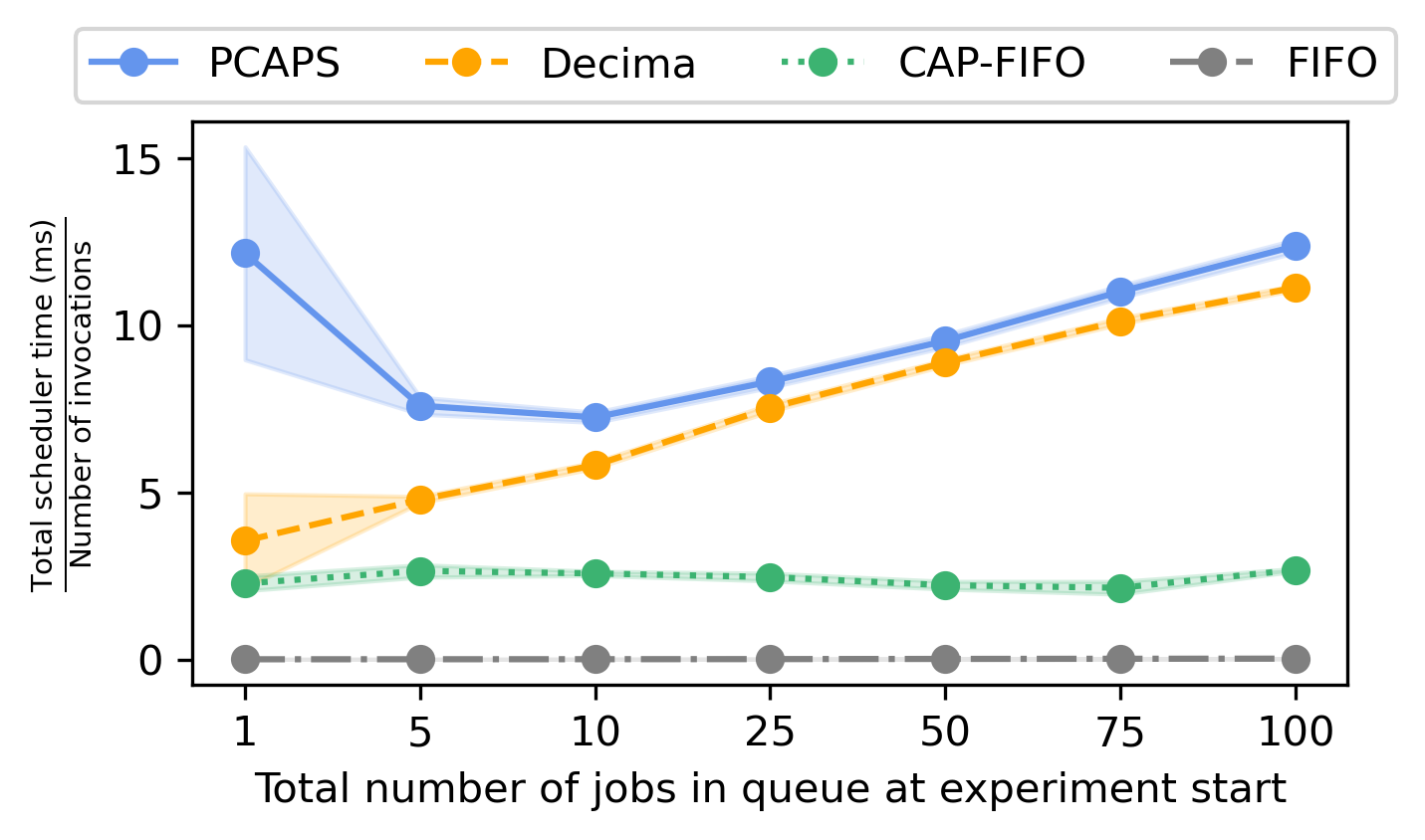} \vspace{-1em}
    {\centering \textbf{\textit{(b)}}}
\end{minipage} 
\hfill \hfill
\caption{  \textbf{\textit{(a)}} Average latency with $N$ jobs in the queue and \textbf{\textit{(b)}} average normalized time  in the scheduler for \PCAPS, \CAP-FIFO, Decima, and FIFO in a single grid region for varying experiment sizes.  Shaded region denotes the standard deviation across all 1000 trials. } \label{fig:latency}
\end{figure*}

\section{Deferred Analytical Results and Discussion}
In this section, we give full proofs for analytical results and detailed discussion for the \DANISH and \CAP scheduler designs introduced in \autoref{sec:design}.

\subsection{Deferred Proofs and Discussion from \autoref{sec:danish-design} (Precedence- and carbon-aware provisioning and scheduling)} \label{apx:danish-proofs}

In this section, we discuss and prove the analytical results for the \DANISH scheduler, introduced in \autoref{sec:danish-design}.  Throughout this section, we let $\texttt{PB}$ denote a \textit{carbon-agnostic, probabilistic} baseline scheduler, as outlined in \sref{Def.}{dfn:pb}.

For the sake of analysis, in the following results we leverage the classic makespan bound of Graham's list scheduling algorithm~\cite{Graham:66}, which is known to produce a schedule that is a $(2 - \nicefrac{1}{K})$-approximation for the optimal makespan on $K$ identical machines.  
Note that any carbon-agnostic probabilistic scheduler is an instance of list scheduling where the list (of tasks) is random, and the next task is assigned to a machine as soon as it becomes idle.

Recall the definition of \DANISH's carbon-awareness filter, parameterized by the $\Psi_\gamma$ function (see \autoref{alg:danish}). $\Psi_\gamma$ exhibits an exponential dependence on $r$, the relative importance of a task.  This is inspired by literature on one-way trading and related online problems~\cite{ElYaniv:01, Zhou:08}, where such an exponential trade-off is derived by balancing the marginal reward of the current price against the risk that better prices may arrive in the future.  
For DAG scheduling, we interpret relative importance in an analogous way: tasks with high importance have a substantial negative impact if they are not scheduled, so they are almost always scheduled at the current ``price'' (i.e., carbon intensity).  On the other hand, tasks with low relative importance (e.g., short tasks) may not significantly affect the DAG's completion time if they are deferred to start later, so they can ``wait'' for better carbon intensities.

To analyze the carbon stretch factor of \DANISH, we define $\mathcal{D}(\gamma, \mbf{c}) \in [0,1]$, a quantity that describes the fraction of tasks (in terms of total runtime) that are deferred by \DANISH when scheduling with a given value of $\gamma$ and given carbon trace $\mbf{c}$.  It is $\leq 1$ for any value of $\gamma$, and $\mathcal{D}(0, \mbf{c}) = 0$ for any $\mbf{c}$.  As $\gamma$ grows and \DANISH becomes ``more carbon-aware'',  $\mathcal{D}(\gamma, \mbf{c})$ grows in expectation.%

\subsubsection{\textbf{Proof of \autoref{thm:danishMakespan}}}\label{apx:danishMakespan}

We are ready to prove \autoref{thm:danishMakespan}, which states that for a given carbon intensity trace $\mbf{c}$, \DANISH's carbon stretch factor is $1 + \frac{\mathcal{D}(\gamma, \mbf{c}) K}{2 - \frac{1}{K}}$.

\begin{proof}
Let $\texttt{PB}_K(\mathcal{J})$ denote the schedule produced by the carbon-agnostic probabilistic baseline scheduler (e.g., Decima) using $K$ machines, and let $\DANISH_K(\mathcal{J})$ denote the schedule produced by $\DANISH$ for the same job $\mathcal{J}$ with $n$ tasks and a maximum of $K$ machines. We denote the carbon intensity trace by $\mbf{c} \coloneqq \{c(t) : t \ge 0\}$.

We henceforth use $\texttt{PB}_\text{K}^\texttt{D}$ to denote the instance of $\texttt{PB}$ that $\DANISH$ interfaces with.

Consider the perspective of a single node $v$: suppose that $\texttt{PB}$ samples node $v$ to be scheduled at time $t \geq 0$, where node $v$ has probability $p_{v,t} > 0$ and $v \in \mathcal{A}_t$.  By definition of $\texttt{PB}$, as soon as a task (node) is sampled, it runs on the available machine in $\texttt{PB}$'s schedule.

Now we consider the same node sampled by $\texttt{PB}_\text{K}^\texttt{D}$.  We denote $D_{v}$ as a random variable that denotes the number of times that node $v$ is \textit{not} scheduled (i.e., deferred) when it is sampled by $\texttt{PB}_\text{K}^\texttt{D}$. It is defined as:
\[
D_{v} = \sum_{t \in \mathcal{T}_v} \mathbb{I} \{ c(t) > \Psi_\gamma(r_{v,t}) \},
\]
where $\mathcal{T}_v$ denotes the times at which node $v$ is sampled by $\texttt{PB}_\text{K}^\texttt{D}$.
In the worst-case, observe that whenever a task $v$ is sampled but not scheduled, a deferral increases the total makespan by at most $X_v$, where $X_v$ is the runtime of task $v$.  Consider the edge case where all other tasks in $\mathcal{A}_t$ (i.e., all other tasks that are ready to run at the same time as task $v$) are scheduled on other machines at time $t$, and all further tasks (i.e., tasks that have not yet been completed but also were not in $\mathcal{A}_t$) are successors of $v$ (i.e., they cannot run until $v$ is completed).

In this case, assuming that the other tasks in $\mathcal{A}_t$ run to completion, there will be some time step $t' > t$ such that $\mathcal{A}_{t'} = \{ v \}$ -- i.e., the only task available to run is task $v$.  As soon as this is the case, $v$ will run -- this is because the \textit{relative importance} of any task in a set of size $1$ is always $1$.  Thus, in the worst-case, the schedule length increases by exactly $X_v$ -- this is the case if all of the other tasks in $\mathcal{A}_t$ finish at the same time $t'$ (i.e., no overlap with $v$).

This gives the following makespan bound in terms of $D_{v}$:
\[
\E [ \DANISH_K(\mathcal{J}) ] \le \E [ \texttt{PB}_K(\mathcal{J}) ] + \E \left[ \sum_{v \in \mathcal{J}} D_v X_v \right]
\]
By linearity of expectation, we have:
\[
\E [ \DANISH_K(\mathcal{J}) ] \le \E [ \texttt{PB}_K(\mathcal{J}) ] + \sum_{v \in \mathcal{J}} \E \left[ D_v \right] X_v
\]

Consider the \textit{total} number of deferrals $D = \sum_{v \in \mathcal{J}} D_v$, and note that $D$ must be $\leq n-1$ -- since at least one machine is active at all times, the maximum number of deferrals is that which causes the schedule to drop to a single machine across the entire time interval.  This immediately implies that $\E[D] \leq n-1$.
Define a sorted list $\{ X'_i : i \in n\}$ such that $X'_0 = \max_{v \in \mathcal{A}} X_v, \dots, X'_n = \min_{v \in \mathcal{A}} X_v$, and note that we can upper bound the above as follows:
\begin{align*}
\E [ \texttt{PB}_K(\mathcal{J}) ] + \sum_{v \in \mathcal{J}} \E \left[ D_v \right] X_v &\le \E [ \texttt{PB}_K(\mathcal{J}) ] + \sum_{i=0}^{\E[D]} X'_i,
\end{align*}
Note that this is a worst-case assumption -- in reality, the tasks with low relative importance (that are likely to be deferred) are unlikely to be the longest tasks for any reasonable scheduler \texttt{PB}.
Note that $\sum_{i=0}^n X'_i = \sum_{v \in \mathcal{J}} = \OPT_1(\mathcal{J})$, i.e., the optimal makespan using a single machine.  To simplify the above bound, we can define a function $\mathcal{D}(\gamma, \mbf{c})$ as follows:
\[
\mathcal{D}(\gamma, \mbf{c}) = \frac{\sum_{i=0}^{\E[D]} X'_i}{\OPT_1(\mathcal{J})}.
\]
Note that $\mathcal{D}(\gamma, \mbf{c}) \le 1$ for any $\gamma$ and any $\mbf{c}$, and $\mathcal{D}(0, \mbf{c}) = 0$ for any $\mbf{c}$.  In practice, $\E [D]$ can be estimated from e.g., historical carbon traces and the characteristic behavior of $\texttt{PB}$ (i.e., in terms of relative importance).
We have the following bound:
\begin{align*}
\E [ \DANISH_K(\mathcal{J}) ] &\le \E [ \texttt{PB}_K(\mathcal{J}) ] +  \mathcal{D}(\gamma, \mbf{c}) \OPT_1(\mathcal{J}).
\end{align*}
This gives insight into the tuning of hyperparameter $\gamma$ -- for a low-carbon period $\mbf{c}'$ where a practitioner desires full throughput, one should tune $\gamma$ such that $\mathcal{D}(\gamma, \mbf{c'}) \approx 0$.

Since \texttt{PB} follows the list scheduling paradigm of scheduling tasks in an order that respects precedence constraints, it inherits the following worst-case theoretical bound on  its makespan:
\begin{align*}
\E [ \texttt{PB}_K(\mathcal{J}) ] \le \left(2 - \frac{1}{K} \right) \OPT_K(\mathcal{J}).
\end{align*}

From the proof of \autoref{thm:ksMakespan}, we also have the following bound for $\OPT_1(\mathcal{J})$:
\begin{align*}
\OPT_1(\mathcal{J}) \le K\cdot \OPT_K(\mathcal{J}).
\end{align*}

Combining these results, we have the following:
\begin{align*}
\E [ \DANISH_K(\mathcal{J}) ] &\le \left(2 - \frac{1}{K} + \mathcal{D}(\gamma, \mbf{c}) K \right) \OPT_K(\mathcal{J}).
\end{align*}

Combined with the bound for $\texttt{PB}$, this shows that $\DANISH$ has a carbon stretch factor of $1 + \frac{\mathcal{D}(\gamma, \mbf{c}) K}{2 - \frac{1}{K}}$.

\end{proof}

We now turn to the question of carbon savings, and the result stated in \autoref{thm:danishCarbonSavings}. 
First, for ease of analysis, we define a \textit{discretized time model} that is motivated by the empirical fact that carbon intensities are reported in discrete time intervals~\cite{electricity-map, watttime}.
Assuming that new carbon intensity values arrive at integers in continuous time, we define a discretized carbon signal for any discrete time step $t$ as $c_{t} \coloneqq c(t') : t' \in [t, t+1)$, where note that $c(t')$ does not change on the interval $t \in [t, t+1)$.

Slightly abusing notation, we let $C_{\DANISH}(t)$ denote the carbon emissions of $\DANISH$'s schedule during discrete time step $t$, and let $C_{\texttt{PB}}(t)$ denote the carbon emissions of $\texttt{PB}$ at discrete time step $t$, respectively.  
The schedules generated by $\DANISH$ and $\texttt{PB}$ each use a certain number of machines during any discrete time interval -- to capture this, we let $E^\DANISH_t : t' \in [t, t+1)$ denote the number of active machines during discrete time step $t$ in \DANISH's schedule, and $E^\texttt{PB}_t \leq K$ denotes the same for \texttt{PB}'s schedule.  In what follows, we use $W$ as shorthand to denote the \textit{excess work} that \DANISH must ``make up'' with respect to \texttt{PB}'s schedule (i.e., due to deferral actions). Formally, $W = \sum_{i=0}^T \max\{ E^\texttt{PB}_t - E^\DANISH_t , \ 0 \}$.

\subsubsection{\textbf{Proof of \autoref{thm:danishCarbonSavings}}}\label{apx:danishCarbonSavings}

In what follows, we prove \autoref{thm:danishCarbonSavings}, which states that for a given carbon intensity trace $\mbf{c}$, \DANISH yields carbon savings of $W \left( \overline{s}_{-}^{(0,T)} - \overline{s}_{+}^{(0,T)} - \overline{c}^{(T, T')} \right)$ compared to a carbon-agnostic baseline \texttt{PB}. %

\begin{proof}
We let $C_s(t)$ denote the \textit{carbon savings} of \DANISH at discrete time step $t$.  Formally, we have:
\[
C_s(t) = \begin{cases}
    C_{\texttt{PB}}(t) - C_{\DANISH}(t) &  1 \leq t \leq T,\\
    - C_{\DANISH}(t) & T < t \leq T'.
\end{cases}
\]
By definition, we have the following by substituting for the carbon emissions of $\texttt{PB}$ and \DANISH:
\[
C_s(t) = \begin{cases}
    E^\texttt{PB}_t c_t  - E^\DANISH_t c_t &  1 \leq t \leq T,\\
    - E^\DANISH_t c_t& T < t \leq T'.
\end{cases}
\]
Summing over all time steps, we have that the carbon savings is given by:
\[
\sum_{i=0}^{T'} C_s(i) = \sum_{i=0}^{T} (E^\texttt{PB}_i-E^\DANISH_i) c_i - \sum_{i=T+1}^{T'} E^\DANISH_i c_i
\]

Note that because of \DANISH's job-specific decisions, it is not necessarily the case that $E^\texttt{PB}_t \geq E^\DANISH_t$ for any $t \leq T$ -- namely, if \texttt{PB}'s schedule is constrained by a bottleneck task during a low-carbon time step, \DANISH's schedule may use that low-carbon time step to schedule deferred tasks from previous time steps.

Thus, to begin simplifying this expression, we consider two cases for the sum from $0$ to $T$ as follows:
\begin{align*}
\sum_{i=0}^{T'} C_s(i) &= \sum_{i=0}^{T} (E^\texttt{PB}_i-E^\DANISH_i) c_i \mathbb{I}(E^\texttt{PB}_i \geq E^\DANISH_i),\\
& - \sum_{i=0}^{T} (E^\DANISH_i - E^\texttt{PB}_i) c_i \mathbb{I}(E^\texttt{PB}_i < E^\DANISH_i) - \sum_{i=T+1}^{T'} E^\DANISH_i c_i.
\end{align*}

We define three terms that capture the \textit{weighted average} carbon intensity per unit of work that is deferred, opportunistically completed, or completed later as follows.  
Note that $\sum_{i=0}^{T} (E^\texttt{PB}_i-E^\DANISH_i) = \sum_{i=T+1}^{T'} E^\DANISH_i$.

We let $\overline{s}_{-}^{(0,T)}$ denote the weighted average carbon intensity of the machine time (work) that is \textit{deferred} in \DANISH's schedule (i.e., carbon saved due to \DANISH's carbon-aware filtering):
\[
\overline{s}^{(0,T)} = \frac{\sum_{i=0}^{T} (E^\texttt{PB}_i-E^\DANISH_i) c_i}{W} \mathbb{I}(E^\texttt{PB}_i \geq E^\DANISH_i)
\]

Furthermore, we let $\overline{s}_{+}^{(0,T)}$ denote the weighted average carbon intensity of the machine time (work) that is \textit{opportunistically completed} in \DANISH's schedule (i.e., when \DANISH does more work than \texttt{PB}, likely during a low-carbon period):
\[
\overline{s}_{+}^{(0,T)} = \frac{\sum_{i=0}^{T} (E^\texttt{PB}_i-E^\DANISH_i) c_i}{W} \mathbb{I}(E^\texttt{PB}_i < E^\DANISH_i)
\]

Finally, we let $\overline{c}^{(T, T')}$ denote the weighted average carbon intensity of the work completed by \DANISH after time $T$:
\[
\overline{c}^{(T,T')} = \frac{\sum_{i=T+1}^{T'} E^\DANISH_i c_i}{W}
\]
Under the above definitions, we can pose the total carbon savings as:
\[
\sum_{i=0}^{T'} C_s(i) = W \left( \overline{s}_{-}^{(0,T)} -\overline{s}_{+}^{(0,T)} - \overline{c}^{(T,T')} \right)
\]

\end{proof}

\noindent We note that an adversary could construct instances such that \PCAPS  uses more carbon than a carbon-agnostic baseline -- for instance, consider the case where the carbon intensity at each time step is strictly increasing over time.  In such a scenario, the ``carbon-optimal'' solution simply finishes the job as soon as it can, and such a scenario implies that $\overline{c}^{(T,T')} + \overline{s}_{+}^{(0,T)} > \overline{s}_{-}^{(0,T)}$, meaning that \PCAPS's carbon savings are negative.  However, we note such scenarios for the carbon intensity on the grid are unrealistic.  In reality, grid carbon intensity exhibits \textit{diurnal (i.e., daily) patterns} that mediate this effect -- see \autoref{fig:CI-traces} for an example.

The above result contextualizes the total carbon savings achieved by \DANISH for a single job, but we also consider the average carbon savings at each (discrete) carbon intensity interval as follows.

Let $\rho_{\DANISH}(c)$ denote the average machine utilization for \DANISH's schedule conditioned on the fact that the current carbon intensity is $c = c_t$.  Denoting the set of discrete time steps with carbon intensity $c$ by $\mathcal{T}_c$, we have the following: $\rho_{\DANISH}(c) = \lim_{T\to \infty} \frac{\sum_{i\in\mathcal{T}_c} \nicefrac{E^\DANISH_i}{K}}{\vert \mathcal{T}_c \vert}$.  
Note that $\rho_{\DANISH}(c)$ can be estimated based on e.g., the observed relative importances of tasks produced by \texttt{PB} -- this \textit{distribution} of relative importances maps (via $\Psi_\gamma$) into a distribution of carbon intensity values -- the fraction of these values that lie below $c$ is proportional to $\rho_{\DANISH}(c)$, since the fraction of values above $c$ correspond to tasks that would be deferred by $\DANISH$.

\begin{cor} \label{cor:danishCarbonSavingsContinuous}
In a setting where there are always jobs with outstanding tasks in the data processing queue, the average carbon savings of \DANISH at any given discrete time step $t$ is given by $\left(\rho_{\texttt{PB}} K - \rho_{\DANISH}(c_t) K \right) c_t$.
\end{cor}
\begin{proof}
In this setting, the expression of the \textit{average} carbon savings at any given discrete time step simplifies as follows:

Let $\rho_{\texttt{PB}}$ denote the average machine utilization of \texttt{PB}'s schedule, i.e., $\rho_{\texttt{PB}} = \lim_{T\to \infty} \frac{\sum_{i=0}^T \nicefrac{E^\texttt{PB}_i}{K}}{T}$.

Then by \autoref{thm:danishCarbonSavings}, the average carbon savings $\overline{C}_s$ at  any time step $t$ is given by the following:
\begin{align*}
\overline{C}_s(t) &= \left(\rho_{\texttt{PB}} K - \rho_{\DANISH}(c_t) K \right) c_t.
\end{align*}
\end{proof}

\subsection{Deferred Proofs and Discussion from \autoref{sec:cap-design} (Carbon-aware provisioning (\CAP))} \label{apx:cap-proofs}

In this section, we discuss and prove the analytical results for \CAP, introduced in \autoref{sec:cap-design}.  Throughout this section, we let $\texttt{AG}$ denote a \textit{carbon-agnostic} baseline scheduler.

For the sake of analysis, in the following results we leverage the classic makespan bound of Graham's list scheduling algorithm~\cite{Graham:66}, which is known to produce a schedule that is a $(2 - \nicefrac{1}{K})$-approximation for the optimal makespan on $K$ identical machines.  
Note that FIFO is an instance of list scheduling where the list (of tasks) is a FIFO queue, and the next task is assigned to a machine as soon as it becomes idle.

Suppose that for a job $\mathcal{J}$, \CAP's schedule completes it at time $T' \geq 0$.
Note that if $c(t) = L$ for all time steps, the schedule of \CAP is identical to that of \texttt{AG} because \CAP always sets $r(t) = K$.  Otherwise, we let $\OPT_{K}(\mathcal{J})$ denote the optimal makespan on $K$ machines, and let $\mathcal{M}(B, \mbf{c})$ denote the minimum resource cap specified by \CAP at any point in its schedule (note this depends on the carbon signal $\mbf{c}$).  Formally, we let: $\mathcal{M}(B, \mbf{c}) = \arg \max_{i\in[K]} \Phi_i : \Phi_i \le c(t) \ \forall t \in [0, T']$.

\subsubsection{\textbf{Proof of \autoref{thm:ksMakespan}}}\label{apx:ksMakespan}

We are now ready to prove \autoref{thm:ksMakespan}, which states that \CAP's
carbon stretch factor is $\left(\frac{K}{\mathcal{M}(B, \mbf{c})}\right)^2 \frac{2\mathcal{M}(B, \mbf{c})-1}{2K-1}$.
We prove the statement by first showing that \CAP's makespan is at most $\left( \frac{2K}{\mathcal{M}(B, \mbf{c})} - \frac{K}{\mathcal{M}(B, \mbf{c})^{2}} \right) \OPT_{K}(\mathcal{J})$.

\smallskip

\begin{proof}
Let $\CAP_K(\mathcal{J} \mid \mathcal{M}(B, \mbf{c}) )$ denote the makespan of \CAP given $K$ machines conditioned on the value of $\mathcal{M}(B, \mbf{c})$, and let $\texttt{AG}_K(\mathcal{J})$ denote the makespan of \texttt{AG} (i.e., Graham's list scheduling with $K$ machines).

First, note that the following holds by~\cite{Graham:66}:
\[
\texttt{AG}_K(\mathcal{J}) \le \left(2 - \frac{1}{K}\right) \OPT_K(\mathcal{J}).
\]

Second, note that the makespan of $\CAP_K(\mathcal{J} \mid \mathcal{M}(B, \mbf{c}) )$ is upper-bounded by that of $\texttt{AG}_{\mathcal{M}(B, \mbf{c})}(\mathcal{J}) $.  By definition of $\mathcal{M}(B, \mbf{c})$, $\CAP$ always has $\mathcal{M}(B, \mbf{c})$ machines available, which is the same as $\texttt{AG}_{\mathcal{M}(B, \mbf{c})}(\mathcal{J})$.  If any other machines become available and process jobs during the schedule of $\CAP_K(\mathcal{J} \mid \mathcal{M}(B, \mbf{c}) )$, these \textit{strictly help} the makespan with respect to $\texttt{AG}_{\mathcal{M}(B, \mbf{c})}(\mathcal{J})$.  Thus, we have:
\[
\CAP_K(\mathcal{J} \mid \mathcal{M}(B, \mbf{c}) ) \le \texttt{AG}_{\mathcal{M}(B, \mbf{c})}(\mathcal{J})
\]

Furthermore, we have the following relationship between the optimal makespans (for the same job) when given different amounts of machines.  Letting $\mathcal{M}(B, \mbf{c}) \le K$, we have that:
\[
\OPT_{\mathcal{M}(B, \mbf{c})}(\mathcal{J}) \le \frac{K}{\mathcal{M}(B, \mbf{c})} \OPT_K(\mathcal{J}).
\]

To observe this, consider the limiting case as $\mathcal{M}(B, \mbf{c}) \to 1$.     When $\mathcal{M}(B, \mbf{c}) = 1$, the optimal makespan contains no ``gaps'', in the sense that the single machine is always being utilized.  If the job is perfectly parallelizable and subdividable, we have that $\OPT_{\mathcal{M}(B, \mbf{c})}(\mathcal{J}) = \frac{K}{\mathcal{M}(B, \mbf{c})} \OPT_K(\mathcal{J})$ by a geometric proof (i.e., $\OPT_K(\mathcal{J})$ has a makespan that is $\nicefrac{1}{K}$ as long as $\OPT_1(\mathcal{J})$).  For any other job, as the number of machines increases, the utilization of machines worsens.  

\begin{figure}[h]
    \centering
    \includegraphics[width=0.5\linewidth]{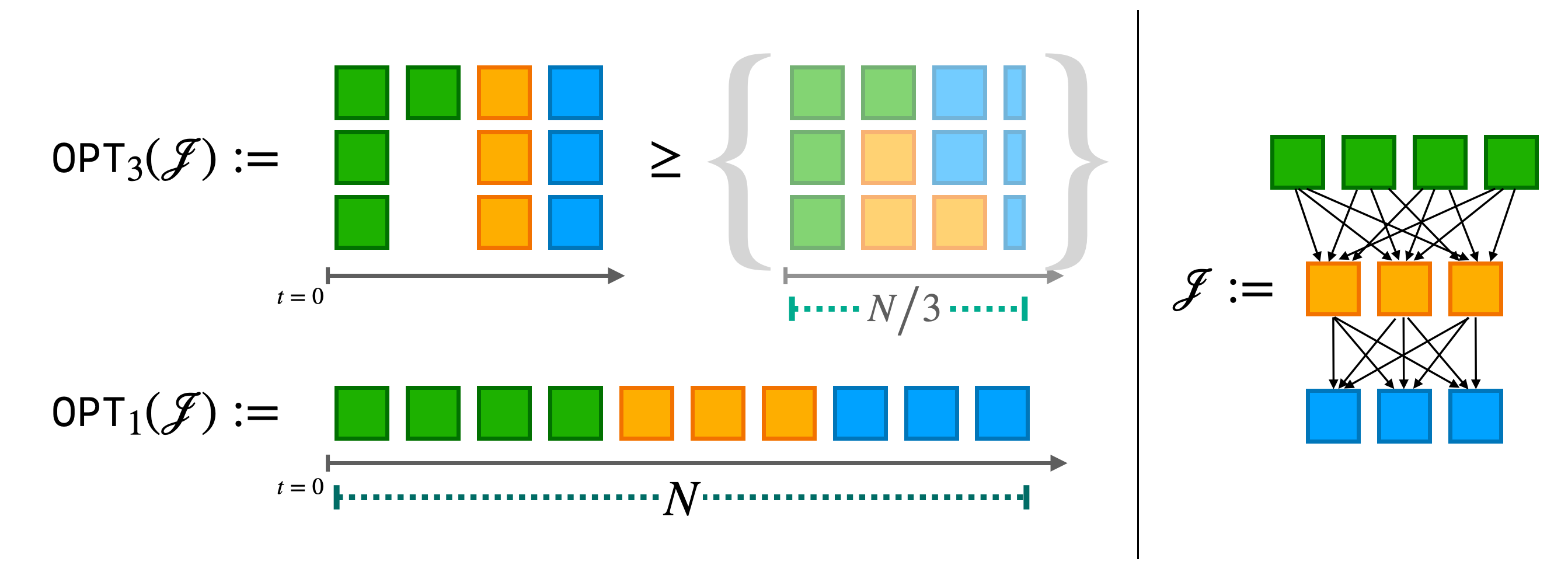} \vspace{-1em}
    \caption{ An example to contextualize how the optimal makespan differs when a job is given different amounts of machines.  In the case of a single machine (i.e., $\OPT_1(\mathcal{J})$), the precedence constraints defined by the DAG on the right-hand side of the figure are non-binding -- there is always one or more tasks that are ready to execute.  As the number of machines increases, situations arise where some machines must be idle (e.g., in $\OPT_3(\mathcal{J})$), indicated by ``blank slots'' in the optimal schedule.  The relation between makespans is then formalized by considering a hypothetical schedule (in brackets) that ignores precedence constraints and subdivides individual tasks across machines -- while this is an infeasible schedule, it forms a \textit{lower bound} on the makespan of any feasible one. }
    \label{fig:cap-makespan-explainer}
\end{figure}

We give a visual example of such a job $\mathcal{J}$ with $N=10$ tasks in \autoref{fig:cap-makespan-explainer} -- observe that respecting the precedence constraints in the case with $K=3$ machines necessarily forces a makespan that is greater than the hypothetical best makespan if jobs are perfectly parallelizable and subdividable. 
Thus, we have that scaling $\OPT_K(\mathcal{J})$ by $\frac{K}{\mathcal{M}(B, \mbf{c})}$ is always an upper bound on the optimal makespan $\OPT_{\mathcal{M}(B, \mbf{c})}(\mathcal{J})$.

Combining the above bounds, we obtain the following:
\begin{align*}
\CAP_K(\mathcal{J} \mid \mathcal{M}(B, \mbf{c}) ) &\le \texttt{AG}_{\mathcal{M}(B, \mbf{c})}(\mathcal{J}),\\
&\le \left(2 - \frac{1}{\mathcal{M}(B, \mbf{c})}\right) \OPT_{\mathcal{M}(B, \mbf{c})}(\mathcal{J}),\\
&\le \left( \frac{2K}{\mathcal{M}(B, \mbf{c})} - \frac{K}{\mathcal{M}(B, \mbf{c})^{2}} \right) \OPT_K(\mathcal{J}).
\end{align*}

This gives that the carbon stretch factor (\sref{Definition}{dfn:csf}) of \CAP is given by $\left(\frac{K}{\mathcal{M}(B, \mbf{c})}\right)^2 \frac{2\mathcal{M}(B, \mbf{c})-1}{2K-1}$. 
\end{proof}

\noindent We now turn to the question of carbon savings, and the result stated in \autoref{thm:ksCarbonSavings}.  
For ease of analysis, we again consider a \textit{discretized time model} as defined in \sref{Appendix}{apx:danishCarbonSavings}.

Slightly abusing notation, we let $C_{\CAP}(t)$ denote the carbon emissions of $\CAP$'s schedule during discrete time step $t$, and let $C_{\texttt{AG}}(t)$ denote the carbon emissions of $\texttt{AG}$ at discrete time step $t$, respectively.  
Schedules generated by $\CAP$ and $\texttt{AG}$ each use a certain number of machines during any discrete time interval -- to capture this, we let $E^\CAP_t \leq r(t') : t' \in [t, t+1)$ denote the number of active machines during discrete time step $t$ in \CAP's schedule.  Note that  $r(t')$ is constant on the interval $t \in [t, t+1)$, and that $E^\CAP_t$ \textit{need not be} an integer -- i.e., if a machine is active for only half of the discrete time interval, that machine contributes fractionally to $E^\CAP_t$.  We let $E^\texttt{AG}_t \leq K$ denote the same quantity for \texttt{AG}'s schedule.

On a per-job basis, let $T$ denote the makespan of $\texttt{AG}$ (i.e., $T = \texttt{AG}_K(\mathcal{J})$), where note that $T \le T'$.
In what follows, we use $W$ as shorthand to denote the \textit{excess work} that \CAP must complete after time $T$ (i.e., after $\texttt{AG}$ has completed).  Formally, $W = \sum_{i=0}^T E^\texttt{AG}_t - E^\CAP_t$.  We also define quantities $\overline{s}^{(0,T)}$ and $\overline{c}^{(T, T')}$, which are weighted averages based on a combination of the carbon intensity and schedules of \texttt{AG} and \CAP, respectively. 
These can be interpreted as the realization of carbon intensities that \CAP ``waited for'' -- in other words, it deferred some work in between time $0$ and time $T$ (saving $\overline{s}^{(0,T)} $ amount of carbon), so it must make up the difference after time $T$.  

\subsubsection{\textbf{Proof of \autoref{thm:ksCarbonSavings}}}\label{apx:ksCarbonSavings}

We are ready prove \autoref{thm:ksMakespan}, which states that \CAP yields carbon savings compared to a carbon-agnostic baseline of $W \left( \overline{s}^{(0,T)} - \overline{c}^{(T, T')} \right)$.

\begin{proof}
Slightly abusing notation, we let $C_s(t)$ denote the \textit{carbon savings} of \CAP at discrete time step $t$.  Formally, we have:
\[
C_s(t) = \begin{cases}
    C_{\texttt{AG}}(t) - C_{\CAP}(t) &  1 \leq t \leq T,\\
    - C_{\CAP}(t) & T < t \leq T'.
\end{cases}
\]
By definition, we have the following by substituting for the carbon emissions of $\texttt{ECA}$ and \CAP:
\[
C_s(t) = \begin{cases}
    E^\texttt{AG}_t c_t  - E^\CAP_t c_t &  1 \leq t \leq T,\\
    - E^\CAP_t c_t& T < t \leq T'.
\end{cases}
\]
Summing over all time steps, we have that the carbon savings is given by:
\[
\sum_{i=0}^{T'} C_s(i) = \sum_{i=0}^{T} (E^\texttt{AG}_i-E^\CAP_i) c_i - \sum_{i=T+1}^{T'} E^\CAP_i c_i
\]
To simplify this expression, we define two terms that capture the \textit{weighted average} carbon intensity per unit of work completed/deferred.  First, note that $\sum_{i=0}^{T} (E^\texttt{AG}_i-E^\CAP_i) = \sum_{i=T+1}^{T'} E^\CAP_i = W$.

Formally, we let $\overline{s}^{(0,T)}$ denote the weighted average carbon intensity of the deferred work $W$ that is completed by \texttt{AG} before time $T$ but must be completed after time $T$ by \CAP:
\[
\overline{s}^{(0,T)} = \nicefrac{\sum_{i=0}^{T} (E^\texttt{AG}_i-E^\CAP_i) c_i}{W}
\]

Similarly, we let $\overline{c}^{(T, T')}$ denote the weighted average carbon intensity of the deferred work $W$ that is completed by \CAP after time $T$:
\[
\overline{c}^{(T,T')} = \nicefrac{\sum_{i=T+1}^{T'} E^\CAP_i c_i}{W}
\]
Under the above definitions, we can pose the total carbon savings as:
\[
\sum_{i=0}^{T'} C_s(i) = W \left( \overline{s}^{(0,T)} - \overline{c}^{(T,T')} \right)
\]
\end{proof}

The above result contextualizes the total carbon savings achieved by \CAP for a single job, but we may also consider the average carbon savings at each (discrete) carbon intensity interval as follows.
We let $\rho_{\texttt{AG}} \in [0,1)$ denote the average cluster utilization of $\texttt{AG}$, and we let $\rho_{\texttt{CAP}} \in [0,1)$ denote the average cluster  utilization of \CAP.  In general, since \CAP allows jobs to use less than or equal the amount of resource that \texttt{AG} allows, we expect $\rho_{\texttt{AG}} \leq \rho_{\texttt{CAP}}$ as for the same number of jobs submitted (e.g., during a given period), jobs will take up a greater proportion of the resources that \CAP allows.

\begin{cor} \label{cor:ksCarbonSavingsContinuous}
In a setting where there are always jobs with outstanding tasks in the data processing queue, the average carbon savings of \CAP at any given discrete time step $t$ is given by $(\rho_{\texttt{AG}} K - \rho_{\CAP} r_t ) \Phi_{r_t + B}$.
\end{cor}
\begin{proof}
In a setting where there are always jobs with outstanding tasks in the data processing queue, the expression of the \textit{average} carbon savings at any given discrete time step simplifies as follows:

Let $\rho_{\texttt{AG}}$ denote the average machine utilization of \texttt{AG}'s schedule, i.e., $\rho_{\texttt{AG}} = \lim_{T\to \infty} \frac{\sum_{i=0}^T \nicefrac{E^\texttt{AG}_i}{K}}{T}$, and let $\rho_{\CAP}$ denote the same for \CAP's schedule, i.e., $\rho_{\CAP} = \lim_{T\to \infty} \frac{\sum_{i=0}^T \nicefrac{E^\CAP_i}{j_i}}{T}$

Then the average carbon savings $\overline{C}_s$ at  any time step $t$ is given by the following, where $r_t \coloneqq r(t') : t' \in [t, t+1)$:
\begin{align*}
\overline{C}_s(t) &= \left(\rho_{\texttt{AG}} K - \rho_{\CAP} r_t \right) c_t, \\
& \ge \left(\rho_{\texttt{AG}} K - \rho_{\CAP} r_t \right) \Phi_{r_t + B}.
\end{align*}
\end{proof}

\end{document}